\pdfoutput=1 

\documentclass[12pt,letterpaper]{article}
\usepackage[latin1]{inputenc}

\usepackage{xcolor}

\usepackage[protrusion=false]{microtype} 

\usepackage{titletoc}
\newcommand\DoToC{%
  \startcontents
  \printcontents{}{1}{\vskip 1.5em\hrule\vskip .75em}
  \vskip .75em\hrule\vskip 2em
}

\usepackage[small]{titlesec}

\titlespacing*{\section} {0pt}{3.5ex plus 1ex minus .2ex}{*1} 

\usepackage{setspace}
\usepackage{tikz}
\usetikzlibrary{positioning, shapes, decorations.pathreplacing,arrows.meta}

\definecolor{Bleu}{RGB}{0,0,204}
\definecolor{DZO}{rgb}{1,.5,0} 
\definecolor{DZG}{rgb}{0.151,0.620,0.151} 
\definecolor{ZB}{rgb}{.255,.42,.882} 

\definecolor{myDarkGrey}{rgb}{.1,.1,.1}
\definecolor{myLessDarkGrey}{rgb}{.125,.125,.125}
\definecolor{slateGrey}{HTML}{708090}

\definecolor{CBmagenta}{HTML}{EE3377}
\definecolor{CBred}{HTML}{CC3311}
\definecolor{CBorange}{HTML}{EE7733}
\definecolor{CBteal}{HTML}{009988}
\definecolor{CBcyan}{HTML}{33BBEE}
\definecolor{CBblue}{HTML}{0077BB}
\definecolor{CBgrey}{HTML}{BBBBBB}

\definecolor{CB2blue}{HTML}{6699CC}
\definecolor{CB2yellow}{HTML}{EECC66}
\definecolor{CB2red}{HTML}{EE99AA}

\definecolor{CB3yellow}{HTML}{CCBB44}
\definecolor{CB3green}{HTML}{228833}
\definecolor{CB3purple}{HTML}{EE6677}

\usepackage{listings}

\definecolor{codegreen}{rgb}{0,0.6,0}
\definecolor{codegray}{rgb}{0.5,0.5,0.5}
\definecolor{codepurple}{rgb}{0.58,0,0.82}
\definecolor{backcolour}{rgb}{0.97,0.97,0.95}

\usepackage{upquote} 

\lstdefinestyle{mystyle}{
    language=Python,
    escapeinside={(*}{*)},
    escapebegin=\color{codegreen},    backgroundcolor=\color{backcolour},   
    commentstyle=\color{codegreen},
    keywordstyle=\color{magenta},
    numberstyle=\tiny\color{codegray},
    stringstyle=\color{codepurple},
    basicstyle=\ttfamily\footnotesize,
    breakatwhitespace=false,         
    breaklines=true,                 
    captionpos=b,                    
    keepspaces=true,                 
    numbers=left,                    
    showspaces=false,                
    showstringspaces=false,
    showtabs=false,                  
    tabsize=2,
    moredelim=**[is][\color{slateGrey}]{@}{@},
    morestring=[b]' 
}

\usepackage{amsmath}
\usepackage{amsfonts}
\usepackage{amssymb}
\usepackage{amsthm}

\DeclareFontFamily{U}{mathb}{}
\DeclareFontShape{U}{mathb}{m}{n}{
  <-5.5> mathb5
  <5.5-6.5> mathb6
  <6.5-7.5> mathb7
  <7.5-8.5> mathb8
  <8.5-9.5> mathb9
  <9.5-11.5> mathb10
  <11.5-> mathb12
}{}
\DeclareSymbolFont{mathb}{U}{mathb}{m}{n}
\DeclareMathSymbol{\olddrsh}{3}{mathb}{"EB} 
\newcommand{\drsh}{\raisebox{0.3ex}{$\olddrsh$}} 

\makeatletter
\toks@\expandafter{\@endtheorem\@endpetrue}
\edef\@endtheorem{\the\toks@}
\makeatother
\usepackage{bm}
\usepackage{graphicx}
\usepackage{enumitem}
\usepackage{accents} 
\graphicspath{ {./figures/} }

\usepackage[left=1in, right=1in, top=1in, bottom=1in]{geometry}
\usepackage[title]{appendix}
\usepackage{authblk}
\usepackage{multirow}

\usepackage[hypertexnames=false]{hyperref}
\hypersetup{
colorlinks,
  citecolor=Bleu,
  linkcolor=Bleu,
  urlcolor=Bleu}

\usepackage{algorithm}
\usepackage[noend]{algpseudocode}
\newcommand{\setalglineno}[1]{%
  \setcounter{ALG@line}{\numexpr#1-1}}

\newcommand{\algrule}[1][.2pt]{\par\vskip.5\baselineskip\hrule height #1\par\vskip.5\baselineskip} 
  
\makeatletter
\newenvironment{algorithmic_special}{%
\begin{algorithmic}[1]%
\def\ALG@step{%
   \addtocounter{ALG@line}{1}%
   \ifthenelse{\equal{\arabic{ALG@line}}{3}}%
      {\alglinenumber{k}}%
      { \ifthenelse{\equal{\arabic{ALG@line}}{4}}%
           {}%
           {\alglinenumber{\arabic{ALG@line}}}%
      }%
   }%
}{%
\end{algorithmic}%
}%
\makeatother

\usepackage{multirow}

\usepackage[labelfont=bf,font=small,singlelinecheck=false]{caption}
\usepackage[singlelinecheck=true]{subcaption}

\usepackage[mathscr]{euscript}
\usepackage{bbm}

\usepackage{thmtools} 
\renewcommand\thmcontinues[1]{continued}

\newtheorem{theorem}{Theorem}
\newtheorem{lemma}{Lemma}

\theoremstyle{definition}

\newtheorem{conditionC}{Condition}

\usepackage{chngcntr}
\newcounter{parentnumber}


\usepackage{cleveref}
\crefname{example}{Example}{Examples} 

\DeclareMathOperator*{\argmin}{argmin}

\DeclareMathOperator*{\esssup}{ess\,sup}
\DeclareMathOperator*{\essinf}{ess\,inf}

\newcommand{\mymid}{\,|\,}

\usepackage[width=\textwidth]{caption}

\usepackage[round]{natbib}
\usepackage{relsize}
\title{\large\bfseries Simplifying debiased inference\\
via automatic differentiation and probabilistic programming\vspace{.5em}}

\author{\normalsize Alex Luedtke\vspace{-.75em}}
\affil{\normalsize Department of Statistics, University of Washington}

\date{}

\providecommand{\keywords}[1]
{
  \small	
  \textit{Keywords:} #1
}

\begin{document}
\allowdisplaybreaks
{\onehalfspacing
\maketitle
}

\begin{abstract}
\noindent 
We introduce an algorithm that simplifies the construction of efficient estimators, making them accessible to a broader audience. `Dimple' takes as input computer code representing a parameter of interest and outputs an efficient estimator. Unlike standard approaches, it does not require users to derive a functional derivative known as the efficient influence function. Dimple avoids this task by applying automatic differentiation to the statistical functional of interest. Doing so requires expressing this functional as a composition of primitives satisfying a novel differentiability condition. Dimple also uses this composition to determine the nuisances it must estimate. In software, primitives can be implemented independently of one another and reused across different estimation problems. We provide a proof-of-concept Python implementation and showcase through examples how it allows users to go from parameter specification to efficient estimation with just a few lines of code.
\end{abstract}\vspace{.25em}
\keywords{efficient influence function, asymptotic efficiency, pathwise differentiability, Hadamard differentiability, automatic differentiation, probabilistic programming, differentiable programming}
\vspace{.75em}

\section{Introduction}

\subsection{Motivation}

There is a long history of constructing estimators that are asymptotically efficient (`efficient' hereafter) in infinite-dimensional models \citep{levit1975efficiency,hasminskii1979nonparametric,pfanzagl1990estimation,bickel1993efficient}. Several approaches for doing this permit the use of arbitrary methods for estimating needed nuisance functions. These include one-step estimation \citep{pfanzagl1982contributions}, estimating equations \citep{van2003unified}, targeted learning \citep{van2006targeted}, and double machine learning  \citep{chernozhukov2017double}. These frameworks construct efficient estimators using an object known as the efficient influence function (EIF). Though EIFs have been derived in many examples \citep{newey1990semiparametric,bickel1993efficient,van2003unified}, computing them in new problems can be challenging, requiring tools from functional analysis and differential geometry. These calculations can represent a bottleneck in scientific discovery, as some novel analyses cannot be performed until an expert derives the EIF.

Investigators have sought to use numerical approaches to overcome this bottleneck. \cite{frangakis2015deductive} built on the concept of an influence curve from robust statistics \citep{hampel1968contributions} to propose a method using finite differences to approximate nonparametric EIFs. This method applies whenever the data are discrete or the parameter is a smooth functional of the distribution function. \cite{luedtke2015discussion} provided a modified finite difference approach that correctly approximates the EIF more generally, even when the data are continuous. \Citet{van2015computerizing} and \cite{carone2018toward} improved upon this proposal by providing numerical strategies for approximating EIFs in the harder case where the model is semiparametric, rather than nonparametric. 
In a concurrent line of work, \cite{ichimura2022influence} developed a numerical approach to evaluate the influence function of a given estimator, which coincides with the EIF when the estimator is efficient. Though one pursues numerical approximation of the influence function of an estimator and the other pursues numerical approximation of the EIF of a statistical functional, the proposals in \cite{luedtke2015discussion} and \cite{ichimura2022influence} address challenges related to the possible non-discreteness of the data in similar ways.

The aforementioned numerical approaches rely on two hyperparameters: a step size $\epsilon$ for a finite difference approximation and a bandwidth $\lambda$ for a kernel approximation of a delta function. While both must be small, $\epsilon$ should be an order of magnitude smaller than $\lambda$. However, setting $\epsilon$ too small can lead to numerical instability. \cite{carone2018toward} provide guidance on the choice of hyperparameters through a diagnostic they call $\epsilon\,$--$\,\lambda$ plots. \cite{jordan2022empirical} studied the conditions that a numerical approximation must meet to guarantee that the resulting one-step estimator possesses desirable statistical properties. The complexity of these requirements underscores the need for careful hyperparameter selection in existing numerical approaches, which could limit their adoption. Researchers in other fields have encountered similar challenges when using numerical differentiation \citep{griewank2008evaluating}.

Automatic differentiation avoids these challenges \citep{wengert1964simple,speelpenning1980compiling}. In this approach, a function $f$ is written as a composition of differentiable functions with known gradients, which we refer to as primitives. Then, an algorithm can iteratively apply the chain rule to derive the gradient of $f$ at a specified point. Unlike numerical differentiation, automatic differentiation is hyperparameter-free and is always accurate up to working precision. Differentiable programming extends automatic differentiation by identifying a compositional structure for $f$ directly via computer code that can be used to evaluate it, simplifying implementation \citep{baydin2018automatic,blondel2024elements}. Automatic differentiation and differentiable programming have been used with success in applications \citep{bischof1995automatic,homescu2011adjoints}, and have played a crucial role in the rapid progress made in deep learning \citep{goodfellow2016deep}. However, they have yet to be used to compute a semi- or nonparametric EIF. This may be due to the fact that EIFs correspond to a gradient of a functional defined on an infinite-dimensional statistical manifold, and so it has been unclear what primitives can be chained together to compute them algorithmically.

\subsection{Contributions}\label{sec:contributions}

Our first two contributions are as follows:
\begin{enumerate}
    \item In Section~\ref{sec:algAndGuarantees}, we provide an automatic differentiation algorithm for computing the EIF of any parameter that can be expressed as a composition of differentiable Hilbert-valued primitives.
    \item In Section~\ref{sec:primitives}, we establish the differentiability of a variety of primitives. These primitives can be composed to express many parameters of interest.
\end{enumerate}
We then turn to estimation, building on developments in probabilistic programming \citep{van2018introduction}. Probabilistic programming and finite-dimensional automatic differentiation have been employed for Bayesian inference \citep{ge2018t,tran2018simple,stan2023}, parametric maximum likelihood estimation \citep{fournier2012ad,kristensen2015tmb}, M-estimation \citep{zivich2022,kosmidis2024empirical}, and efficient estimation in high-dimensional parametric models \citep{agrawal2024automated}. However, to date there is not a probabilistic programming framework for efficiently estimating a generic smooth parameter that may rely on nonparametric nuisances like density and regression functions. Our remaining contributions address this gap.
\begin{enumerate}[resume*]
    \item In Section~\ref{sec:probProg}, we provide an algorithm to estimate any parameter that is expressed as a composition of known primitives and conditions under which this estimator will be efficient.
    \item In Section~\ref{sec:illustrationPackage}, we introduce a proof-of-concept Python package implementing our approach, available at \url{http://github.com/alexluedtke12/pydimple}.
\end{enumerate}
We refer to our framework for constructing efficient estimators using automatic differentiation and probabilistic programming as `dimple', short for `debiased inference made simple'.

The goal of dimple is similar in spirit to that of automatic debiased machine learning methods: to avoid analytical derivations when constructing efficient estimators \citep{chernozhukov2022riesznet,chernozhukov2022nested,chernozhukov2022automatic}. Those existing methods apply to regression functionals, whose EIFs depend on the functional's form through nuisance functions known as Riesz representers. Automatic debiased machine learning provides loss functions that can be used to estimate these nuisances without explicitly deriving their form. Compared to these methods, dimple innovates in several directions. It allows the functional to depend on the distribution through non-regression summaries, uses automatic differentiation to compute the EIF, and is integrated into a user-friendly probabilistic programming framework. It also applies to general Hilbert-valued parameters \citep{luedtke2023one}.

Beyond simplifying estimation, our automatic differentiation approach and the chain rule that underlies it may be of independent interest for their use in deriving semiparametric efficiency bounds, which are given by the variance of the EIF when it exists \citep{van1989prohorov,van1991differentiable}. Our theoretical results show this EIF exists for any parameter expressible as a composition of differentiable primitives, and applying automatic differentiation by hand gives its analytic form. An example of this is given in Section~\ref{sec:illustrationR2}.

\subsection{Illustrations of proposed approach}\label{sec:illustrationPackage}

Our proof-of-concept Python package, \texttt{pydimple}, uses dimple to construct nonparametric efficient estimators. We illustrate it with three examples. The following Python code imports the classes and functions from the package needed to run them:\vspace{.25em}
\begin{lstlisting}[numbers=none]
 from pydimple import Distribution, E, Density, Var, RV, estimate
\end{lstlisting}
In our first example, the objective is to estimate the expected density $\psi(P):=E_P[p(Z)]$, with $p$ the probability density function of $Z\sim P$.\vspace{.25em}
\begin{lstlisting}[caption=Expected density \citep{bickel1988estimating}.,label=ex:expectedDensity]
 P = Distribution(data=dat) 
 dens = Density(P, 'Z') (* \hfill \# $p$\label{ln:expectedDensEst1} *)
 expected_density = E(P,dens) (* \hfill \# $\psi(P)$\label{ln:expectedDensEst2} *)
 estimate(expected_density) 
 # Output:  {'est': 1.699, 'se': 0.035, 'ci': [1.631, 1.768]}
\end{lstlisting}
The first line of code defines a distribution $P$ and identifies a dataset \texttt{dat} containing iid draws from it to be used for estimation. The next two lines serve as a blueprint of the statistical estimation problem, formulating it using primitives from \texttt{pydimple}. The fourth line computes a 5-fold cross-fitted estimator, Wald-type standard error, and 95\% confidence interval of $\psi(P)$. Internally, \texttt{pydimple} estimates density functions using kernel density estimation and all other nuisances using lightgbm \citep{ke2017lightgbm}, with hyperparameters tuned via Optuna \citep{akiba2019optuna}.

In our next example, the objective is to estimate the nonparametric $R^2$ criterion, defined as $\psi(P):=1-\int \{y-E_P(Y\mymid X=x)\}^2\,dP(x,y)/\mathrm{Var}_P(Y)$, where $P$ is a distribution of $(X,Y)$.\vspace{.25em}
\begin{lstlisting}[caption=Nonparametric $R^2$ criterion \citep{williamson2021nonparametric}.,label=ex:r2]
 P = Distribution(data=dat) 
 v = Var(P,'Y') (* \hfill \# $\mathrm{Var}_P(Y)$\label{ln:vdef} *)
 mu = E(P,'Y',indep_vars=['X1','X2']) (* \hfill \# $E_P[Y\mymid X=\cdot\,]$\label{ln:mdef} *)
 R2 = 1-E(P,(RV('Y')-mu)**2)/v (* \hfill \# $\psi(P)$\label{ln:r2def} *)
 estimate(R2) (*\label{ln:r2est} *)
 # Output:  {'est': 0.4334, 'se': 0.0206, 'ci': [0.3930, 0.4738]}
\end{lstlisting}

The final example concerns the estimation of the longitudinal G-formula \citep{robins1986new,bang2005doubly,luedtke2017sequential,rotnitzky2017multiply,chernozhukov2022nested}. This parameter is defined in longitudinal settings where covariate-treatment pairs $(X_t,A_t)$ are collected at times $t=0,1,\ldots,T-1$, and then an outcome $Y$ is collected at the end of the study. For $t\in\{0,1,\ldots,T-1\}$, let $H_t$ denote the history $(X_0,A_0,X_1,A_1,\ldots,X_{t-1},A_{t-1},X_t)$ available just before treatment $A_t$ is administered. Also let $H_T=(H_{T-1},A_T,Y)$ denote all of the covariate, treatment, and outcome data. The parameter $\psi(P)$ is defined recursively by letting $\mu_{P,T}(h_T)=y$, $\mu_{P,t}(h_t)=E_P[\mu_{P,t+1}(H_{t+1})\mymid A_t=1,H_t=h_t]$ for $t=T-1,T-2,\ldots,0$, and $\psi(P)=E_P[\mu_{P,0}(H_0)]$. Under causal conditions, $\psi(P)$ gives the counterfactual mean outcome when everyone receives treatment $1$ at all time points. The recursive definition of this parameter can be easily expressed in \texttt{pydimple} via a for loop, as illustrated below.\vspace{.25em}
\begin{lstlisting}[caption={Longitudinal G-Formula \citep{robins1986new}.},label=ex:longitudinalG]
 P = Distribution(data=dat)
 T = 3   (* \hfill \# $T$ *)
 mu = 'Y'  (* \hfill \# $\mu_{P,T}$ *)
 for t in reversed(range(T)): (* \hfill \# $t=T-1,T-2,\ldots,0$ *)
     H = [f'X{j}' for j in range(t+1)]+[f'A{j}' for j in range(t)]  (* \hfill \# $H_t$ *)
     mu = E(P,dep=mu,indep_vars=H, fixed_vars={f'A{t}==1'}) (* \hfill \# $\mu_{P,t}$ *)
 mu = E(P,dep=mu) (* \hfill \# $\psi(P)$ *)
 estimate(mu)
 # Output:  {'est': 0.5313, 'se': 0.0304, 'ci': [0.4717, 0.5908]}
\end{lstlisting}

In Section~\ref{sec:simulation}, we evaluate the performance of \texttt{pydimple} in a Monte Carlo study. The results support that dimple can produce asymptotically efficient estimators and valid confidence intervals. At sample sizes of $n\ge 1000$, the confidence intervals have approximately nominal coverage and nearly optimal widths, scaling with the standard deviation of the EIF. These results are not surprising given that dimple uses a variant of one-step estimation, which has been shown to perform well in a variety of settings \citep{pfanzagl1982contributions,bickel1993efficient}. Details on this estimation framework and its theoretical properties are given in Section~\ref{sec:estimator}.

\section{Automatic differentiation of statistical functionals}

\subsection{Algorithm and theoretical guarantee}\label{sec:algAndGuarantees}

We present an algorithm for automatically differentiating a parameter $\psi$ that maps from a set of distributions $\mathcal{M}$ into a Hilbert space $\mathcal{W}_{\psi}$. This and all Hilbert spaces in this work are real. The set $\mathcal{M}$, called the model, may be nonparametric or semiparametric.

Using our approach requires that $\psi$ can be expressed as in Algorithm~\ref{alg:parameter}, meaning that, for any $P\in\mathcal{M}$, calling that algorithm with input $P$ returns $\psi(P)$. The $j$-th line of that algorithm computes the evaluation $h_j$ of a primitive $\theta_j : \mathcal{M}\times\mathcal{U}_j\rightarrow \mathcal{V}_j$ applied to $P$ and $h_{\mathrm{pa}(j)}:=(h_i : i\in\mathrm{pa}(j))$. Here, $\mathcal{U}_j\supseteq \prod_{i\in\mathrm{pa}(j)} \mathcal{V}_i$ and $\mathrm{pa}(j)\subseteq [j-1]:=\{1,2,\ldots,j-1\}$, where $[0]:=\emptyset$. When $\mathrm{pa}(j)=\emptyset$, we let $\mathcal{U}_j:=\{0\}$ and $h_{\mathrm{pa}(j)}:=0$. This adheres to the following general convention: for any tuple $(a_i : i\in [j])$ and set $\mathcal{S}\subseteq [j]$, $a_\mathcal{S}:=(a_i : i\in\mathcal{S})$ when $\mathcal{S}\not=\emptyset$ and $a_\mathcal{S}:=0$ otherwise.

Though Algorithm~\ref{alg:parameter} does not rely on $\mathcal{V}_j$ having an inner product structure, our automatic differentiation scheme does. Hence, we suppose we can specify an ambient Hilbert space $\mathcal{W}_j\supseteq \mathcal{V}_j$ for each $j\in [k-1]$. We also let $\mathcal{W}_k:=\mathcal{W}_\psi$ and view $\mathcal{U}_j$ as a subset of the direct sum Hilbert space $\bigoplus_{i\in\mathrm{pa}(j)}\mathcal{W}_i$, which we take to be a trivial Hilbert space $\{0\}$ when $\mathrm{pa}(j)=0$. 

\begin{algorithm}[tb]
   \caption{To use automatic differentiation, there must exist an algorithm of the following form that takes as input a generic $P\in\mathcal{M}$ and returns $\psi(P)$}
   \label{alg:parameter}
   \linespread{1.05}\selectfont
\begin{algorithmic_special}
    \Statex $\triangleright$ \textit{\footnotesize {\color{CBblue}Color} is used to facilitate identification of similar or identical objects across algorithm environments} \vspace{.25em}
    \State ${\color{CBblue}h_1}=\theta_1(P,{\color{CBteal}h_{\mathrm{pa}(1)}})$
    \State ${\color{CBblue}h_2}=\theta_2(P,{\color{CBteal}h_{\mathrm{pa}(2)}})$\hfill $\triangleright$ ${\color{CBteal}h_{\mathrm{pa}(j)}}:=({\color{CBblue}h_i} : i\in\mathrm{pa}(j))\;\forall j$\vspace{-.25em}
    \Statex $\vdots$\vspace{-.2em}
    \State ${\color{CBblue}h_k}=\theta_k(P,{\color{CBteal}h_{\mathrm{pa}(k)}})$
    \State \Return $\psi(P)={\color{CBblue}h_k}$
\end{algorithmic_special}
\end{algorithm}

To ensure that a chain rule can be applied to differentiate $\psi$, each $\theta_j$ must be differentiable in a sense we now introduce. 
We do this for a generic primitive $\theta : \mathcal{M}\times\mathcal{U}\rightarrow\mathcal{V}$ with $\mathcal{U}\subseteq \mathcal{T}$ and $\mathcal{V}\subseteq\mathcal{W}$ for ambient Hilbert spaces $\mathcal{T}$ and $\mathcal{W}$. We denote the norms on these spaces by $\|\cdot\|_{\mathcal{T}}$ and $\|\cdot\|_{\mathcal{W}}$. Given $u\in\mathcal{U}$ and $t\in\mathcal{T}$, we let $\mathscr{P}(u,\mathcal{U},t)$ be the set of paths $\{u_\epsilon : \epsilon\in [0,1]\}\subseteq \mathcal{U}$ that satisfy $\|u_\epsilon-u-\epsilon t\|_{\mathcal{T}}=o(\epsilon)$.  For $u\in\mathcal{U}$, we call $\check{\mathcal{U}}_u:=\{t\in\mathcal{T} : \mathscr{P}(u,\mathcal{U},t)\not=\emptyset\}$ the tangent set of $\mathcal{U}$ at $u$ and the $\mathcal{T}$-closure of its linear span, $\dot{\mathcal{U}}_u$, the tangent space of $\mathcal{U}$ at $u$. 
The tangent space $\dot{\mathcal{M}}_P\subseteq L^2(P)$ of $\mathcal{M}$ at $P$ is similarly defined as the $L^2(P)$-closure of the linear span of the tangent set $\check{\mathcal{M}}_P:=\{s\in L^2(P) : \mathscr{P}(P,\mathcal{M},s)\not=\emptyset\}$, where $\mathscr{P}(P,\mathcal{M},s)$ is the collection of quadratic mean differentiable univariate submodels $\{P_\epsilon : \epsilon\in[0,1]\}\subseteq\mathcal{M}$ with score $s$ at $\epsilon=0$ and $P_{\epsilon=0}=P$. We call $\theta$ totally pathwise differentiable at $(P,u)$ if there exists a continuous linear operator $\dot{\theta}_{P,u} : \dot{\mathcal{M}}_P\oplus\dot{\mathcal{U}}_u\rightarrow\mathcal{W}$ such that, for all $s\in\check{\mathcal{M}}_P$, $t\in\check{\mathcal{U}}_u$, $\{P_\epsilon : \epsilon\}\in \mathscr{P}(P,\mathcal{M},s)$, and $\{u_\epsilon : \epsilon\}\in \mathscr{P}(u,\mathcal{U},t)$,
\begin{align}
    \left\|\theta(P_\epsilon,u_\epsilon) - \theta(P,u) - \epsilon\,  \dot{\theta}_{P,u}(s,t)\right\|_{\mathcal{W}}&= o(\epsilon). \label{eq:totalpd}
\end{align}
The differential operator $\dot{\theta}_{P,u}$ is unique. We denote its Hermitian adjoint by $\dot{\theta}_{P,u}^* : \mathcal{W}\rightarrow \dot{\mathcal{M}}_P\oplus\dot{\mathcal{U}}_u$.

Though we have not seen this precise notion of differentiability introduced previously, total pathwise differentiability is closely related to classical notions that apply to functions of a single argument, $u\in\mathcal{U}$ or $P\in\mathcal{M}$. In particular, the total pathwise differentiability of $\theta$ at $(P,u)$ implies $\theta(P,\,\cdot\,)$ is Hadamard differentiable \citep{averbukh1967theory} and $\theta(\,\cdot\,,u)$ is pathwise differentiable \citep{van1991differentiable}, and the converse implication holds under additional regularity conditions --- see Lemmas~\ref{lem:totalImpliesPartial} and \ref{lem:diffTheorem} in the appendix. In the special case where $\mathcal{U}$ is the trivial vector space $\{0\}$, $\theta$ is totally pathwise differentiable if and only if $\nu(\cdot):=\theta(\,\cdot\,,0)$ is pathwise differentiable.
In these cases, $\dot{\theta}_P^*(w)=(\dot{\nu}_P^*(w),0)$. The map $\dot{\nu}_P^*  : \mathcal{W}\rightarrow \dot{\mathcal{M}}_P$ is called the efficient influence operator of $\nu$ at $P$; it can be used to draw inferences about $\nu(P)$ \citep{luedtke2023one}. When $\mathcal{W}=\mathbb{R}$ and $w=1$, $\dot{\nu}_P^*(w)$ is known as the EIF of $\nu$ at $P$ \citep[][Definition 5.2.3]{bickel1993efficient}.

Like Hadamard and pathwise differentiability \citep{averbukh1967theory,van1991efficiency,kennedy2022semiparametric}, total pathwise differentiability satisfies a chain rule --- see Lemma~\ref{lem:chainRule} in the appendix. The availability of a chain rule suggests a form of automatic differentiation \citep{linnainmaa1970representation}, which we present in Algorithm~\ref{alg:backprop}. Initially, this algorithm requires the evaluation of the composition that defines $\psi$. Subsequently, information backpropagates through the indices $j$ of the composition to evaluate the efficient influence operator. This information is derived from the adjoint $\dot{\theta}_{j,P,h_{\mathrm{pa}(j)}}^*$ of the differential operator of $\theta_j$ at $(P,h_{\mathrm{pa}(j)})$, where $h_{\mathrm{pa}(j)}$ is as defined as in Algorithm~\ref{alg:parameter} when that algorithm is called with input $P$.
\begin{theorem}[Automatic differentiation works]\label{thm:backpropWorks}
Suppose $\psi$ can be expressed as in Algorithm~\ref{alg:parameter} and $\theta_j$ is totally pathwise differentiable at $(P,h_{\mathrm{pa}(j)})$ for all $j\in [k]$. Then, $\psi$ is pathwise differentiable at $P$ and Algorithm~\ref{alg:backprop} returns the efficient influence operator evaluation $\dot{\psi}_P^*(f_k)$.
\end{theorem}
To obtain the EIF of a real-valued parameter $\psi$, Algorithm~\ref{alg:backprop} can be run with $f_k=1$. 
When proving Theorem~\ref{thm:backpropWorks}, we define maps $\eta_j : \mathcal{M}\times\prod_{\ell=1}^j \mathcal{V}_\ell\rightarrow\mathcal{W}_\psi$ such that $\eta_j(P,v_{[j]})$ is the output of a modification of Algorithm~\ref{alg:parameter} that replaces each line $i\le j$ by the assignment $h_i=v_i$ and leaves each line $i>j$ unchanged, so that $h_i=\theta_i(P,h_{\mathrm{pa}(i)})$. The crux of our proof involves showing that, just before step $j$ of the loop in Algorithm~\ref{alg:backprop}, $f_0=\dot{\eta}_{j,P,h_{[j]}}^*(f_k)$.

\begin{algorithm}[tb]
   \caption{Automatic differentiation to evaluate the efficient influence operator at $P$}
   \label{alg:backprop}
   \linespread{1.05}\selectfont
\begin{algorithmic}[1]
    \Require user-specified ${\color{CBorange}f_k}\in\mathcal{W}_\psi$ and ${\color{CBblue}h_1},{\color{CBblue}h_2},\ldots,{\color{CBblue}h_k}$ as defined in Algorithm~\ref*{alg:parameter} with input $P$
    \State \textbf{initialize} ${\color{CBmagenta}f_0},{\color{CBorange}f_1},{\color{CBorange}f_2},\ldots,{\color{CBorange}f_{k-1}}$ as the 0 elements of $L^2(P),\mathcal{W}_1,\mathcal{W}_2,\ldots,\mathcal{W}_{k-1}$, respectively
    \For {$j=k,k-1,\ldots,1$} \label{ln:for}
        \State \textbf{augment} $({\color{CBmagenta}f_0},{\color{CBorange}f_{\mathrm{pa}(j)}}) \mathrel{+{=}} \dot{\theta}_{j,P,{\color{CBteal}h_{\mathrm{pa}(j)}}}^*({\color{CBorange}f_j})$ \label{ln:augment}
    \EndFor
    \State \Return ${\color{CBmagenta}f_0}$\hfill $\triangleright$ Theorem~\ref*{thm:backpropWorks} shows ${\color{CBmagenta}f_0}=\dot{\psi}_P^*({\color{CBorange}f_k})$
\end{algorithmic}
\end{algorithm}

We now discuss two strategies to increase the applicability of Algorithm~\ref{alg:backprop}. The first uses that the choice of ambient Hilbert spaces $\mathcal{W}_1,\mathcal{W}_2,\ldots,\mathcal{W}_{k-1}$ can impact whether the primitives used to express $\psi$ are totally pathwise differentiable. Hence, a good strategy is to choose these Hilbert spaces in a way that makes the primitives differentiable under mild conditions --- see the next subsection for examples. 
The second uses that differentiability is a local property, and so Theorem~\ref{thm:backpropWorks} remains true if the representation of $\psi$ in Algorithm~\ref{alg:parameter} is only valid in some Hellinger neighborhood $\mathcal{N}\subset\mathcal{M}$ of $P$. Concretely, it suffices that calling Algorithm~\ref{alg:parameter} with input $P'$ returns $\psi(P')$ for all $P'\in\mathcal{N}$; no such requirement is needed for $P'\not\in\mathcal{N}$. This relaxation may help identify compositions of primitives to automatically differentiate.

There are many forms of automatic differentiation. They differ in how derivative information is propagated. The two extremes are reverse mode \citep{linnainmaa1970representation} and forward mode \citep{wengert1964simple}. As exhibited in Algorithm~\ref{alg:backprop}, reverse mode backpropagates information through adjoints. Forward mode takes the opposite approach, propagating it forward through differential operators as the function is evaluated. There are also intermediate approaches that propagate information both forward and backward, each of which may have a different time complexity. Because identifying the best such traversal is an NP-complete problem \citep{naumann2008optimal}, many implementations of automatic differentiation employ either a simple forward or reverse mode. A general rule of thumb is that reverse mode will be faster than forward mode when a function's input dimension exceeds its output dimension \citep[Chapter 3 of][]{griewank2008evaluating}. In our context, the input dimension is infinite unless $\mathcal{M}$ is parametric. Consequently, we pursue reverse mode automatic differentiation in this work. We leave consideration of alternative traversal strategies to future study.

\subsection{Useful primitives}\label{sec:primitives}

Table~\ref{tab:primitives} provides a selection of useful primitives. In Appendix~\ref{app:primitives}, we establish that, under conditions, each is totally pathwise differentiable when $\mathcal{M}$ is a locally nonparametric model of distributions on $\mathcal{Z}$. Some of these primitives involve a coarsened random variable $X=\mathscr{C}(Z)$, where $\mathscr{C} : \mathcal{Z}\rightarrow\mathcal{X}$ is a many-to-one map; when $Z\sim Q\in\mathcal{M}$, we let $Q_X$ denote the marginal distribution of $X$. The primitives fall into three distinct categories, which we elaborate on below.

{
\renewcommand{\arraystretch}{1.225} 
\begin{table}[htb]\small
\centering
\caption{Primitives $\theta$ with domain $\mathcal{M}\times\mathcal{U}$. Given the specified ambient Hilbert spaces $\mathcal{T}\supseteq \mathcal{U}$ and $\mathcal{W}\supseteq \mathrm{Image}(\theta)$, each is totally pathwise differentiable under conditions given in Appendix~\ref{app:primitives}. \\[.5em] 
{\footnotesize\textit{Notation:} 
$Q\in\mathcal{M}$; $X\in\mathcal{X}$ and $A\in\{0,1\}$ are coarsenings or subvectors of $Z$;  
$Q_X$, $Q_{A,X}$ = marginal distribution of $X$, $(A,X)$ when $Z\sim Q$;
$L^\star(\rho)$ = a set of $\rho$-a.s. equivalence classes, each satisfying a moment condition, where $\rho$ and $Q$ are equivalent measures; $L^\star(\rho_X)$ is defined similarly but with $Q$ replaced by $Q_X$; 
$\Pi_P(u)(z)=u(z)-E_P[u(Z)\mymid X=x]$; 
$q_X,p_X$ ($q_{Z|X},p_{Z|X}$) = marginal (conditional) densities under sampling from $Q,P$;
$\lambda,\lambda_1,\lambda_2$ = $\sigma$-finite measures.
}
}
\resizebox{\textwidth}{!}{
 \begin{tabular}{l | c | c | l | l | l  } 
 \hline
  & $\mathcal{U}$ & $\mathcal{T}$ & \multicolumn{1}{|c|}{$\mathcal{W}$} & \multicolumn{1}{|c|}{$\theta(P,u)$} & \multicolumn{1}{|c}{$\dot{\theta}_{P,u}^*(w)$} \\ [0.5ex] 
 \hline\hline
 \hyperref[app:condExp]{Conditional mean} & $L^\star(\rho)$ & $L^2(Q)$ & $L^2(Q_X)$ & $E_P[u(Z)\mymid X=\cdot\;]$ & $\left(\frac{q_X}{p_X}[u-\theta(P,u)]w,\frac{p_{Z|X}}{q_{Z|X}}w\right)$ \\\hline 
 \hyperref[app:multilinearForm]{$r$-fold conditional mean} & $L^\star(\rho^r)$ & $L^2(Q^r)$ & $L^2(Q_X^r)$ & $E_{P^r}[u(Z_{[r]})\mymid X_{[r]}=\cdot\;]$ & see appendix \\\hline
\hyperref[app:condVar]{Conditional variance} & $L^\star(\rho)$ & $L^2(Q)$ & $L^2(Q_X)$ & $\mathrm{Var}_P[u(Z)\mymid X=\cdot\;]$ & $\left(\frac{q_X}{p_X}[\Pi_P(u)^2-\theta(P,u)]w,2\frac{p_{Z|X}}{q_{Z|X}}\Pi_P(u) w\right)$ \\\hline 
 \hyperref[app:condCovar]{Conditional covariance} & $L^\star(\rho)^2$ & $L^2(Q)^{\oplus 2}$ & $L^2(Q_X)$ & $\mathrm{cov}_P[u_1(Z),u_2(Z)\mymid X=\cdot\;]$ & see appendix \\\hline 
  \hyperref[app:kernelEmbed]{Kernel embedding} & $L^\star(\rho_X)$ & $L^2(Q_X)$ & RKHS & $\int K(x,\,\cdot\,) u(x)\, P_X(dx)$ & $\left(wu-\int wu \,dP_X,\frac{p_X}{q_X}w\right)$ \\\hline
 \hyperref[app:optimalValue]{Optimal value} & $\mathcal{U}$ & $\mathcal{T}$ & $\mathbb{R}$ & $\min_{y\in\mathcal{Y}} F_y(P,u)$ & $\dot{F}_{\argmin_yF_y(P,u),P,u}^*$ \\\hline
  \hyperref[app:optimalSolution]{Optimal solution} & $\mathcal{U}$ & $\mathcal{T}$ & $\mathbb{R}^d$ & $\argmin_{w\in \mathcal{W}} F_w(P,u)$ & see appendix \\\hline  \hline
 \hyperref[app:pathwiseDiff]{Pathwise diff. parameter} & \multicolumn{2}{|c|}{$\{0\}$} & $\mathcal{W}$ & $\nu(P)$ & $\left(\dot{\nu}_P^*(w),0\right)$ \\\hline
 $\drsh$  \hyperref[app:rootDensity]{Root-density} & \multicolumn{2}{|c|}{$\{0\}$} & $L^2(\lambda)$ & $(dP/d\lambda)^{1/2}$ &  $(\frac{w}{2\nu(P)}-\int \frac{w}{2\nu(P)}\,dP,0)$ \\\hline 
 $\drsh$  \hyperref[app:condDensity]{Conditional density} & \multicolumn{2}{|c|}{$\{0\}$} & $L^2(Q)$ & $p_{Y|X}$, with $Z=(X,Y)$ &  $\big(\Pi_P(\frac{q}{p}wp_{Y|X}),0\big)$ \\\hline 
 $\drsh$  \hyperref[app:doseResponse]{Dose-response function} & \multicolumn{2}{|c|}{$\{0\}$} & $L^2(\lambda)$ & $\int E_P[Y|A=\cdot\,,X=x] P_X(dx)$ & see appendix \\\hline
  $\drsh$ \hyperref[app:countDens]{Counterfactual density} & \multicolumn{2}{|c|}{$\{0\}$} & $L^2(\lambda)$ & $y\mapsto E_{P}[p(y| A=1,X)]$ & see appendix \\\hline
 \hline
 \hyperref[app:hadDiff]{Hadamard diff. map} & $\mathcal{U}$ & $\mathcal{T}$ & $\mathcal{W}$ & $\zeta(u)$ & $\big(0,\dot{\zeta}_u^*(w)\big)$ \\\hline
 $\drsh$ \hyperref[app:innerProd]{Inner product} & \multicolumn{2}{|c|}{$\mathcal{R}\oplus\mathcal{R}$} & $\mathbb{R}$ & $\langle u_1,u_2\rangle_{\mathcal{R}}$, with $u=(u_1,u_2)$ & $\left(0,(wu_2,wu_1)\right)$ \\\hline
 $\drsh$ \hyperref[app:squaredNorm]{Squared norm} & \multicolumn{2}{|c|}{$\mathcal{T}$} & $\mathbb{R}$ & $\| u\|_{\mathcal{T}}^2$ & $\left(0,2wu\right)$ \\\hline
 $\drsh$ \hyperref[app:diffFun]{Differentiable function} & \multicolumn{2}{|c|}{$\mathbb{R}^d$} & $\mathbb{R}$ & $\zeta(u)$ & $(0,w\, \nabla \zeta(u))$ \\\hline
 $\drsh$ \hyperref[app:pointwise]{Pointwise operations} & \multicolumn{2}{|c|}{$L^2(\lambda)^{\oplus d}$} & $L^2(\lambda)$ & $x\mapsto f_x\circ \bar{u}(x)$, with $\bar{u}(x)=(u_j(x))_{j=1}^d$ & $\left(0,(x\mapsto w(x)\nabla_j f_x\circ \bar{u}(x))_{j=1}^d\right)$ \\\hline
 $\drsh$ \hyperref[app:affine]{Bounded affine map} & \multicolumn{2}{|c|}{$\mathcal{T}$} & $\mathcal{W}$ & $\kappa(u)+c$, with $c\in\mathcal{W}$, bounded linear $\kappa$ & $\big(0,\kappa^*(w)\big)$ \\\hline
 \phantom{$\drsh$ }{\color{myDarkGrey}$\drsh$} \hyperref[app:constant]{Constant map} & \multicolumn{2}{|c|}{$\{0\}$} & $\mathcal{W}$ & $c$, with $c\in\mathcal{W}$ & $\big(0,0\big)$ \\\hline
 \phantom{$\drsh$ }{\color{myDarkGrey}$\drsh$} \hyperref[app:coorProj]{Coordinate projection} & \multicolumn{2}{|c|}{$\mathcal{R}_1\oplus\mathcal{R}_2$} & $\mathcal{R}_1$ & $u_1$, with $u=(u_1,u_2)$ & $\left(0,(w,0)\right)$ \\\hline
 \phantom{$\drsh$ }{\color{myDarkGrey}$\drsh$} \hyperref[app:changeMeasure]{Change of measure} & \multicolumn{2}{|c|}{$L^2(\lambda_1)$} & $L^2(\lambda_2)$ & $u$ & $(0,\frac{d\lambda_2}{d\lambda_1} w)$ \\\hline
 \phantom{$\drsh$ }{\color{myDarkGrey}$\drsh$} \hyperref[app:lifting]{Lift  to new domain} & \multicolumn{2}{|c|}{$L^2(Q_X)$} & $L^2(Q)$ & $z\mapsto u(x)$ & $\big(0,E_Q[w(Z)|X=\cdot\;]\big)$ \\\hline
 \phantom{$\drsh$ }{\color{myDarkGrey}$\drsh$} \hyperref[app:fix]{Fix binary argument} & \multicolumn{2}{|c|}{$L^2(Q_{A,X})$} & $L^2(Q_X)$ & $x\mapsto u(1,x)$ & $\big(0,(a,x)\mapsto \frac{a}{Q(A=1|x)}w(x)\big)$ \\\hline
 \hline
 \end{tabular}
 }
 \label{tab:primitives}
\end{table}
}

The first consists of maps $\theta$ that depend on both their distributional input $P$ and Hilbert-valued input $u$. A simple example is a conditional mean operator $\theta(P,u)(\cdot)=E_P[u(Z)\mid X=\,\cdot\,]$, with $u$ a real-valued function. An $r$-fold conditional mean is a slightly more complex case, corresponding to a conditional variant of the estimand pursued by U-statistics \citep{hoeffding1948class}. Here, $u$ is a function of $r\ge 2$ independent observations and $\theta(P,u)(\cdot)=E_{P^r}[u(Z_1,Z_2,\ldots,Z_r)\mid (X_1,X_2,\ldots,X_r)=\cdot\,]$. Conditional variances and covariances are other examples, as are the optimal value and solution of optimization problems whose objective functions depend on the data-generating distribution $P$ and a Hilbert-valued input $u$. Kernel embedding operators $\theta(P,u)=\int K(x,\,\cdot\,) u(x)\, P_X(dx)$ represent another example, where $u$ is a bounded function and $K$ is a bounded kernel of a reproducing kernel Hilbert space (RKHS) \citep[Chapter 4.9.1.1 of][]{berlinet2011reproducing}. This primitive has been used to construct tests of equality of distributions \citep{gretton2012kernel,muandet2021counterfactual} and independence \citep{gretton2007kernel}.

The second category consists of maps $\theta$ that solely depend on their distribution-valued input, so that there is some pathwise differentiable parameter $\nu : \mathcal{M}\rightarrow\mathcal{W}$ such that $\theta(P,u)=\nu(P)$ for all $(P,u)\in\mathcal{M}\times\mathcal{U}$. Since $u$ plays no role in determining the value of $\theta(P,u)$, it suffices to assume that $\mathcal{U}$ is a trivial vector space $\{0\}$. Examples of $\nu$ include root density functions, conditional density functions, dose-response functions \citep{diaz2013targeted}, and counterfactual density functions \citep{kennedy2021semiparametric}, as detailed in \cite{luedtke2023one}. That work also shows that conditional average treatment effect functions \citep{hill2011bayesian} are pathwise differentiable.

The third category consists of maps $\theta$ that solely depend on their Hilbert-valued input, so that there is some Hadamard differentiable map $\zeta : \mathcal{U}\rightarrow\mathcal{W}$ such that $\theta(P,u)=\zeta(u)$ for all $(P,u)$. Simple examples are inner products, squared norms, and finite-dimensional differentiable functions. Pointwise operations provide a richer example --- for instance, $\zeta$ may take as input functions $u_1$ and $u_2$ and return $x\mapsto f(u_1(x),u_2(x))$ with $f(a,b):=[a+b]\mathrm{cos}(a)$. The pointwise operation may also depend on $x$, as would be the case if $\zeta(u_1,u_2)(x)=f_x(u_1(x),u_2(x))$ with $f_x(a,b)=(a+x)\mathrm{cos}(b)$. 

Bounded affine maps provide further examples of primitives from the third category. Simple special cases are given by constant maps and coordinate projections. Changes of measure are another special case, so that for measures $\lambda_2\ll \lambda_1$ with $d\lambda_2/d\lambda_1$ essentially bounded, $\zeta(u)$ is the isometric embedding of $u\in L^2(\lambda_1)$ into $L^2(\lambda_2)$; this map facilitates the chaining of primitives $\zeta_1 : \mathcal{U}_0\rightarrow L^2(\lambda_1)$ and $\zeta_2 : L^2(\lambda_2)\rightarrow \mathcal{W}_0$ via $\zeta_2\circ \zeta\circ \zeta_1$. Primitives $\zeta_1 : \mathcal{U}_0\rightarrow L^2(Q_X)$ and $\zeta_2 : L^2(Q)\rightarrow \mathcal{W}_0$ can similarly be chained together using a lifting map $\zeta : L^2(Q_X)\rightarrow L^2(Q)$. Further examples of affine maps are those that fix a binary argument to unity, so that $\zeta(u)(x)=u(1,x)$. When paired with a conditional mean operator $(P,u)\mapsto E_P[u(Z)\mymid (A,X)=\cdot\,]$, this makes it possible to express $(P,u)\mapsto E_P[u(Z)\mymid A=1,X=\cdot\,]$ as a composition of totally pathwise differentiable primitives; such conditional expectations appear frequently in parameters arising in causal inference \citep{robins1986new}. Bounded affine maps $\zeta$ can also be composed with conditional mean primitives to express affine functionals of regressions commonly studied in the literature \citep{chernozhukov2018learning,chernozhukov2022automatic}, namely $\psi(P)=E_P[\zeta(E_P[Y\mymid X=\cdot\,])(Z)]$ with $\zeta : L^2(Q_X)\rightarrow L^2(Q)$.

The choice of ambient Hilbert spaces $\mathcal{T}$ and $\mathcal{W}$ can impact a primitive's total pathwise differentiability. For the first five primitives in Table~\ref{tab:primitives} and the conditional density primitive, this choice involves selecting a $Q\in\mathcal{M}$ indexing one or more $L^2$ spaces. As detailed in Appendix~\ref{app:primitives}, sufficient conditions for each primitive's total pathwise differentiability at a given $(P,u)$ require one or more of the following to be bounded: $dP/dQ$, $dQ/dP$, $dP_X/dQ_X$, or $dQ_X/dP_X$. If $Q=P$ then these functions are all equal to unity, and so these boundedness conditions certainly hold. Given this, when composing these primitives to differentiate $\psi$ at $P$ we generally recommend choosing $Q=P$ to index the ambient Hilbert spaces. 

In Appendix~\ref{app:primitiveRemarks}, we describe how to transform primitives for conditional quantities to those for marginal quantities. This provides a way to derive a marginal mean primitive, for example. We also propose a way to adapt the primitives from nonparametric to semiparametric models, though we caution that this proposal requires projecting onto the tangent space of the model and so implementing it can require specialized expertise.

\subsection{Example: nonparametric $R^2$}\label{sec:illustrationR2}

We apply Algorithm~\ref{alg:backprop} to compute the EIF of the nonparametric $R^2$ parameter at a distribution $P$ in a nonparametric model $\mathcal{M}$. To avoid notational overload on $P$, we denote a generic element of $\mathcal{M}$ by $P'$ when presenting this example.

\begin{figure}[b]
    \centering
   \includegraphics[width=\textwidth]{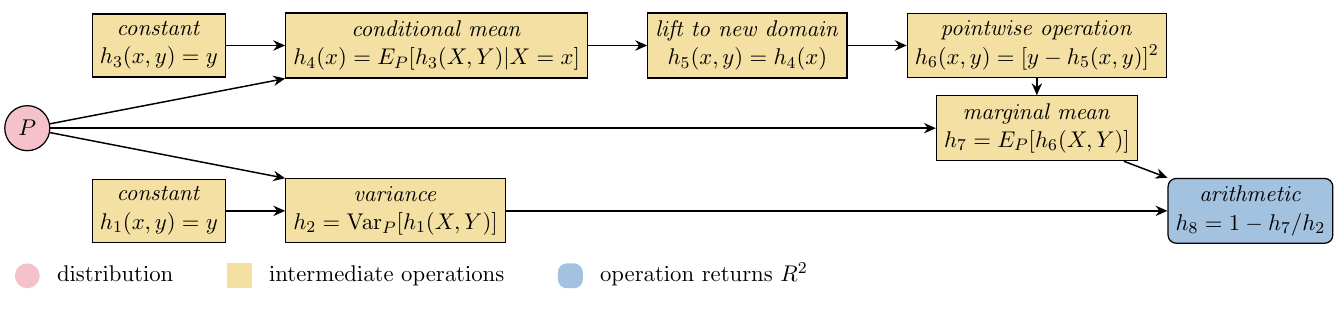}

    \caption{Computation graph for the nonparametric $R^2$ parameter. 
    }\label{fig:compGraph}
\end{figure}

We express $\psi(P'):=1-E_{P'}[\{Y-\mu_{P'}(X)\}^2]/\mathrm{Var}_{P'}(Y)$ as a composition of eight primitives: a constant $\theta_1(P',0)=(z\mapsto y)\in L^2(P)$, variance $\theta_2(P',h_1)=\mathrm{Var}_{P'}[h_1(Z)]\in \mathbb{R}$, constant $\theta_3(P',0)=(z\mapsto y)\in L^2(P)$, conditional mean $\theta_4(P',h_3)(\cdot)=E_{P'}[h_3(Z)\mymid X=\cdot\,]\in L^2(P_X)$, lifting $\theta_5(P',h_4)=(z\mapsto h_4(x))\in L^2(P)$, pointwise operation $\theta_6(P',h_5)=(z\mapsto [y-h_5(z)]^2)\in L^2(P)$, mean $\theta_7(P',h_6)=E_{P'}[h_6(Z)]\in\mathbb{R}$, and differentiable function $\theta_8(P',(h_2,h_7))=1-h_7/h_2\in \mathbb{R}$. Here, $\mathrm{pa}(1)=\mathrm{pa}(3)=\emptyset$, $\mathrm{pa}(2)=\{1\}$, $\mathrm{pa}(4)=\{3\}$, $\mathrm{pa}(5)=\{4\}$, $\mathrm{pa}(6)=\{5\}$, $\mathrm{pa}(7)=\{6\}$, and $\mathrm{pa}(8)=\{2,7\}$. The redundancy of the constant maps $\theta_1$ and $\theta_3$ has no impact on the applicability of Algorithm~\ref{alg:backprop}. For the composition of primitives to express $\psi$, it must return $\psi(P')$ for every $P'\in\mathcal{M}$. This will hold if each claimed element of $L^2(P)$ or $L^2(P_X)$ is a $P$-a.s. equivalence class of functions that satisfies an appropriate second-moment condition. To this end, we suppose all pairs of distributions in $\mathcal{M}$ are mutually absolutely continuous, which makes it so, for each $P'\in\mathcal{M}$, $P'$-a.s. and $P$-a.s. equivalence classes of functions are one and the same. For the moment condition, we suppose there exists $c<\infty$ such that $P'(|Y|\le c)=1$ for all $P'\in\mathcal{M}$. To avoid dividing by zero, we also suppose $\mathrm{Var}_{P'}(Y)>0$ for all $P'$.

Figure~\ref{fig:compGraph} illustrates the computation graph for our representation of $\psi$ \citep{linnainmaa1970representation,bauer1974computational}. In it, each primitive is a node. A directed arrow points from node $i$ to node $j$ if $i\in\mathrm{pa}(j)$, and from $P$ to node $j$ if $\theta_j$ depends nontrivially on its distribution-valued argument. Table~\ref{tab:backpropIllustrationR2} provides a step-by-step illustration of how Algorithm~\ref{alg:backprop} computes the EIF at $P$.  As anticipated in Section~\ref{sec:contributions}, this manual application of automatic differentiation has yielded an analytical expression for the EIF.

\begin{table}[tb]
\centering
\caption{Step-by-step evaluation of Algorithm~\ref{alg:backprop} for the nonparametric $R^2$ parameter when it is expressed as in Figure~\ref{fig:compGraph}. Each row $j$ lists the values $f_0,f_1,\ldots,f_j$ take just before executing line~\ref{ln:augment} of the for loop in Algorithm~\ref{alg:backprop}. The value of $f_0$ in the bottom row is the {\color{CBmagenta}efficient influence function}.\\[.5em] 
{\footnotesize\textit{Notation:}  
$\mu_P(x)=E_P[Y\mymid X=x]$, $\sigma_P^2=\mathrm{Var}_P(Y)$, and $\gamma_P=E_P[\mathrm{Var}_P(Y\mymid X)]$.}}
 \resizebox{\textwidth}{!}{
\begin{tabular}{c || c | c | c | c | c | c | c | c | c} 
 \multirow{2}{*}{\shortstack{Step\\($j$)}} & \multicolumn{7}{|c}{Variables in Algorithm~\ref{alg:backprop}}\\[.2em]
 & \hspace{.25em}$f_8$\hspace{.25em} & $f_7$ & $f_6$ & $f_5$ & $f_4$& \hspace{-.25em}\phantom{$^\S$\hspace{.25em}}$f_3$\hspace{.25em}$^\S$\hspace{-.25em} & \hspace{.5em}$f_2$\hspace{.5em} & \hspace{-.25em}\phantom{$^\S$\hspace{.25em}}$f_1$\hspace{.25em}$^\S$\hspace{-.25em} & $f_0$ \\\hline\hline
8 & $1$\vphantom{$\frac{1}{\sigma_P^2}$} & $0$ & $0$ & $0$ & $0$ & $-$ & $0$ & $-$ & $0$ \\
7 & & $-\frac{1}{\sigma_P^2}$ & $0$ & $0$ & $0$ & $-$ & $\frac{\gamma_P}{\sigma_P^{4}}$ & $-$ & $0$ \\
6 & & & $-\frac{1}{\sigma_P^2}$ & $0$ & $0$ & $-$ & $\frac{\gamma_P}{\sigma_P^{4}}$ & $-$ & $z\mapsto -\frac{[y-\mu_P(x)]^2-\gamma_P}{\sigma_P^{2}}$ \\
5 & & & & $z\mapsto \frac{2[y-\mu_P(x)]}{\sigma_P^2}$ & $0$ & $-$ & $\frac{\gamma_P}{\sigma_P^{4}}$ & $-$ & $z\mapsto -\frac{[y-\mu_P(x)]^2-\gamma_P}{\sigma_P^{2}}$ \\
4 & & & & & \phantom{ $^\dagger$}$0$ $^\dagger$ & $-$ & $\frac{\gamma_P}{\sigma_P^{4}}$ & $-$ & $z\mapsto -\frac{[y-\mu_P(x)]^2-\gamma_P}{\sigma_P^{2}}$ \\
3 & & & & & & $-$ &  $\frac{\gamma_P}{\sigma_P^{4}}$ & $-$ & $z\mapsto -\frac{[y-\mu_P(x)]^2-\gamma_P}{\sigma_P^{2}}$ \\
2 & & & & & & & $\frac{\gamma_P}{\sigma_P^{4}}$ & $-$ & $z\mapsto -\frac{[y-\mu_P(x)]^2-\gamma_P}{\sigma_P^{2}}$ \\
1 & & & & & & & & $-$ & $z\mapsto \frac{[\{y-E_P[Y]\}^2-\sigma_P^2]\gamma_P}{\sigma_P^{4}}-\frac{[y-\mu_P(x)]^2-\gamma_P}{\sigma_P^2}$ \\
0 & & & & & & & & & {\color{CBmagenta}$\bm{z\mapsto \frac{[\{y-E_P[Y]\}^2-\sigma_P^2]\gamma_P}{\sigma_P^{4}}-\frac{[y-\mu_P(x)]^2-\gamma_P}{\sigma_P^2}}$} \\
\end{tabular}
}\vspace{1em}
\begin{minipage}{\linewidth}\footnotesize
\begin{itemize}[leftmargin=*]
    \item[$^\S$] $\theta_3$ and $\theta_1$ are constant maps and so have derivative zero. Hence, backpropagating $f_3$ and $f_1$ on steps $j=3$ and $j=1$ has no impact on earlier nodes ($f_i : i<j$), and so these values need not be stored.
    \item[$^\dagger$] $\theta_5 : L^2(P_X)\rightarrow L^2(P)$ is a lifting, so $f_4$ updates as $f_4\mathrel{+{=}}E_P[f_5(Z)\mymid X=\cdot\,]$. This update equals zero.
\end{itemize}
\end{minipage}
 \label{tab:backpropIllustrationR2}
\end{table}

In Appendix~\ref{app:illustratingBackprop}, we apply Algorithm~\ref{alg:backprop} to two other parameters: the expected density $\psi(P)=E_P[p(Z)]$ and expected conditional covariance $\psi(P)=E_P[\mathrm{cov}_P(A\,,Y\mymid X)]$.

\section{Probabilistic programming of efficient estimators}\label{sec:probProg}

\subsection{Estimator construction}\label{sec:estimator}

Suppose now that we observe independent draws $Z_1,Z_2,\ldots,Z_n$ from an unknown $P\in\mathcal{M}$. Our goal is to infer $\psi(P)$. We do this using a variant of a one-step estimator \citep{pfanzagl1982contributions}, which debiases an initial estimator using the efficient influence operator.

\begin{algorithm}[tb]
   \caption{Obtain initial estimates of $\psi(P)$ and $\dot{\psi}_P^*(f_k)$ using data drawn from $P$}
   \label{alg:estimatedBackprop}
   \linespread{1.05}\selectfont
\begin{algorithmic}[1]
    \Require $f_k\in\mathcal{W}_\psi$ and data $Z_{\mathscr{I}}:=\{Z_i : i\in\mathscr{I}\subseteq [n]\}$ for nuisance estimation on lines~\ref*{ln:estForward} and \ref*{ln:estBackward}
    \algrule
    \Statex $\triangleright$ {\color{CBblue}\textbf{forward pass}} to obtain an initial estimate
    \For {$j=1,2,\ldots,k$}
        \State \textbf{estimate} $\theta_j(P,{\color{CBteal}\widehat{h}_{\mathrm{pa}(j)}})$ with {\color{CBblue}$\widehat{h}_j$} \hfill $\triangleright$ ${\color{CBteal}\widehat{h}_{\mathrm{pa}(j)}} :=({\color{CBblue}\widehat{h}_i} : i \in\mathrm{pa}(j))$ \label{ln:estForward}
    \EndFor
    \algrule
    \Statex $\triangleright$ {\color{CBmagenta}\textbf{backward pass}} to estimate efficient influence operator
    \State \textbf{initialize} ${\color{CBmagenta}\widehat{f}_0}:z\mapsto 0$; ${\color{CBorange}\widehat{f}_1},{\color{CBorange}\widehat{f}_2},\ldots,{\color{CBorange}\widehat{f}_{k-1}}$ as the 0 elements of $\mathcal{W}_1,\mathcal{W}_2,\ldots,\mathcal{W}_{k-1}$; and ${\color{CBorange}\widehat{f}_k}=f_k$
    \For {$j=k,k-1,\ldots,1$} 
        \State \textbf{estimate} $\dot{\theta}_{j,P,{\color{CBteal}\widehat{h}_{\mathrm{pa}(j)}}}^*({\color{CBorange}\widehat{f}_j})$ with a {\color{CBred}$\widehat{\vartheta}_j$} whose first entry, {\color{CBred}$\widehat{\vartheta}_{j,0}$}, is compatible with {\color{CBblue}$\widehat{h}_j$}, {\color{CBteal}$\widehat{h}_{\mathrm{pa}(j)}$}\label{ln:estBackward} 
        \Statex \hspace{\algorithmicindent} $\triangleright$ definition of compatibility: $\exists\, {\color{slateGrey}\widehat{P}_j}\in\mathcal{M}$ s.t. ${\color{CBblue}\widehat{h}_j}=\theta_j({\color{slateGrey}\widehat{P}_j},{\color{CBteal}\widehat{h}_{\mathrm{pa}(j)}})$ and $\int {\color{CBred}\widehat{\vartheta}_{j,0}}(z)\,{\color{slateGrey}\widehat{P}_j}(dz)=0$
        \State \textbf{augment} $({\color{CBmagenta}\widehat{f}_0},{\color{CBorange}\widehat{f}_{\mathrm{pa}(j)}}) \mathrel{+{=}} {\color{CBred}\widehat{\vartheta}_j}$ \label{ln:estAugment}
    \EndFor
    \algrule
    \State \Return initial estimate {\color{CBblue}$\widehat{h}_k$} and estimated efficient influence operator ${\color{CBmagenta}\widehat{f}_0}$
\end{algorithmic}
\end{algorithm}

Algorithm~\ref{alg:estimatedBackprop} presents a method to obtain this initial estimate and the evaluation of the efficient influence operator at a user-specified $f_k$. The algorithm employs a forward and backward pass. In the forward pass, it recursively estimates the $P$-dependent evaluations of the operators $\theta_j$ that appear in Algorithm~\ref{alg:parameter}, ultimately yielding the initial estimate $\widehat{h}_k$ of $\psi(P)$. In the backward pass, it recursively estimates $P$-dependent adjoint evaluations and backpropagates them as in Algorithm~\ref{alg:backprop}, yielding an estimate $\widehat{f}_0$ of the efficient influence operator evaluation $\dot{\psi}_P^*(f_k)$. Black-box machine learning tools may be used to estimate all of the needed nuisances. When $\psi$ is real-valued and $f_k=1$, $\widehat{f}_0$ is an estimate of the EIF. To fix ideas, we focus on this real-valued setting in the remainder of the main text and refer the reader to Appendix~\ref{app:HilbertValParam} for a discussion of more general settings where $\psi$ takes values in a generic separable Hilbert space, such as an $L^2$ space---these more general settings cover, for example, the case where $\psi$ returns a square integrable dose-response function..

Inspired by one-step estimation, the outputs of Algorithm~\ref{alg:estimatedBackprop} can be used to construct 
\begin{align*}
    \widehat{\psi}:=\widehat{h}_k + \frac{1}{n}\sum_{i=1}^n \widehat{f}_0(Z_i)
\end{align*}
and 95\% confidence interval $\widehat{\psi}\pm 1.96 \,\widehat{\sigma}/n^{1/2}$, where $\widehat{\sigma}^2$ is the empirical variance of $\widehat{f}_0$. In fact, $\widehat{\psi}$ is exactly the traditional one-step estimator $\psi(\widehat{P})+\frac{1}{n}\sum_{i=1}^n \dot{\psi}_{\widehat{P}}^*(1)(Z_i)$ under the strong compatibility condition that there exists a $\widehat{P}\in\mathcal{M}$ such that $\widehat{h}_j=\theta_j(\widehat{P},\widehat{h}_{\mathrm{pa}(j)})$ and $\widehat{\vartheta}_j=\dot{\theta}_{j,\widehat{P},\widehat{h}_{\mathrm{pa}(j)}}^*(\widehat{f}_j)$ for all $j\in [k]$. To ease implementation, Algorithm~\ref{alg:estimatedBackprop} makes a weaker compatibility requirement that allows each $\widehat{\vartheta}_{j,0}$, $\widehat{h}_j$, and $\widehat{h}_{\mathrm{pa}(j)}$ to be compatible with a $\widehat{P}_j$ that may depend on $j$. This dependence makes it so $\widehat{\psi}$ may not be a traditional one-step estimator; the benefit of permitting this is that it allows each primitive to be paired with a nuisance estimation routine that can be implemented once and reused across different parameters $\psi$---see the next subsection for details.

Because $\widehat{\psi}$ is not generally a traditional one-step estimator, we state sufficient conditions for its efficiency here. These conditions concern the same objects as for a traditional one-step estimator \citep{bickel1993efficient}, namely a remainder  $\mathcal{R}_n:=\widehat{h}_k + \int \widehat{f}_0(z)\, P(dz)  - \psi(P)$ and drift term $\mathcal{D}_n:=\int [\widehat{f}_0(z)-f_0(z)] (P_n-P)(dz)$, with $P_n$ the empirical distribution and $f_0:=\dot{\psi}_P^*(1)$ the EIF. 
\begin{theorem}[Efficiency of $\widehat{\psi}$]\label{thm:efficient}
    Suppose $\psi$ is real-valued, $f_k=1$, and the conditions of Theorem~\ref{thm:backpropWorks} hold. If $\mathcal{D}_n$ and $\mathcal{R}_n$ are $o_p(n^{-1/2})$, then $\widehat{\psi}$ is an efficient estimator of $\psi(P)$, in that
    \begin{align*}
        \widehat{\psi} - \psi(P)=\frac{1}{n}\sum_{i=1}^n f_0(Z_i) + o_p(n^{-1/2}).
    \end{align*}
\end{theorem}
Our Python package implements a cross-fitted counterpart of $\widehat{\psi}$, which we present and study in Appendix~\ref{app:estReal}. Like $\widehat{\psi}$, this estimator is efficient when a remainder and drift term are negligible (see Theorem~\ref{thm:CFefficient}). The cross-fitted estimator may be preferred over $\widehat{\psi}$ because its drift term is $o_p(n^{-1/2})$ even without an empirical process condition \citep{schick1986asymptotically}. This is in contrast to $\widehat{\psi}$, for which such a condition is often used to ensure negligibility of that term  \citep[Lemma 19.24 of][]{van2000asymptotic}.

The compatibility condition in Algorithm~\ref{alg:estimatedBackprop} helps make $\mathcal{R}_n=o_p(n^{-1/2})$ plausible. To see this, we add and subtract terms and use that $\int \widehat{\vartheta}_{j,0}\, d\widehat{P}_j=0$ and $\widehat{f}_0=\sum_{j=1}^k \widehat{\vartheta}_{j,0}$ to reexpress $\mathcal{R}_n$ as
\begin{align*}
    \mathcal{R}_n&= \left[\widehat{h}_k - \psi(P) - \sum_{j=1}^k \int \vartheta_{j,0}(z) (\widehat{P}_j-P)(dz)\right] + \left[\sum_{j=1}^k \int [\vartheta_{j,0}(z)-\widehat{\vartheta}_{j,0}(z)] (\widehat{P}_j-P)(dz)\right],
\end{align*}
where $\vartheta_{j,0}$ is the first entry of the augmentation term $\dot{\theta}_{j,P,h_{\mathrm{pa}(j)}}^*(f_j)$ from Algorithm~\ref{alg:backprop}. If both terms on the right are $o_p(n^{-1/2})$, then $\mathcal{R}_n$ is as well. Moreover, since the above holds for any $\widehat{P}_j$ compatible with the nuisances $\widehat{h}_j$, $\widehat{h}_{\mathrm{pa}(j)}$, and $\widehat{\vartheta}_{j,0}$, it suffices that there exist such $\widehat{P}_j$, $j\in [k]$, that make each of the two terms on the right $o_p(n^{-1/2})$. By Cauchy-Schwarz, the magnitude of the latter term upper bounds by $\sum_{j=1}^k \|d\widehat{P}_j/dP-1\|_{L^2(P)}\|\widehat{\vartheta}_{j,0}-\vartheta_{j,0}\|_{L^2(P)}$, which is $o_p(n^{-1/2})$ under an $n^{-1/4}$-rate condition akin to those arising in the study of traditional one-step estimators. In Appendix~\ref{app:vonMises}, we show that the first term on the right is the remainder in a von Mises expansion, which makes it so it too will be $o_p(n^{-1/2})$ under appropriate conditions. Similar arguments have been made previously to motivate one-step estimators when $\psi$ is a generic pathwise differentiable parameter \citep{robins2009quadratic,luedtke2023one}.

For a particular parameter $\psi$, the remainder $\mathcal{R}_n$ can also be analyzed explicitly. In Appendix~\ref{app:vonMisesR2}, we provide this analysis for the nonparametric $R^2$ parameter, and we also provide an analogue of Table~\ref{tab:backpropIllustrationR2} that helps to visualize Algorithm~\ref{alg:estimatedBackprop}.

\subsection{Software implementation}\label{sec:software}

To implement Algorithm~\ref{alg:estimatedBackprop} in software, it is necessary to convert a computer code representation of the parameter $\psi$ --- such as the one given on lines \ref{ln:vdef}-\ref{ln:r2def} of Example~\ref{ex:r2} --- into the compositional form in Algorithm~\ref{alg:parameter}. Any computer program, including those with function calls and control flow structures like loops and conditionals, can be represented via its computation graph \citep{linnainmaa1970representation,bauer1974computational}. When the operations in this graph are totally pathwise differentiable primitives, as they are in Figure~\ref{fig:compGraph}, this graph yields an expression for $\psi$ as in Algorithm~\ref{alg:parameter}. Representing programs via their computation graphs enables user-friendly automatic differentiation, as seen in software like TensorFlow, PyTorch, and JAX \citep{tensorflow2015-whitepaper,jax2018github,NEURIPS2019_9015}.

Our proof-of-concept package also adopts computation graphs, using a static graph structure \citep{baydin2018automatic}. When defining the parameter, each primitive call takes dimple objects as inputs and returns another dimple object. At this stage, a dimple object comprises a primitive and a set of pointers to its inputs, serving the roles of $\theta_j$ and $\mathrm{pa}(j)$ in Algorithm~\ref{alg:parameter}. The dimple objects can be used to express $\psi$ as in Algorithm~\ref{alg:parameter}, and also provide the form of the computation graph. 
Estimation --- as executed, for instance, via `\verb|estimate(R2)|' in Example~\ref{ex:r2} --- calls Algorithm~\ref{alg:estimatedBackprop} with this computation graph to populate each dimple object with initial estimates $\widehat{h}_j$ in the forward pass, with adjoint estimates $\widehat{f}_j$ added in the backward pass.

Primitives can be implemented independently of one another and reused across different parameters $\psi$. In this modular approach, each primitive $\theta_j$ in Algorithm~\ref{alg:estimatedBackprop} requires two associated nuisance estimation routines: one for the forward pass and another for the backward pass. Algorithm~\ref{alg:nuisanceRoutines} details the inputs and outputs of these routines for a generic primitive $\theta : \mathcal{M}\times\mathcal{U}\rightarrow\mathcal{W}$. Like Algorithm~\ref{alg:estimatedBackprop}, the backward routine makes a compatibility requirement: the estimates of $\theta(P,u)$ and the first entry of $\dot{\theta}_{P,u}^*(w)$ must be compatible with a common $\widehat{P}\in\mathcal{M}$. There is no need to explicitly exhibit $\widehat{P}$, which in many cases will make it possible to only estimate certain summaries of $P$, such as regression functions or conditional densities. The forward and backward routines may also take hyperparameters as inputs, such as specifications of statistical learning algorithms to use to fit nuisances or folds to be used during cross-fitting.

\begin{algorithm}[tb]
   \caption{Requirements of nuisance estimation routines for a generic $\theta : \mathcal{M}\times\mathcal{U}\rightarrow\mathcal{W}$}
   \label{alg:nuisanceRoutines}
   \linespread{1.05}\selectfont
\begin{algorithmic}[1]
    \Statex $\triangleright$ {\color{CBblue}\textbf{forward routine}}
    \Statex \textbf{input} data $Z_{\mathscr{I}}:=\{Z_i : i\in\mathscr{I}\subseteq [n]\}$ and ${\color{CBteal}u}\in\mathcal{U}$
    \Statex \textbf{output} an estimate ${\color{CBblue}\widehat{h}}$ of $\theta(P,{\color{CBteal}u})$
    \algrule
    \Statex $\triangleright$ {\color{CBred}\textbf{backward routine}}
    \Statex \textbf{input} same $Z_{\mathscr{I}}$ and {\color{CBteal}$u$} as forward routine, output {\color{CBblue}$\widehat{h}$} from forward routine, and ${\color{CBorange}w}\in\mathcal{W}$
    \Statex \textbf{output} an estimate {\color{CBred}$\widehat{\vartheta}$} of $\dot{\theta}_{P,{\color{CBteal}u}}^*({\color{CBorange}w})$ whose first entry, {\color{CBred}$\widehat{\vartheta}_0$}, is compatible with {\color{CBblue}$\widehat{h}$}, {\color{CBteal}$u$}
    \Statex \hspace{\algorithmicindent} $\triangleright$ definition of compatibility: $\exists\, {\color{slateGrey}\widehat{P}}\in\mathcal{M}$ s.t. ${\color{CBblue}\widehat{h}}=\theta({\color{slateGrey}\widehat{P}},{\color{CBteal}u})$ and 
    $\int {\color{CBred}\widehat{\vartheta}_0}(z)\,{\color{slateGrey}\widehat{P}}(dz)=0$
\end{algorithmic}
\end{algorithm}

Algorithm~\ref{alg:estimatedBackprop} can be executed using these nuisance estimation routines. At step $j$ of the forward pass, the forward routine for $\theta_j$ is called with inputs $Z_{\mathscr{I}}$ and $\widehat{h}_{\mathrm{pa}(j)}$, yielding $\widehat{h}_j$. At step $j$ of the backward pass, the backward routine is called with inputs $Z_{\mathscr{I}}$, $\widehat{h}_{\mathrm{pa}(j)}$, $\widehat{h}_j$, and $\widehat{f}_j$, yielding $\widehat{\vartheta}_j$.

Forward and backward routines for the primitives in Table~\ref{tab:primitives} are given in Appendix~\ref{app:nuisanceEstimation}. Those routines allow needed regression, density functions, and density ratios to be estimated using flexible statistical learning tools, such as gradient-boosted trees \citep{ke2017lightgbm}. Riesz losses \citep{chernozhukov2022automatic} are used when estimating the adjoint of bounded affine maps whose ambient Hilbert spaces depend on $P$.

\section{Simulation study}\label{sec:simulation}

A simulation study was conducted to verify whether a handful of lines of \texttt{pydimple} code truly suffice to construct performant one-step estimators for the examples in Section~\ref{sec:illustrationPackage}. Such a finding would be important given that existing implementations of one-step estimators in these problems can consist of hundreds of lines of code \citep[e.g.,][]{luedtke2017sequential,williamson2021nonparametric}. For a given Monte Carlo replicate of our study, a dataset \texttt{dat} was generated as an iid sample of $n\in \{250,1000,4000,16000\}$ draws from a distribution $P$, and then the \texttt{pydimple} code from Section~\ref{sec:illustrationPackage} was executed to perform estimation. A total of 1000 Monte Carlo replicates were performed for each scenario. Code for this study is available at \url{http://github.com/alexluedtke12/pydimple}.

For the expected density parameter, we followed \cite{carone2018toward} in letting $P=\mathrm{Beta}(3,5)$. For the nonparametric $R^2$ parameter, we followed \cite{williamson2021nonparametric} in letting $P$ be the distribution of $(X,Y)$ with $X=(X_1,X_2)$ a pair of independent $\mathrm{Unif}[-1,1]$ random variables and $Y\mymid X\sim N(25X_1^2/9,1)$. For the longitudinal G-formula example, we used the data-generating process from Simulation 1 of \cite{luedtke2017sequential}.

{
\renewcommand{\arraystretch}{1.225} 
\begin{table}[tb]
\centering
\caption{Performance of \texttt{pydimple} package in the examples from Section~\ref{sec:illustrationPackage}. Assessed via coverage of 95\% confidence intervals, relative width = $n^{1/2}\times(\textnormal{mean width of CI})/(2\times 1.96\times \textnormal{standard deviation of EIF})$, relative variance = $n\times(\textnormal{variance of estimator})/(\textnormal{variance of EIF})$, and squared bias divided by the mean squared error. Asymptotically, an efficient estimator would have these four quantities equal to the values in the \textit{Target Value} row.}
\resizebox{\textwidth}{!}{%
     \begin{tabular}{llcccc}
         & $n$ & Coverage & Rel. Width & Rel. Variance & Bias$^2$/MSE \\[-.25em]
        \textit{Target Value} & & $\ge\,$\textit{95\%}\hphantom{$\,\ge$} & $\le\,$\textit{1.00}\hphantom{$\,\le$} & $\le\,$\textit{1.00}\hphantom{$\,\le$} & $=\,$\textit{0.00}\hphantom{$\,=$}  \\\hline\hline
        Expected Density & 250 & 91\% & 0.94 & 1.13 & 0.04 \\
        & 1000 & 92\% & 0.96 & 0.95  & 0.07 \\
         & 4000 & 93\% & 0.97 & 1.06  & 0.02 \\
         & 16000 & 95\% & 0.98 & 0.98  & 0.03 \\\hline
         Nonparametric $R^2$ & 250 & 87\% & 0.89 & 1.12 &  0.10 \\
          & 1000 & 92\% & 0.94 & 1.05 & 0.05 \\
         & 4000 & 94\% & 0.97 & 1.03  & 0.00 \\
         & 16000 & 93\% & 0.98 & 1.10  & 0.01 \\\hline
         Longitudinal G-Formula & 250 & 79\% & 1.09 & 4.44 & 0.12 \\
         & 1000 & 93\% & 1.28 & 2.01 & 0.02 \\
         & 4000 & 94\% & 1.25 & 2.10  & 0.01 \\
         & 16000\hspace{.25em} & 94\% & 1.06 & 1.26  & 0.00 
    \end{tabular}
}
    \label{tab:simResults}
\end{table}
}

Table~\ref{tab:simResults} reports the results. The coverage of 95\% confidence intervals returned by \texttt{pydimple} approximately attains the nominal value for all scenarios with $n\ge 1000$. Coverage tends to be below nominal when $n=250$, taking a minimal value of 79\%. The widths of the confidence intervals are close to the asymptotically optimal width of $2\times 1.96\times \sigma/ n^{1/2}$, with $\sigma$ the standard deviation of the EIF. For each expected density and nonparametric $R^2$ scenario, the variance of the estimators always approximately equals the efficient variance, $\sigma^2/n$. In the longitudinal G-formula scenario, this variance is much larger than the efficient variance at smaller sample sizes but the relative difference diminishes as $n$ grows. The inflated variance appears to be partially driven by outliers arising from inverse propensity weighting in the EIF. Bias is minimal in all scenarios. Together these results support our theoretical findings that dimple can produce asymptotically efficient estimators under reasonable conditions (Theorem~\ref{thm:efficient}).

On the author's 2021 MacBook Pro, the average time required to evaluate a single Monte Carlo replicate for Examples~\ref{ex:expectedDensity}, \ref{ex:r2}, and \ref{ex:longitudinalG} was about 3, 20, and 110 seconds respectively when $n=1000$.

\section{Discussion}

Our results suggest researchers should refocus their efforts when constructing an efficient estimator of a new parameter. Traditionally, this process involves labor-intensive calculations to derive the EIF. Our approach avoids this task for parameters expressible as compositions of known primitives. Yet, we recognize that this method will not apply to all new parameters, particularly those that cannot be represented in this form. In such cases, we propose a shift in focus: rather than immediately studying the differentiability of the parameter itself, researchers should aim to express it as a composition involving both known and novel primitives. By establishing the differentiability of each new primitive, they can deduce the differentiability of the parameter itself. This modular approach can both simplify the immediate analysis and the study of other parameters in future research. An example of a primitive that would be useful for ongoing research---but is missing from Table~\ref{tab:primitives}---is one that returns a solution to an infinite-dimensional inverse problem \citep{bennett2023source,cui2024semiparametric,smucler2025asymptotic}.

As noted in \cite{carone2018toward}, any algorithmic approach to statistical inference inevitably has potential for misuse. Per Theorems~\ref{thm:backpropWorks} and \ref{thm:efficient}, the validity of our method depends on differentiability and nuisance estimation conditions. While an expert can interpret and verify these conditions, non-specialists may employ a software implementation without doing so. This may lead to an overconfidence in the inferential conclusions reached, especially in high-dimensional problems or other settings where these conditions are unlikely to be met. To mitigate this risk, future software could provide a list of conditions sufficient for ensuring the differentiability of each primitive and the rates of convergence required on nuisance estimators for efficiency. Though Appendix~\ref{app:primitives} already gives sufficient conditions for the differentiability of the primitives in Table~\ref{tab:primitives}, work is still needed to figure out how to have software return this information in an accessible manner.

The composition of primitives used to express the parameter $\psi$ will generally be nonunique. The same holds in traditional automatic differentiation settings, where, for example, the function $f(x)=2x^2+1$ can be expressed as $g\circ h$ with $g(x)=2x+1$ and $h(x)=x^2$, $g(x)=\cosh(x)$ and $h(x)=2\sinh^{-1}(x)$, or $g(x)=2x^{1/2}+1$ and $h(x)=x^4$. In traditional settings, choosing different compositions can impact the conditions under which the chain rule can be applied---for example, $g(x)=2x^{1/2}+1$ is not differentiable at $0$, even though $f(x)=2x^2+1$ is. As indicated by the differentiability conditions in Theorem~\ref{thm:backpropWorks}, the applicability of the chain rule also depends on the choice of representation for $\psi$ in our statistical setting. Beyond that, this choice can also impact the nuisances that need to be estimated. For example, suppose $\psi(P)=E_P[Z]^2$ is expressed as $\langle z\mapsto z, dP/d\lambda\rangle_{L^2(\lambda)}^2$ using inner product and density primitives. Algorithm~\ref{alg:estimatedBackprop} would estimate the nuisance $dP/d\lambda$ for this representation and the efficiency of the one-step estimator would rely on an $n^{-1/4}$ rate condition. Estimation of this nuisance, and the associated rate condition, would be avoided if $\psi$ were instead expressed using a marginal mean primitive and square function. Algorithmically finding alternative representations of a parameter that weaken nuisance estimation requirements is an important area for future work.

There are also other interesting avenues for future work. First, though our algorithms apply to arbitrary models, the particular nuisance estimation routines given in Appendix~\ref{app:nuisanceEstimation} are designed for nonparametric settings. Future work should seek to identify semiparametric model restrictions that allow for modular nuisance estimation. Second, while we focused on algorithmically constructing one-step estimators, it would be worth considering whether this can be done for other estimation strategies, such as targeted minimum loss-based estimation \citep{van2006targeted,van2016one}. Third, it would be interesting to investigate whether automatic differentiation can be used to obtain higher-order influence functions \citep{robins2008higher}.

Dimple provides a means to rapidly translate new developments in debiased estimation from theory to practice. While our proof-of-concept software uses a cross-fitted one-step estimator, future implementations should incorporate recent proposals like iteratively debiasing nuisance function estimators for added robustness \citep{rotnitzky2017multiply,luedtke2017sequential}, employing separate data subsamples to estimate different nuisance functions \citep{newey2018cross}, and enforcing known shape constraints on the parameter of interest \citep{westling2020unified,ham2024doubly}. This mirrors the approach taken by other differentiable and probabilistic programming frameworks, where new algorithms, architectures, and regularization techniques are routinely integrated into existing software \citep{jax2018github,NEURIPS2019_9015,stan2023}. This practice ensures that state-of-the-art methods are widely accessible, a model we strive to emulate.

\section*{Acknowledgements}
This work was supported by the National Science Foundation and National Institutes of Health under award numbers DMS-2210216 and DP2-LM013340. The author is grateful to Jonathan Kernes for kindly letting him use code from his \href{https://towardsdatascience.com/build-your-own-automatic-differentiation-program-6ecd585eec2a}{Medium tutorial} in \texttt{pydimple} \citep{kernes2021}. He is also grateful to Saksham Jain for taking the time to test the package on a sunny Seattle afternoon.

{\singlespacing\small
\bibliography{References}
}

\appendix

\setcounter{equation}{0}
\renewcommand{\theequation}{S\arabic{equation}}
\setcounter{theorem}{0}
\setcounter{figure}{0}
\setcounter{table}{0}
\setcounter{lemma}{0}
\setcounter{corollary}{0}
\renewcommand{\thetheorem}{S\arabic{theorem}}
\renewcommand{\thecorollary}{S\arabic{corollary}}
\renewcommand{\thelemma}{S\arabic{lemma}}
\renewcommand{\thefigure}{S\arabic{figure}}
\renewcommand{\thetable}{S\arabic{table}}
\renewcommand{\thealgorithm}{S\arabic{algorithm}}

\section*{\LARGE Appendices}

\onehalfspacing

\titleformat{\section}{\normalfont\Large\bfseries}{\thesection}{1em}{}  
\titleformat{\subsection}{\normalfont\large\bfseries}{\thesubsection}{1em}{} 
\titleformat{\subsubsection}[runin]
  {\normalfont\normalsize\bfseries}{\thesubsubsection}{1em}{}[.]


\DoToC

\section{Total pathwise differentiability: supplemental theoretical results and proofs}
\subsection{Preliminary lemmas}
We begin by defining some notation and conventions used throughout the appendix. If $\mathcal{A}_1,\mathcal{A}_2,\ldots,\mathcal{A}_j$ are subsets of Hilbert spaces $\mathcal{B}_1,\mathcal{B}_2,\ldots,\mathcal{B}_j$, then $\prod_{i=1}^j \mathcal{A}_i$ is treated as a subset of the direct sum of those Hilbert spaces, $\bigoplus_{i=1}^j\mathcal{B}_i$. For an element $(b_1,b_2,\ldots,b_j)\in \bigoplus_{i=1}^j\mathcal{B}_i$ and nonempty $\mathcal{S}\subseteq [j]$, we let $b_{\mathcal{S}}:=(b_i : i\in\mathcal{S})$. We similarly let $\mathcal{B}_{\mathcal{S}}$ be the direct sum Hilbert space $\bigoplus_{i\in\mathcal{S}} \mathcal{B}_i$, and for $\prod_{i=1}^j \mathcal{A}_i\subseteq\bigoplus_{i=1}^j\mathcal{B}_i$, we define $\mathcal{A}_{\mathcal{S}}:=\prod_{i\in\mathcal{S}}\mathcal{A}_i$. If $\mathcal{S}=\emptyset$, then $\mathcal{A}_{\mathcal{S}}=\mathcal{B}_{\mathcal{S}}=\{0\}$ and $b_{\mathcal{S}}=0$. We denote the inner product and norm of a generic Hilbert space $\mathcal{B}$ by $\langle\cdot,\cdot\rangle_{\mathcal{B}}$ and $\|\cdot\|_\mathcal{B}$, respectively.

The following lemma shows that there are two equivalent ways of characterizing smooth paths $\{v_{[j],\epsilon} : \epsilon\}$ through some $v_{[j]}\in\mathcal{V}_{[j]}$. The first directly defines these smooth paths as elements of $\mathscr{P}(v_{[j]},\mathcal{V}_{[j]},t_{[j]})$ for some $t_{[j]}$. The second defines smooth paths $\{v_{i,\epsilon} : \epsilon\}\in\mathscr{P}(v_i,\mathcal{V}_i,t_i)$ through the coordinate projections $v_i\in\mathcal{V}_i$, $i\in [j]$, and then concatenates them so that $v_{[j],\epsilon}:=(v_{1,\epsilon},v_{2,\epsilon},\ldots,v_{j,\epsilon})$. 
\begin{lemma}[Defining smooth paths through $v_{[j]}$ via smooth paths through its coordinate projections]\label{lem:tanSetOfProduct}
    Fix $j\in [k]$ and $v_{[j]}\in\mathcal{V}_{[j]}$. Let $\check{\mathcal{V}}_{[j],v_{[j]}}$ denote the tangent set of $\mathcal{V}_{[j]}$ at $v_{[j]}$ and $\check{\mathcal{V}}_{i,v_i}$ the tangent set of $\mathcal{V}_i$ at $v_i$, $i\in [j]$; let $\dot{\mathcal{V}}_{[j],v_{[j]}}$ and $\dot{\mathcal{V}}_{i,v_i}$ denote the corresponding tangent spaces. Both of the following hold:
    \begin{enumerate}[label=(\roman*),ref=(\roman*)]
        \item\label{it:prodImpliesProj} for any $t_{[j]}$ in $\check{\mathcal{V}}_{[j],v_{[j]}}$ and $\{v_{[j],\epsilon}=(v_{i,\epsilon})_{i=1}^j : \epsilon\in[0,1]\}\in\mathscr{P}(v_{[j]},\mathcal{V}_{[j]},t_{[j]})$, it holds that $\{v_{i,\epsilon} : \epsilon\in[0,1]\}\in \mathscr{P}(v_i,\mathcal{V}_i,t_i)$ for all $i\in [j]$;
        \item\label{it:projImpliesProd} for any $t_{[j]}\in \prod_{i=1}^j \check{\mathcal{V}}_{i,v_i}$ and $\{v_{i,\epsilon} : \epsilon\in[0,1]\}\in \mathscr{P}(v_i,\mathcal{V}_i,t_i)$, $i\in [j]$, it holds that $\{v_{[j],\epsilon}=(v_{i,\epsilon})_{i=1}^j : \epsilon\in[0,1]\}\in\mathscr{P}(v_{[j]},\mathcal{V}_{[j]},t_{[j]})$.
    \end{enumerate}
    Consequently, $\check{\mathcal{V}}_{[j],v_{[j]}}=\prod_{i=1}^j \check{\mathcal{V}}_{i,v_i}$ and $\dot{\mathcal{V}}_{[j],v_{[j]}}=\bigoplus_{i=1}^j \dot{\mathcal{V}}_{i,v_i}$.
\end{lemma}
\begin{proof}[Proof of Lemma \ref{lem:tanSetOfProduct}]
    Observe that \ref{it:prodImpliesProj} implies that $\check{\mathcal{V}}_{[j],v_{[j]}}\subseteq \prod_{i=1}^j \check{\mathcal{V}}_{i,v_i}$ and \ref{it:projImpliesProd} implies that $\check{\mathcal{V}}_{[j],v_{[j]}}\supseteq \prod_{i=1}^j \check{\mathcal{V}}_{i,v_i}$, and so together \ref{it:prodImpliesProj} and \ref{it:projImpliesProd} imply that $\check{\mathcal{V}}_{[j],v_{[j]}}=\prod_{i=1}^j \check{\mathcal{V}}_{i,v_i}$. It is straightforward to verify that the $\mathcal{W}_{[j]}$-closure of the linear span of $\prod_{i=1}^j \check{\mathcal{V}}_{i,v_i}$ is equal to $\bigoplus_{i=1}^j \dot{\mathcal{V}}_{i,v_i}$, and so the tangent set relation $\check{\mathcal{V}}_{[j],v_{[j]}}=\prod_{i=1}^j \check{\mathcal{V}}_{i,v_i}$ implies the tangent space relation $\dot{\mathcal{V}}_{[j],v_{[j]}}=\bigoplus_{i=1}^j \dot{\mathcal{V}}_{i,v_i}$. Hence, the proof will be complete if we can establish \ref{it:prodImpliesProj} and \ref{it:projImpliesProd}.

    Before proceeding, it will be useful to note that, for any $t_{[j]}\in\mathcal{W}_{[j]}$ and $(v_{i,\epsilon})_{i=1}^j\in\mathcal{V}_{[j]}$, $\epsilon\in [0,1]$, the definition of the norm on $\mathcal{W}_{[j]}:=\bigoplus_{i=1}^j \mathcal{W}_i$ makes it so that
    \begin{align}
    \|(v_{i,\epsilon})_{i=1}^j - (v_i)_{i=1}^j -\epsilon (t_i)_{i=1}^j\|_{\mathcal{W}_{[j]}}^2=\sum_{i=1}^j \|v_{i,\epsilon} - v_i -\epsilon t_i\|_{\mathcal{W}_i}^2. \label{eq:directSumNorm}
    \end{align}
    To establish \ref{it:prodImpliesProj}, we note that, for any $t_{[j]}\in \check{\mathcal{V}}_{[j],v_{[j]}}\subseteq\mathcal{W}_{[j]}$ and $\{(v_{i,\epsilon})_{i=1}^j : \epsilon\in [0,1]\}\in \mathscr{P}(v_{[j]},\mathcal{V}_{[j]},t_{[j]})$, the left-hand side above is $o(\epsilon^2)$, and so each squared norm on the right-hand side must be $o(\epsilon^2)$ as well. Consequently, $\{v_{i,\epsilon} : \epsilon\in [0,1]\}\in \mathscr{P}(v_i,\mathcal{V}_i,t_i)$ for each $i\in [j]$, and so \ref{it:prodImpliesProj} holds. To establish \ref{it:projImpliesProd}, we note that, for any $t_i\in\check{\mathcal{V}}_{i,v_i}$ and $\{v_{i,\epsilon} : \epsilon\in [0,1]\}\in \mathscr{P}(v_i,\mathcal{V}_i,t_i)$, $i\in [j]$, each squared norm on the right-hand side above is $o(\epsilon^2)$, and so the squared norm on the left-hand side must be $o(\epsilon^2)$ as well. Consequently, $\{(v_{i,\epsilon})_{i=1}^j : \epsilon\in[0,1]\}\in\mathscr{P}(v_{[j]},\mathcal{V}_{[j]},t_{[j]})$, and so \ref{it:projImpliesProd} holds. 
\end{proof}

In the following lemma and throughout we write $\mathrm{proj}_{\mathcal{A}}\{\,\cdot\,\mymid \mathcal{B}\}$ to denote the orthogonal projection in a Hilbert space $\mathcal{A}$ onto a closed subspace $\mathcal{B}$.

\begin{lemma}[Total pathwise differentiability of coordinate projections]\label{lem:coordProj}
	Fix $j\in [k]$ and let $\mathcal{J}\subseteq [j]$. The coordinate projection map $\Lambda_{\mathcal{J}} : \mathcal{M}\times\mathcal{V}_{[j]}\rightarrow \mathcal{V}_{\mathcal{J}}$, defined so that $\Lambda_{\mathcal{J}}(P,v_{[j]})=v_{\mathcal{J}}$, is totally pathwise differentiable and, for all $w_{\mathcal{J}}:=(w_i : i\in\mathcal{J})\in\mathcal{W}_{\mathcal{J}}$ and $(P,v_{[j]})\in\mathcal{M}\times\mathcal{V}_{[j]}$,
	\begin{align*}
		\dot{\Lambda}_{\mathcal{J},P,v_{[j]}}^*(w_\mathcal{J})_i&= \left(0,\left(\mathbbm{1}_{\mathcal{J}}(i)\cdot \mathrm{proj}_{\mathcal{W}_i}\{w_i\mymid \dot{\mathcal{V}}_{i,v_i}\}\right)_{i=1}^j\right),
	\end{align*}
	where $\mathbbm{1}_{\mathcal{J}}(i)\cdot \mathrm{proj}_{\mathcal{W}_i}\{w_i\mymid \dot{\mathcal{V}}_{i,v_i}\}$ is equal to the zero element of $\mathcal{W}_i$ if $i\not\in\mathcal{J}$ and is otherwise equal to $\mathrm{proj}_{\mathcal{W}_i}\{w_i\mymid \dot{\mathcal{V}}_{i,v_i}\}$.
\end{lemma}
\begin{proof}[Proof of Lemma~\ref{lem:coordProj}]
	We first derive the differential operator $\dot{\Lambda}_{\mathcal{J},P,v_{[j]}}$, and then we derive its adjoint. Fix $s\in \check{\mathcal{M}}_P$ and $t_{[j]}$ in the tangent set $\check{\mathcal{V}}_{[j],v_{[j]}}$ of $\mathcal{V}_{[j]}$ at $v_{[j]}$. Let $\{P_\epsilon : \epsilon\}\in \mathscr{P}(P,\mathcal{M},s)$ and $\{v_{[j],\epsilon} : \epsilon\}\in\mathscr{P}(v_{[j]},\mathcal{V}_{[j]},t_{[j]})$.  Observe that
	\begin{align*}
		\left\|\Lambda_{\mathcal{J}}(P_\epsilon,v_{[j],\epsilon}) - \Lambda_{\mathcal{J}}(P,v_{[j]}) - \epsilon t_{\mathcal{J}}\right\|_{\mathcal{V}_{\mathcal{J}}}^2&= \left\|v_{\mathcal{J},\epsilon}- v_{\mathcal{J}} - \epsilon t_{\mathcal{J}}\right\|_{\mathcal{V}_{\mathcal{J}}}^2 = \sum_{i\in\mathcal{J}} \left\|v_{i,\epsilon} - v_i - \epsilon t_{i}\right\|_{\mathcal{V}_{\mathcal{J}}}^2.
	\end{align*}
	By part~\ref{it:prodImpliesProj} of Lemma~\ref{lem:tanSetOfProduct}, $\{v_{i,\epsilon} : \epsilon\}\in\mathscr{P}(v_{i},\mathcal{V}_{i},t_{i})$ for each $i\in [j]$, and so each term in the sum on the right-hand side is $o(\epsilon^2)$. Recalling the left-hand side above, this shows that $\Lambda_{\mathcal{J}}$ is totally pathwise differentiable at $(P,v_{[j]})$ with (bounded and linear) differential operator $\dot{\Lambda}_{\mathcal{J},P,v_{[j]}} : \dot{\mathcal{M}}_P\times \dot{\mathcal{V}}_{[j],v_{[j]}}\rightarrow \mathcal{W}_{\mathcal{J}}$ defined so that $\dot{\Lambda}_{\mathcal{J},P,v_{[j]}}(s,t_{[j]})=t_{\mathcal{J}}$.
	
	We now derive the form of $\dot{\Lambda}_{\mathcal{J},P,v_{[j]}}^*$. To this end, fix $w_{\mathcal{J}}=(w_i : i\in\mathcal{J})\in \mathcal{W}_{\mathcal{J}}$. For any $(s,t_{[j]})\in \dot{\mathcal{M}}_P\times \dot{\mathcal{V}}_{[j],v_{[j]}}$,
	\begin{align*}
		\left\langle \dot{\Lambda}_{\mathcal{J},P,v_{[j]}}(s,t_{[j]}), w_{\mathcal{J}}\right\rangle_{\mathcal{W}_{\mathcal{J}}}&= \left\langle t_{\mathcal{J}}, w_{\mathcal{J}}\right\rangle_{\mathcal{W}_{\mathcal{J}}} = \sum_{i\in\mathcal{J}}\left\langle t_i, w_i\right\rangle_{\mathcal{W}_i} = \sum_{i\in\mathcal{J}}\left\langle t_i, \mathrm{proj}_{\mathcal{W}_i}\{w_i\mymid \dot{\mathcal{V}}_{i,v_i}\}\right\rangle_{\mathcal{W}_i} \\
		&= \sum_{i\in\mathcal{J}}\left\langle t_i, \mathrm{proj}_{\mathcal{W}_i}\{w_i\mymid \dot{\mathcal{V}}_{i,v_i}\}\right\rangle_{\dot{\mathcal{V}}_{i,v_i}} \\
		&= \left\langle (s,t_{[j]}),\left(0,\left(\mathbbm{1}_{\mathcal{J}}(i)\cdot \mathrm{proj}_{\mathcal{W}_i}\{w_i\mymid \dot{\mathcal{V}}_{i,v_i}\}\right)_{i=1}^j\right)\right\rangle_{\dot{\mathcal{M}}_P\oplus \left(\oplus_{i=1}^j \dot{\mathcal{V}}_{i,v_i}\right)},
	\end{align*}
	where the first equality uses the already-derived form of the differential operator, the second and fifth use the definition of inner products on direct sums of Hilbert spaces, and the third and fourth use properties of projections and the fact that $t_i\in \dot{\mathcal{V}}_{i,v_i}$. The proof concludes by noting that $\bigoplus_{i=1}^j \dot{\mathcal{V}}_{i,v_i}=\dot{\mathcal{V}}_{[j],v_{[j]}}$, by Lemma~\ref{lem:tanSetOfProduct}, and so the Hermitian adjoint $\dot{\Lambda}_{\mathcal{J},P,v_{[j]}}^*$ of $\dot{\Lambda}_{\mathcal{J},P,v_{[j]}} : \dot{\mathcal{M}}_P\times \dot{\mathcal{V}}_{[j],v_{[j]}}\rightarrow \mathcal{W}_{\mathcal{J}}$ takes the form claimed in the lemma statement.
\end{proof}

\subsection{Chain rule}

In this appendix, we let $\mathcal{Q}$, $\mathcal{T}$, and $\mathcal{W}$ be Hilbert spaces. We also let $\mathcal{R}\subseteq \mathcal{Q}$ and $\mathcal{U}\subseteq\mathcal{T}$. 
\begin{lemma}[Chain rule for total pathwise differentiability]\label{lem:chainRule}
    Let $\gamma : \mathcal{M}\times\mathcal{U}\rightarrow\mathcal{R}$ and $\eta : \mathcal{M}\times  \mathcal{R}\rightarrow\mathcal{W}$ be totally pathwise differentiable at $(P,u)$ and $(P,\gamma(P,u))$, respectively. Then, $\beta : \mathcal{M}\times\mathcal{U}\rightarrow\mathcal{W}$, defined so that $\beta(P',u')= \eta(P',\gamma(P',u'))$, is totally pathwise differentiable at $(P,u)$ with $\dot{\beta}_{P,u}(s,t)=\dot{\eta}_{P,\gamma(P,u)}\left(s,\dot{\gamma}_{P,u}(s,t)\right)$ and
    \begin{align*}
        \dot{\beta}_{P,u}^*(w)&= \left(\dot{\eta}_{P,\gamma(P,u)}^*(w)_0,0\right) + \dot{\gamma}_{P,u}^*\left(\dot{\eta}_{P,\gamma(P,u)}^*(w)_1\right),
    \end{align*}
    where $\dot{\eta}_{P,\gamma(P,u)}^*(w)_0\in\dot{\mathcal{M}}_P$ and $\dot{\eta}_{P,\gamma(P,u)}^*(w)_1\in \dot{\mathcal{U}}_u$ are defined so that
    $$\dot{\eta}_{P,\gamma(P,u)}^*(w)=\left(\dot{\eta}_{P,\gamma(P,u)}^*(w)_0,\dot{\eta}_{P,\gamma(P,u)}^*(w)_1\right).$$
\end{lemma}
\begin{proof}[Proof of Lemma~\ref{lem:chainRule}]
    We first show that $\beta$ is totally pathwise differentiable at $(P,u)$ with the specified differential operator. Fix $\{P_\epsilon : \epsilon\in [0,1]\}\in\mathscr{P}(P,\mathcal{M},s)$ and $\{u_\epsilon : \epsilon\in [0,1]\}\in \mathscr{P}(u,\mathcal{U},t)$. Let $r_\epsilon:= \gamma(P_\epsilon,u_\epsilon)$ and $r:=\gamma(P,u)$. As $\gamma$ is totally pathwise differentiable at $(P,u)$, $\{r_\epsilon : \epsilon\in [0,1]\}$ belongs to $\mathscr{P}(r,\mathcal{R},\dot{\gamma}_{P,u}(s,t))$. As $\eta$ is totally pathwise differentiable, this implies that, as $\epsilon\rightarrow 0$,
    \begin{align*}
        \epsilon^{-1}\left[\eta(P_\epsilon,r_\epsilon) - \eta(P,r)\right]\rightarrow \dot{\eta}_{P,r}\left(s,\dot{\gamma}_{P,u}(s,t)\right)=: \dot{\beta}_{P,u}(s,t) .
    \end{align*}
    The boundedness and linearity of $\dot{\eta}_{P,r}$ and $\dot{\gamma}_{P,u}$ imply that $\dot{\beta}_{P,u}$ is bounded and linear, and so $\beta$ is totally pathwise differentiable at $(P,u)$ with differential operator $\dot{\beta}_{P,u}$.
    
    As for the adjoint $\dot{\beta}_{P,u}^*$, observe that, for any $w\in\mathcal{W}$ and $(s,t)\in\dot{\mathcal{M}}_P\oplus \dot{\mathcal{U}}_u$,
    \begin{align*}
    \left\langle \dot{\beta}_{P,u}(s,t),w\right\rangle_{\mathcal{W}} &= \left\langle \dot{\eta}_{P,\gamma(P,u)}\left(s,\dot{\gamma}_{P,u}(s,t)\right),w\right\rangle_{\mathcal{W}} \\
    &= \left\langle \left(s,\dot{\gamma}_{P,u}(s,t)\right),\dot{\eta}_{P,\gamma(P,u)}^*(w)\right\rangle_{\dot{\mathcal{M}}_P\oplus \mathcal{Q}} \\
    &= \left\langle s,\dot{\eta}_{P,\gamma(P,u)}^*(w)_0\right\rangle_{\dot{\mathcal{M}}_P} + \left\langle \dot{\gamma}_{P,u}(s,t),\dot{\eta}_{P,\gamma(P,u)}^*(w)_1\right\rangle_{\mathcal{Q}} \\
    &= \left\langle s,\dot{\eta}_{P,\gamma(P,u)}^*(w)_0\right\rangle_{\dot{\mathcal{M}}_P} + \left\langle (s,t),\dot{\gamma}_{P,u}^*\left(\dot{\eta}_{P,\gamma(P,u)}^*(w)_1\right)\right\rangle_{\dot{\mathcal{M}}_P\oplus\dot{\mathcal{U}}_u} \\
    &= \left\langle (s,t),\left(\dot{\eta}_{P,\gamma(P,u)}^*(w)_0,0\right)\right\rangle_{\dot{\mathcal{M}}_P\oplus \dot{\mathcal{U}}_u} + \left\langle (s,t),\dot{\gamma}_{P,u}^*\left(\dot{\eta}_{P,\gamma(P,u)}^*(w)_1\right)\right\rangle_{\dot{\mathcal{M}}_P\oplus\dot{\mathcal{U}}_u} \\
    &= \left\langle \left(s,t\right),\left(\dot{\eta}_{P,\gamma(P,u)}^*(w)_0,0\right) + \dot{\gamma}_{P,u}^*\left(\dot{\eta}_{P,\gamma(P,u)}^*(w)_1\right)\right\rangle_{\dot{\mathcal{M}}_P\oplus\dot{\mathcal{U}}_u},
    \end{align*}
    where: the first equality used the already-derived form of $\dot{\beta}_{P,u}$; the second and fourth used the definition of an adjoint; the third and fifth used the definition of inner products in direct sums of Hilbert spaces; and the final equality used the bilinearity of inner products. The right-hand side writes as $\left\langle (s,t),\beta_{P,u}^*(w)\right\rangle_{\dot{\mathcal{M}}_P\oplus \dot{\mathcal{U}}_u}$, with $\beta_{P,u}^*$ as defined in the lemma statement. Hence, $\beta_{P,u}^*$ is the Hermitian adjoint of $\dot{\beta}_{P,u} : \dot{\mathcal{M}}_P\oplus \dot{\mathcal{U}}_u\rightarrow \mathcal{W}$.
\end{proof}

The following consequence of the chain rule will prove useful when studying the backpropagation algorithm. In this lemma, we let $\Omega : \dot{\mathcal{M}}_P\times (\mathcal{T}\times\mathcal{Q})\rightarrow \dot{\mathcal{M}}_P\times \mathcal{U}$ and $\Lambda : \dot{\mathcal{M}}_P\times (\mathcal{T}\times\mathcal{Q})\rightarrow \mathcal{R}$ be the coordinate projections defined so that $\Omega(s,(u,r))=(s,u)$ and $\Lambda(s,(u,r)) = r$.
\begin{lemma}[Consequence of the chain rule]\label{lem:chainRuleConvenient}
    Let $\gamma : \mathcal{M}\times\mathcal{U}\rightarrow\mathcal{R}$ and $\eta : \mathcal{M}\times (\mathcal{U}\times \mathcal{R})\rightarrow\mathcal{W}$ be totally pathwise differentiable at $(P,u)$ and $(P,(u,\gamma(P,u)))$, respectively. Then, $\beta : \mathcal{M}\times\mathcal{U}\rightarrow\mathcal{W}$, defined so that $\beta(P',u')= \eta(P',(u',\gamma(P',u')))$, is totally pathwise differentiable at $(P,u)$ with $\dot{\beta}_{P,u}(s,t)=\dot{\eta}_{P,(u,\gamma(P,u))}\left(s,(t,\dot{\gamma}_{P,u}(s,t))\right)$ and
    \begin{align*}
        \dot{\beta}_{P,u}^*(w)&= \Omega\circ \dot{\eta}_{P,(u,\gamma(P,u))}^*(w) + \dot{\gamma}_{P,u}^*\circ \Lambda\circ \dot{\eta}_{P,(u,\gamma(P,u))}^*(w).
    \end{align*}
\end{lemma}
\begin{proof}[Proof of Lemma~\ref{lem:chainRuleConvenient}]
    Let $\underline{\mathcal{R}}:=\mathcal{U}\times\mathcal{R}$ be a subset of $\underline{\mathcal{Q}}:= \mathcal{T}\oplus\mathcal{Q}$. 
    Define $\underline{\gamma} : \mathcal{M}\times\mathcal{U}\rightarrow \underline{\mathcal{R}}$ so that $\underline{\gamma}(P',u')=(u',\gamma(P',u'))$, and note that $\beta(P',u')= (P',\underline{\gamma}(P',u'))$. We will use Lemma~\ref{lem:chainRule} to study $\beta$. Before doing this, we must establish that $\underline{\gamma}$ is totally pathwise differentiable at $(P,u)$.
    
    For any $s\in\check{\mathcal{M}}_P$, $\{P_\epsilon : \epsilon\in [0,1]\}\in\mathscr{P}(P,\mathcal{M},s)$, $t\in\check{\mathcal{U}}_u$, and $\{u_\epsilon : \epsilon\in [0,1]\}\in \mathscr{P}(u,\mathcal{U},t)$,
    \begin{align*}
        \|\underline{\gamma}(P_\epsilon,u_\epsilon)-\underline{\gamma}(P,u) - \epsilon (t,\dot{\gamma}_{P,u}(s,t))\|_{\mathcal{T}\oplus\mathcal{Q}}^2&= \|u_\epsilon-u-\epsilon t\|_{\mathcal{T}}^2 + \|\gamma(P_\epsilon,u_\epsilon)-\gamma(P,u)-\epsilon \dot{\gamma}_{P,u}(s,t)\|_{\mathcal{Q}}^2.
    \end{align*}
    The first term on the right is $o(\epsilon^2)$ since $\{u_\epsilon : \epsilon\in [0,1]\}\in \mathscr{P}(u,\mathcal{U},t)$, and the second is $o(\epsilon^2)$ by the total pathwise differentiability of $\gamma$. The boundedness and linearity of $\dot{\gamma}_{P,u}$ implies that $(s,t)\mapsto (t,\dot{\gamma}_{P,u}(s,t))$ is bounded and linear as well. Hence, $\underline{\gamma}$ is totally pathwise differentiable at $(P,u)$ with $\dot{\underline{\gamma}}_{P,u} : (s,t)\mapsto (t,\dot{\gamma}_{P,u}(s,t))$.  By the definition of inner products in direct sum Hilbert spaces, it also readily follows that the Hermitian adjoint of $\dot{\underline{\gamma}}_{P,u}$ is the map $\dot{\underline{\gamma}}_{P,u}^* : \mathcal{T}\oplus\mathcal{Q}\rightarrow \dot{\mathcal{M}}_P\oplus\dot{\mathcal{U}}_u$ defined so that $\dot{\underline{\gamma}}_{P,u}^*(t,q)= (0,\mathrm{proj}_{\mathcal{T}}\{t\mymid \dot{\mathcal{U}}_u\}) + \dot{\gamma}_{P,u}^*(q)$, where $0$ denotes the zero element of $\dot{\mathcal{M}}_P$.

    We are now in a position to apply Lemma~\ref{lem:chainRule}. That lemma shows that $\beta$ is totally pathwise differentiable at $(P,u)$ with differential operator $\dot{\beta}_{P,u}(s,t)=\dot{\eta}_{P,\gamma(P,u)}(s,\dot{\underline{\gamma}}_{P,u}(s,t))$ and adjoint
    \begin{align*}
        \dot{\beta}_{P,u}^*(w)&= \left(\dot{\eta}_{P,\gamma(P,u)}^*(w)_0,0\right) + \dot{\underline{\gamma}}_{P,u}^*\left(\dot{\eta}_{P,\gamma(P,u)}^*(w)_1\right).
    \end{align*}
    Since $\dot{\eta}_{P,(u,\gamma(P,u))}^*(w)\in \dot{\mathcal{M}}_P\oplus (\dot{\mathcal{U}}_u\oplus \dot{\mathcal{R}}_{\gamma(P,u)})$ for any $w\in\mathcal{W}$, the projection onto $\dot{\mathcal{U}}_u$ that appears in the adjoint $\dot{\underline{\gamma}}_{P,u}^*$ evaluates as the identity operator when $\dot{\underline{\gamma}}_{P,u}^*$ is applied to $\dot{\eta}_{P,\gamma(P,u)}^*(w)_1$, and so
    \begin{align*}
        \dot{\beta}_{P,u}^*(w)&= \Omega\circ \dot{\eta}_{P,(u,\gamma(P,u))}^*(w) + \dot{\gamma}_{P,u}^*\circ \Lambda\circ \dot{\eta}_{P,(u,\gamma(P,u))}^*(w).
    \end{align*}
\end{proof}

\subsection{Proof of Theorem~\ref{thm:backpropWorks}}\label{app:backpropWorks}
\begin{proof}[Proof of Theorem~\ref{thm:backpropWorks}]
    This proof concerns maps $\eta_j : \mathcal{M}\times\mathcal{V}_{[j]}\rightarrow \mathcal{W}_\psi$, $j\in \{0\}\cup [k]$. These maps are defined so that, for all $(P',v_{[j]})\in\mathcal{M}\times\mathcal{V}_{[j]}$, $\eta_j(P',v_{[j]})$ is the output of a modification of Algorithm~\ref{alg:parameter} that replaces each line $i\in [j]$ by the assignment $h_i=v_i$ and leaves each line $i\in \{j+1,j+2,\ldots,k\}$ unchanged, so that $h_i=\theta_i(P',h_{\mathrm{pa}(i)})$. Concretely, $\eta_k(P',v_{[k]})=v_k$ and, for $j=k-1,k-2,\ldots,0$,
    \begin{align}
       \eta_j(P',v_{[j]}):=\eta_{j+1}\left(P',(v_{[j]},\theta_{j+1}(P',v_{\mathrm{pa}(j+1)}))\right).\label{eq:etaRecursion}
    \end{align}
    In the remainder of this proof, we leverage the chain rule and the recursive relationship between $\eta_j$ and $\eta_{j+1}$ to show $\eta_0$ is totally pathwise differentiable at $(P,0)$ with the first entry of $\dot{\eta}_{0,P,0}^*(f_k)$ equal to $f_0$. Because $\psi(\cdot)=\eta_0(\,\cdot\,,0)$, Lemma~\ref{lem:totalImpliesPartial} will then yield that $\dot{\psi}_P^*(f_k)=f_0$.
    
    To proceed, it will be helpful to have notation to denote the values $f_0,f_1,\ldots,f_j$ take before line~\ref{ln:augment} of Algorithm~\ref{alg:backprop} is evaluated for a given value of $j$ in the for loop. To this end, we add the following variable definition between lines~\ref{ln:for} and \ref{ln:augment} of the algorithm:
    \begin{center}
    \begin{algorithmic}[1]
      \setalglineno{2}
    \For {$j=k,k-1,\ldots,1$}
        \Statex \hspace{1.15em} $(f_0^{(j)},f_{[j]}^{(j)})=(f_0,f_{[j]})$\hfill $\triangleright$ value $(f_0,f_{[j]})$ takes before executing the line below
        \State $(f_0,f_{\mathrm{pa}(j)}) \mathrel{+{=}} \dot{\theta}_{j,P,h_{\mathrm{pa}(j)}}^*(f_j)$
    \EndFor
    \end{algorithmic}
    \end{center}
    Since, $(f_0^{(j)},f_{[j]}^{(j)})$, $j\in [k]$, are not used by the algorithm once defined, adding this line of code will not change the algorithm's output. 
    We also let $f_0^{(0)}$ denote the value that $f_0$ takes after executing the entirety of the for loop on lines~\ref{ln:for} and \ref{ln:augment} of Algorithm~\ref{alg:backprop}.

    We use induction to establish the following hypothesis holds for all $j\in \{k,k-1,\ldots,0\}$. Throughout our induction argument, we let $h_{[j]}$ be the quantity defined in Algorithm~\ref{alg:parameter} when that algorithm is called with input $P$.\\[.5em] 
    \textbf{Inductive hypothesis $j$, IH($j$):} $\eta_j$ is totally pathwise differentiable at $(P,h_{[j]})$ with
    \begin{align}
        \left(f_0^{(j)},\mathrm{proj}_{\mathcal{W}_{[j]}}\left\{f_{[j]}^{(j)}\,\middle|\, \dot{\mathcal{V}}_{[j],h_{[j]}}\right\}\right)=\dot{\eta}_{j,P,h_{[j]}}^*(f_k). \label{eq:IH}
    \end{align}
    In the case where $j=0$, we let $f_{[j]}^{(j)}$ denote the zero element in a trivial vector space $\{0\}$. Noting that $f_0^{(0)}$ is the value of $f_0$ returned by Algorithm~\ref{alg:backprop}, IH($0$) will imply that Algorithm~\ref{alg:backprop} returns the first entry of $\dot{\eta}_{j,P,h_{[j]}}^*(f_k)$, which, as noted earlier, is equal to $\dot{\psi}_P^*(f_k)$. Hence, if we can establish IH($0$), as we will now do by induction, then we will have completed our proof.\\[.5em]
    \textbf{Base case:} $j=k$. Since $\eta_k$ is the coordinate projection $(P',v_{[k]})\mapsto v_k$, Lemma~\ref{lem:coordProj} shows $\eta_k$ is totally pathwise differentiable at $(P,h_{[k]})$ with $\dot{\eta}_{k,P,h_{[k]}}^* : w_k \mapsto (0,(0,0,\ldots,0,\mathrm{proj}_{\mathcal{W}_k}\{w_k\mymid \dot{\mathcal{V}}_{k,h_k}\}))$. Since $f_0$ and $f_{[k-1]}$ are initialized to the be the zero elements in $L^2(P)$ and $\mathcal{W}_{[k-1]}$, respectively, $f_0^{(k)}$ and $f_{[k-1]}^{(k)}$ are also equal to these zero elements. Also, $f_{k}^{(k)}=f_k$. Hence, Lemma~\ref{lem:tanSetOfProduct} yields that $\mathrm{proj}_{\mathcal{W}_{[k]}}\{f_{[k]}^{(k)}\mymid \dot{\mathcal{V}}_{[k],h_{[k]}}\}=(0,0,\ldots,0,\mathrm{proj}_{\mathcal{W}_k}\{f_k^{(k)}\mymid \dot{\mathcal{V}}_{k,h_k}\})$. Putting all these facts together shows that \eqref{eq:IH} holds when $j=k$, that is, IH($k$) holds.\\[.5em]
    \textbf{Inductive step:} Fix $j<k$ and suppose IH($j+1$) holds.     Define the maps $\Omega_j : \dot{\mathcal{M}}_P\oplus \mathcal{V}_{[j+1]}\rightarrow \dot{\mathcal{M}}_P\oplus \mathcal{V}_{[j]}$ and $\Lambda_{j+1} : \dot{\mathcal{M}}_P\oplus \mathcal{V}_{[j+1]}\rightarrow \mathcal{V}_{j+1}$ as the coordinate projections $\Omega_j : (s,t_{[j+1]})\mapsto (s,t_{[j]})$ and $\Lambda_{j+1} : (s,t_{[j+1]})\mapsto t_{j+1}$. Defining $\gamma_{j+1} : \mathcal{M}\times\mathcal{V}_{[j]}\rightarrow\mathcal{V}_{j+1}$ so that $\gamma_{j+1}(P',v_{[j]})=\theta_{j+1}(P',v_{\mathrm{pa}(j+1)})$, \eqref{eq:etaRecursion} rewrites as
    \begin{align}
       \eta_j(P',v_{[j]}):=\eta_{j+1}\left(P',(v_{[j]},\gamma_{j+1}(P',v_{[j]}))\right). \label{eq:etajIdent}
    \end{align}
    We will apply Lemma~\ref{lem:chainRuleConvenient} to the above to show that $\eta_j$ is totally pathwise differentiable at $(P,h_{[j]})$, but to do so we must first show that $\gamma_{j+1}$ is totally pathwise differentiable at $(P,h_{[j]})$. To see that this is true, note that $\gamma_{j+1}(P',v_{[j]})=\theta_{j+1}(P',\Lambda_{\mathrm{pa}(j+1)}(P',v_{[j]}))$ for all $(P',v_{[j]})$, where $\Lambda_{\mathrm{pa}(j+1)}(P',v_{[j]}):=v_{\mathrm{pa}(j+1)}$, and so Lemmas~\ref{lem:coordProj} and \ref{lem:chainRule} together yield that $\gamma_{j+1}$ is totally pathwise differentiable at $(P,h_{[j]})$ and
    \begin{align}
    \dot{\gamma}_{j+1,P,h_{[j]}}^*(f_{j+1}^{(j+1)})&= \left(\dot{\theta}_{j+1,P,h_{\mathrm{pa}(j+1)}}^*(f_{j+1}^{(j+1)})_0,0\right) + \dot{\Lambda}_{\mathrm{pa}(j+1),P,h_{[j]}}^*\left(\dot{\theta}_{j+1,P,h_{\mathrm{pa}(j+1)}}^*(f_{j+1}^{(j+1)})_1\right) \nonumber \\
    &= \left(\dot{\theta}_{j+1,P,h_{\mathrm{pa}(j+1)}}^*(f_{j+1}^{(j+1)})_0,\mathrm{proj}_{\mathcal{W}_{[j]}}\{a\mymid  \dot{\mathcal{V}}_{[j],h_{[j]}}\}\right), \label{eq:gamAdj}
    \end{align}
    where $a\in \mathcal{W}_{[j]}$ satisfies $a_{\mathrm{pa}(j+1)}=\dot{\theta}_{j+1,P,h_{\mathrm{pa}(j+1)}}^*(f_{j+1}^{(j+1)})_1$ and $a_i=0$ for all $i\not\in\mathrm{pa}(j+1)$. 

    We are now in a position to apply Lemma~\ref{lem:chainRuleConvenient} to \eqref{eq:etajIdent}. Combining that lemma with the fact that $h_{[j+1]}=(h_{[j]},\gamma_{j+1}(P,h_{[j]}))$ shows $\eta_j$ is totally pathwise differentiable at $(P,h_{[j]})$ with
    \begin{align*}
        \dot{\eta}_{j,P,h_{[j]}}^*(f_k)&= \Omega_j\circ \dot{\eta}_{j+1,P,(h_{[j]},\gamma_{j+1}(P,h_{[j]}))}^*(f_k) + \dot{\gamma}_{j+1,P,h_{[j]}}^*\circ \Lambda_{j+1}\circ \dot{\eta}_{j+1,P,(h_{[j]},\gamma_{j+1}(P,h_{[j]}))}^*(f_k) \\
        &= \Omega_j\circ \dot{\eta}_{j+1,P,h_{[j+1]}}^*(f_k) + \dot{\gamma}_{j+1,P,h_{[j]}}^*\circ \Lambda_{j+1}\circ \dot{\eta}_{j+1,P,h_{[j+1]}}^*(f_k).
    \end{align*}
    Since IH($j+1$) holds, $\dot{\eta}_{j+1,P,h_{[j]}}^*(f_k)=(f_0^{(j+1)},\mathrm{proj}_{\mathcal{W}_{[j+1]}}\{f_{[j+1]}^{(j+1)}\mymid  \dot{\mathcal{V}}_{[j+1],h_{[j+1]}}\})$. Combining this with the above and the definitions of $\Omega_j$ and $\Lambda_{j+1}$ yields
    \begin{align}
        \dot{\eta}_{j,P,h_{[j]}}^*(f_k)&= (f_0^{(j+1)},\mathrm{proj}_{\mathcal{W}_{[j]}}\{f_{[j]}^{(j+1)}\mymid  \dot{\mathcal{V}}_{[j],h_{[j]}}\}) + \dot{\gamma}_{j+1,P,h_{[j]}}^*(\mathrm{proj}_{\mathcal{W}_{j+1}}\{f_{j+1}^{(j+1)}\mymid  \dot{\mathcal{V}}_{j+1,h_{j+1}}\}). \label{eq:etaIntermediate}
    \end{align}
    The latter term on the right-hand side equals $\dot{\gamma}_{j+1,P,h_{[j]}}^*(f_{j+1}^{(j+1)})$. To see this, note that the range of $\dot{\gamma}_{j+1,P,h_{[j]}}$ is necessarily a subset of $\dot{\mathcal{V}}_{j+1,h_{j+1}}$, and so, for any $(s,t)\in\dot{\mathcal{M}}_P\oplus \dot{\mathcal{V}}_{[j],h_{[j]}}$,
    \begin{align*}
        &\left\langle (s,t), \dot{\gamma}_{j+1,P,h_{[j]}}^*(\mathrm{proj}_{\mathcal{W}_{j+1}}\{f_{j+1}^{(j+1)}\mymid  \dot{\mathcal{V}}_{j+1,h_{j+1}}\})\right\rangle_{\dot{\mathcal{M}}_P\oplus \dot{\mathcal{V}}_{[j],h_{[j]}}} \\
        &= \left\langle \dot{\gamma}_{j+1,P,h_{[j]}}(s,t), \mathrm{proj}_{\mathcal{W}_{j+1}}\{f_{j+1}^{(j+1)}\mymid  \dot{\mathcal{V}}_{j+1,h_{j+1}}\}\right\rangle_{\mathcal{W}_{j+1}} \\
        &= \left\langle \dot{\gamma}_{j+1,P,h_{[j]}}(s,t), f_{j+1}^{(j+1)}\right\rangle_{\mathcal{W}_{j+1}} = \left\langle (s,t), \dot{\gamma}_{j+1,P,h_{[j]}}^*(f_{j+1}^{(j+1)})\right\rangle_{\dot{\mathcal{M}}_P\oplus \dot{\mathcal{V}}_{[j],h_{[j]}}}.
    \end{align*}
    Replacing the latter term on the right-hand side of \eqref{eq:etaIntermediate} with $\dot{\gamma}_{j+1,P,h_{[j]}}^*(f_{j+1}^{(j+1)})$, combining that display with \eqref{eq:gamAdj}, and then using the linearity of the projection operator yields
    \begin{align*}
        \dot{\eta}_{j,P,h_{[j]}}^*(f_k)&= \left(f_0^{(j+1)}+\dot{\theta}_{j+1,P,h_{\mathrm{pa}(j+1)}}^*(f_{j+1}^{(j+1)})_0,\mathrm{proj}_{\mathcal{W}_{[j]}}\{f_{[j]}^{(j+1)}\mymid  \dot{\mathcal{V}}_{[j],h_{[j]}}\} + \mathrm{proj}_{\mathcal{W}_{[j]}}\{a\mymid  \dot{\mathcal{V}}_{[j],h_{[j]}}\}\right) \\
        &= \left(f_0^{(j+1)}+\dot{\theta}_{j+1,P,h_{\mathrm{pa}(j+1)}}^*(f_{j+1}^{(j+1)})_0,\mathrm{proj}_{\mathcal{W}_{[j]}}\{f_{[j]}^{(j+1)} + a\mymid  \dot{\mathcal{V}}_{[j],h_{[j]}}\}\right).
    \end{align*}
    Recalling the definition of $a$ and the form of the update that occurs on step $j+1$ of the for loop in Algorithm~\ref{alg:backprop} reveals that
    \begin{align*}
(f_0^{(j)},f_{[j]}^{(j)})=\left(f_0^{(j+1)}+\dot{\theta}_{j+1,P,h_{\mathrm{pa}(j+1)}}^*(f_{j+1}^{(j+1)})_0,f_{[j]}^{(j+1)} + a\right).
    \end{align*}
    Combining the preceding two displays shows that IH($j$) holds.
\end{proof}

\subsection{Relationship between total pathwise differentiability and other notions of differentiability}

We now establish results relating total pathwise differentiability to Hadamard and pathwise differentiability. When doing so, we leverage the notation introduced in the paragraph surrounding Eq.~\ref{eq:totalpd} in the main text.

We begin by reviewing the definitions of Hadamard and pathwise differentiability. A map $\zeta : \mathcal{U}\rightarrow \mathcal{W}$ is called Hadamard differentiable at $u\in\mathcal{U}$ if there exists a continuous linear operator $\dot{\zeta}_{u} : \dot{\mathcal{U}}_u\rightarrow\mathcal{W}$ such that, for all $t\in\check{\mathcal{U}}_u$ and $\{u_\epsilon : \epsilon\}\in \mathscr{P}(u,\mathcal{U},t)$,
\begin{align}
    \left\|\zeta(u_\epsilon) - \zeta(u) - \epsilon\,  \dot{\zeta}_{u}(t)\right\|_{\mathcal{W}}&= o(\epsilon). \label{eq:Hadamard}
\end{align}
We call $\dot{\zeta}_u$ the differential operator of $\zeta$ at $u$. 
As noted in Section~20.7 of \cite{van2000asymptotic}, Hadamard differentiability is equivalent to Gateaux differentiability uniformly over compacts. Moving now to pathwise differentiability \citep{van1991differentiable,luedtke2023one}, $\nu : \mathcal{M}\rightarrow\mathcal{W}$ is called pathwise differentiable at $P\in\mathcal{M}$ if there exists a continuous linear operator $\dot{\nu}_{P} : \dot{\mathcal{M}}_P\rightarrow\mathcal{W}$ such that, for all $s\in\check{\mathcal{M}}_P$ and $\{P_\epsilon : \epsilon\}\in \mathscr{P}(P,\mathcal{M},s)$,
\begin{align}
    \left\|\nu(P_\epsilon) - \nu(P) - \epsilon\,  \dot{\nu}_P(s)\right\|_{\mathcal{W}}&= o(\epsilon). \label{eq:pd}
\end{align}
The map $\dot{\nu}_P$ is called the local parameter of $\nu$ at $P$.

As we now show, total pathwise differentiability satisfies an analogue of the fact that the total differentiability of a function $f : \mathbb{R}^2\rightarrow\mathbb{R}$  implies its partial differentiability. In our context, we call $\theta : \mathcal{M}\times\mathcal{U}\rightarrow\mathcal{W}$ partially  differentiable in its first argument at $(P,u)$ if $\theta(\,\cdot\,,u)$ is pathwise differentiable at $P$; we denote the local parameter of $\theta(\,\cdot\,,u)$ at $P$ by $\dot{\nu}_{P,u} : \dot{\mathcal{M}}_P\rightarrow\mathcal{W}$, where the indexing by $u$ emphasizes that $\theta(\,\cdot\,,u)$ depends on $u$. We similarly  call $\theta$ partially differentiable in its second argument at $(P,u)$ if $\theta(P,\,\cdot\,)$ is Hadamard differentiable at $u$, and we denote the differential operator by $\dot{\zeta}_{P,u} : \dot{\mathcal{U}}_u\rightarrow\mathcal{W}$.

\begin{lemma}[Total differentiability implies partial differentiability]\label{lem:totalImpliesPartial}
    If $\theta : \mathcal{M}\times\mathcal{U}\rightarrow\mathcal{W}$ is totally pathwise differentiable at $(P,u)$, then $\theta$ is partially differentiable in its first and second arguments at $(P,u)$. Moreover, $\dot{\theta}_{P,u}(s,t)=\dot{\nu}_{P,u}(s)+\dot{\zeta}_{P,u}(t)$ and $\dot{\theta}_{P,u}^*(w)=(\dot{\nu}_{P,u}^*(w),\dot{\zeta}_{P,u}^*(w))$.
\end{lemma}
\begin{proof}[Proof of Lemma~\ref{lem:totalImpliesPartial}]
    The Hadamard differentiability of $\theta(P,\,\cdot\,) : \mathcal{U}\rightarrow\mathcal{W}$ follows by applying \eqref{eq:totalpd} with the constant path $\{P : \epsilon\in [0,1]\}\in \mathscr{P}(P,\mathcal{M},0)$. Indeed, this shows \eqref{eq:Hadamard} holds with $\zeta(\cdot)$ equal to $\theta(P,\,\cdot\,)$ and $\dot{\zeta}_{P,u}(\cdot)$ equal to the continuous linear operator $\dot{\theta}_{P,u}(0,\,\cdot\,) : \dot{\mathcal{U}}_u\rightarrow\mathcal{W}$.
    
    The pathwise differentiability of  $\theta(\,\cdot\,,u) : \mathcal{M}\rightarrow\mathcal{W}$ similarly follows by applying \eqref{eq:totalpd} with the constant path $\{u : \epsilon\in [0,1]\}\in \mathscr{P}(u,\mathcal{U},0)$. Indeed, this shows \eqref{eq:pd} holds with $\nu(\cdot)$ equal to $\theta(\,\cdot\,,u)$ and $\dot{\nu}_{P,u}$ equal to the continuous linear operator $\dot{\theta}_{P,u}(\,\cdot\,,0) : \dot{\mathcal{M}}_P\rightarrow\mathcal{W}$.

    Because $\dot{\theta}_{P,u}$ is a linear operator, $\dot{\theta}_{P,u}(s,t)=\dot{\theta}_{P,u}(s,0)+\dot{\theta}_{P,u}(0,t)=\dot{\nu}_{P,u}(s)+\dot{\zeta}_{P,u}(t)$ for all $s,t$. Also, by the definition of inner products in direct sum spaces, $\dot{\theta}_{P,u}^*(w)=(\dot{\nu}_{P,u}^*(w),\dot{\zeta}_{P,u}^*(w))$ for all $w$.
\end{proof}

A partial converse of Lemma~\ref{lem:totalImpliesPartial} holds when $\mathcal{U}$ is a linear space. This partial converse is a natural analogue of the fact that a function $f : \mathbb{R}^2\rightarrow\mathbb{R}$ is totally differentiable at $(x,y)$ if, at this point, it is partially differentiable in its first argument and continuously partially differentiable in its second \citep[Theorem~12.11 of][]{apostol1974mathematical}. To define a notion of continuous partial differentiability of $\theta$ in its second argument, we define a notion of continuity for the map $\xi : (P',u')\mapsto \dot{\zeta}_{P',u'}$. When doing so, we require that $\theta$ is partially differentiable in its second argument at all $(P',u')$ in a neighborhood $\mathcal{N}\subseteq \mathcal{M}\times\mathcal{U}$ of $(P,u)$, where here and throughout $\mathcal{M}\times\mathcal{U}$ is equipped with the product topology derived from the Hellinger metric on $\mathcal{M}$ and the $\|\cdot\|_{\mathcal{T}}$-induced metric on $\mathcal{U}$. We take the domain of $\xi$ to be equal to $\mathcal{N}$, where $\mathcal{N}\subseteq \mathcal{M}\times\mathcal{U}$ is equipped with the subspace topology. The image of $\xi$ is contained in the space $\mathcal{B}(\overline{\mathcal{U}},\mathcal{W})$ of bounded linear operators from the $\mathcal{T}$-closure $\overline{\mathcal{U}}$ of $\mathcal{U}$ to $\mathcal{W}$; this is true because the linearity of $\mathcal{U}$ implies that $\dot{\mathcal{U}}_{u'}=\overline{\mathcal{U}}$ for all $u'\in\mathcal{U}$, and a differential operator is, by definition, bounded and linear. We equip $\mathcal{B}(\overline{\mathcal{U}},\mathcal{W})$ with the operator norm $\|f\|_{\mathrm{op}}:=\sup_{t\in\overline{\mathcal{U}} : \|t\|_{\mathcal{T}}\le 1} \|f(t)\|_{\mathcal{W}}$. This results in the following definition: $\theta$ is called continuously partially differentiable in its second argument at $(P,u)$ if $\theta$ is partially differentiable in its second argument at all $(P',u')$ in a neighborhood $\mathcal{N}$ of $(P,u)$ and, moreover, $\xi : \mathcal{N}\rightarrow\mathcal{B}(\overline{\mathcal{U}},\mathcal{W})$ is continuous at $(P,u)$.

\begin{lemma}[Partial converse of Lemma~\ref{lem:totalImpliesPartial}]\label{lem:diffTheorem}
    Suppose $\mathcal{U}$ is a linear space. If $\theta : \mathcal{M}\times\mathcal{U}\rightarrow\mathcal{W}$ is partially differentiable in its first argument at $(P,u)$ and continuously partially differentiable in its second argument at $(P,u)$, then $\theta$ is totally pathwise differentiable at $(P,u)$ with differential operator $\dot{\theta}_{P,u}(s,t)=\dot{\nu}_{P,u}(s) + \dot{\zeta}_{P,u}(t)$.
\end{lemma}
\begin{proof}[Proof of Lemma~\ref{lem:diffTheorem}]
    We suppose that the neighborhood $\mathcal{N}$ on which $\theta$ is partially differentiable in its second argument takes the form $\mathcal{N}_P\times\mathcal{N}_u\subset \mathcal{M}\times\mathcal{U}$, where $\mathcal{N}_u=\{u'\in\mathcal{U} : \|u'-u\|_{\mathcal{T}}< r\}$ for some $r>0$. No generality is lost by doing this since Cartesian products of open sets form a base for the product topology and open balls form a base for the topology on the normed space $(\mathcal{U},\|\cdot\|_{\mathcal{T}})$.

    Fix $s\in\check{\mathcal{M}}_P$, $t\in\check{\mathcal{U}}_u$, $\{P_\epsilon : \epsilon\in [0,1]\}\in \mathscr{P}(P,\mathcal{M},s)$, and $\{u_\epsilon : \epsilon\in [0,1]\}\in \mathscr{P}(u,\mathcal{U},t)$. We will show that \eqref{eq:totalpd} holds with $\dot{\theta}_{P,u}(s,t)=\dot{\nu}_{P,u}(s) + \dot{\zeta}_{P,u}(t)$. We leverage the decomposition
    \begin{align}
        \theta(P_\epsilon,u_\epsilon) - \theta(P,u)&= \left[\theta(P_\epsilon,u) - \theta(P,u)\right] + \epsilon\, \dot{\zeta}_{P_\epsilon,u}(t) + \left[\theta(P_\epsilon,u_\epsilon) -\theta(P_\epsilon,u) - \epsilon\, \dot{\zeta}_{P_\epsilon,u}(t)\right]. \label{eq:partialImpliesTotalDecomp}
    \end{align}
    By the partial differentiability of $\theta$ in its first argument, the first term is equal to $\epsilon\, \dot{\nu}_{P,u}(s) + o(\epsilon)$. By the Hellinger-continuity of  $P'\mapsto \dot{\zeta}_{P',u}$ and the fact that the quadratic mean differentiability of $\{P_\epsilon : \epsilon\}$ implies that $P_\epsilon\rightarrow P$ in a Hellinger sense, the second term is equal to $\epsilon\, \dot{\zeta}_{P,u}(t) + o(\epsilon)$. In the remainder of this proof, we show that the third term is $o(\epsilon)$. Since the boundedness and linearity of $\dot{\nu}_{P,u}$ and $\dot{\zeta}_{P,u}$ imply that $(s,t)\mapsto \dot{\nu}_{P,u}(s) + \dot{\zeta}_{P,u}(t)$ is bounded and linear as well, this will establish the result.

    Our study of the third term in \eqref{eq:partialImpliesTotalDecomp} relies on the fact that, for any $\epsilon>0$,
    \begin{align}
        \|\theta(P_\epsilon,u_\epsilon)-\theta(P_\epsilon,u) - \epsilon\, \dot{\zeta}_{P_\epsilon,u}(t)\|_{\mathcal{W}}&= \sup_{w\in\mathcal{W} : \|w\|_{\mathcal{W}}=1}\left\langle\theta(P_\epsilon,u_\epsilon)-\theta(P_\epsilon,u) - \epsilon\, \dot{\zeta}_{P_\epsilon,u}(t),w\right\rangle_{\mathcal{W}}. \label{eq:normInnerProd}
    \end{align}
    In what follows we shall obtain a bound on the quantity in the supremum that holds uniformly over all $w\in\mathcal{W}$ with $\|w\|_{\mathcal{W}}=1$. For now, fix such a $w$. Our first goal will be to show that, for all $\epsilon$ less than some $\epsilon_0>0$ to be defined momentarily, there is an intermediate element $u^\dagger$ between $u_\epsilon$ and $u$ such that
    \begin{align}
        \langle\theta(P_\epsilon,u_\epsilon)-\theta(P_\epsilon,u),w\rangle_{\mathcal{W}}=\langle\dot{\zeta}_{P_\epsilon,u^\dagger}(u_\epsilon-u),w\rangle_{\mathcal{W}}, \label{eq:mvt}
    \end{align}
    This result will follow from the mean value theorem, once we show that theorem is applicable. We take $\epsilon_0$ to be any positive quantity such that $u_\epsilon\in\mathcal{N}_u$ and $P_\epsilon\in\mathcal{N}_P$ for all $\epsilon<\epsilon_0$ --- such an $\epsilon_0$ necessarily exists since, as $\epsilon\rightarrow 0$, $u_\epsilon\rightarrow u$ and $P_\epsilon\rightarrow P$. We suppose $\epsilon<\epsilon_0$ hereafter. Since $u_\epsilon$ and $u$ both belong to $\mathcal{N}_u$ and the linearity of $\mathcal{U}$ implies that $\mathcal{N}_u$ is convex, $u_\epsilon^{(a)}:=au_\epsilon + (1-a)u\in\mathcal{N}_u$ for all $a\in[0,1]$. Hence, $\theta(P_\epsilon,\,\cdot\,)$ is Hadamard differentiable at $u_\epsilon^{(a)}$ for each $a\in [0,1]$. Using that $u_\epsilon^{(a+\delta)}=u_\epsilon^{(a)} + \delta(u_\epsilon-u)$ for all $a\in [0,1)$ and $\delta>0$, this implies that, for all $a\in [0,1)$,
    \begin{align*}
        \lim_{\delta\downarrow 0} \frac{\theta(P_\epsilon,u_\epsilon^{(a+\delta)}) - \theta(P_\epsilon,u_\epsilon^{(a)})}{\delta}= \dot{\zeta}_{P_\epsilon,u_\epsilon^{(a)}}(u_\epsilon-u).
    \end{align*}
    Similarly, for all $a\in (0,1]$, the same display holds but with $\lim_{\delta\downarrow 0}$ replaced by $\lim_{\delta\uparrow 0}$. Since $\langle\,\cdot\, , w\rangle_{\mathcal{W}} : \mathcal{W}\rightarrow\mathbb{R}$ is continuous, our expressions for the limits as $\delta\downarrow 0$ and $\delta\uparrow 0$ show that $a\mapsto \langle \theta(P_\epsilon,u_\epsilon^{(a)}), w\rangle_{\mathcal{W}}$ is (i) continuous on $[0,1]$ and (ii) differentiable on $(0,1)$ with derivative $a\mapsto \langle \dot{\zeta}_{P_\epsilon,u_\epsilon^{(a)}}(u_\epsilon-u), w\rangle_{\mathcal{W}}$. Consequently, by the mean value theorem, there exists an $a^\star\in (0,1)$ such that \eqref{eq:mvt} holds with $u^\dagger=u_\epsilon^{(a^\star)}$. Taking a supremum in \eqref{eq:mvt} over all possible values $a^\star$ could take yields
    \begin{align*}
        \langle\theta(P_\epsilon,u_\epsilon)-\theta(P_\epsilon,u),w\rangle_{\mathcal{W}}\le \sup_{a\in (0,1)} \left\langle\dot{\zeta}_{P_\epsilon,u_\epsilon^{(a)}}(u_\epsilon-u),w\right\rangle_{\mathcal{W}}.
    \end{align*}
    We let $t_\epsilon:=(u_\epsilon-u)/\epsilon$, and note that the linearity of $\dot{\zeta}_{P_\epsilon,u_\epsilon^{(a)}}$ implies that $\dot{\zeta}_{P_\epsilon,u_\epsilon^{(a)}}(u_\epsilon-u)=\epsilon \,\dot{\zeta}_{P_\epsilon,u_\epsilon^{(a)}}(t_\epsilon)$. Plugging this into the above, dividing both sides by $\epsilon$, and then subtracting $\langle \dot{\zeta}_{P_\epsilon,u}(t),w\rangle_{\mathcal{W}}$ yields that
    \begin{align*}
        \epsilon^{-1}\langle\theta(P_\epsilon,u_\epsilon)-\theta(P_\epsilon,u) - \epsilon\, \dot{\zeta}_{P_\epsilon,u}(t),w\rangle_{\mathcal{W}}\le \sup_{a\in (0,1)} \left\langle\dot{\zeta}_{P_\epsilon,u_\epsilon^{(a)}}(t_\epsilon) - \dot{\zeta}_{P_\epsilon,u}(t),w\right\rangle_{\mathcal{W}}.
    \end{align*}
    Applying Cauchy-Schwarz and using that $\|w\|_{\mathcal{W}}=1$ shows the right-hand side above can be upper bounded by $\sup_{a\in (0,1)} \|\dot{\zeta}_{P_\epsilon,u_\epsilon^{(a)}}(t_\epsilon) - \dot{\zeta}_{P_\epsilon,u}(t)\|_{\mathcal{W}}$. This bound does not depend on $w$, and so \eqref{eq:normInnerProd} implies
    \begin{align*}
         \epsilon^{-1}\|\theta(P_\epsilon,u_\epsilon)-\theta(P_\epsilon,u) - \epsilon\, \dot{\zeta}_{P_\epsilon,u}(t)\|_{\mathcal{W}}\le \sup_{a\in (0,1)} \left\|\dot{\zeta}_{P_\epsilon,u_\epsilon^{(a)}}(t_\epsilon) - \dot{\zeta}_{P_\epsilon,u}(t)\right\|_{\mathcal{W}}.
    \end{align*}
    Applying the triangle inequality and leveraging the definition of the operator norm, we see that
    \begin{align*}
         \epsilon^{-1}\|\theta(P_\epsilon,u_\epsilon)-\theta(P_\epsilon,u) - \epsilon\, \dot{\zeta}_{P_\epsilon,u}(t)\|_{\mathcal{W}}&\le \sup_{a\in (0,1)} \left\|\dot{\zeta}_{P_\epsilon,u_\epsilon^{(a)}}(t_\epsilon) - \dot{\zeta}_{P_\epsilon,u}(t)\right\|_{\mathcal{W}} \\
         &\le \sup_{a\in (0,1)} \left\|\dot{\zeta}_{P_\epsilon,u_\epsilon^{(a)}}(t_\epsilon) - \dot{\zeta}_{P_\epsilon,u}(t_\epsilon)\right\|_{\mathcal{W}} + \left\|\dot{\zeta}_{P_\epsilon,u}(t_\epsilon-t)\right\|_{\mathcal{W}} \\
         &\le \|t_\epsilon\|_{\mathcal{T}} \sup_{a\in (0,1)} \left\|\dot{\zeta}_{P_\epsilon,u_\epsilon^{(a)}} - \dot{\zeta}_{P_\epsilon,u}\right\|_{\mathrm{op}}+ \left\|t_\epsilon-t\right\|_{\mathcal{T}}\|\dot{\zeta}_{P_\epsilon,u}\|_{\mathrm{op}}.
    \end{align*}
    We now explain why the right-hand side goes to zero when $\epsilon\rightarrow 0$, which will show that the third term in \eqref{eq:partialImpliesTotalDecomp} is $o(\epsilon)$. 
    Since $u_\epsilon\rightarrow u$, $\sup_{a\in (0,1)}\|u_\epsilon^{(a)}-u\|_{\mathcal{T}}\le \|u_\epsilon-u\|_{\mathcal{T}}=o(1)$. Combining this with the fact that $P_\epsilon\rightarrow P$ and the continuity of $(P',u')\mapsto \dot{\zeta}_{P',u'}$ at $(P,u)$ shows that $\sup_{a\in (0,1)} \left\|\dot{\zeta}_{P_\epsilon,u_\epsilon^{(a)}} - \dot{\zeta}_{P_\epsilon,u}\right\|_{\mathrm{op}}=o(1)$; this continuity also shows that $\|\dot{\zeta}_{P_\epsilon,u}\|_{\mathrm{op}}=O(1)$. Finally, since $t_\epsilon\rightarrow t$, $\|t_\epsilon\|_{\mathcal{T}}=O(1)$ and $\|t_\epsilon-t\|_{\mathcal{T}}=o(1)$. Combining all of these observations shows that the right-hand side above is $o(1)$.
\end{proof}

\section{Estimation: supplemental theoretical results and proofs}\label{app:est}

\subsection{Theoretical guarantees}\label{app:estReal}

\begin{proof}[Proof of Theorem~\ref{thm:efficient}]
    By Theorem~\ref{thm:backpropWorks}, $\psi$ is pathwise differentiable at $P$ with EIF $f_0$. Adding and subtracting terms shows that $\widehat{\psi}-\psi(P)= \frac{1}{n}\sum_{i=1}^n f_0(Z_i) + \mathcal{D}_n + \mathcal{R}_n$. The result follows since $\mathcal{D}_n$ and $\mathcal{R}_n$ are $o_p(n^{-1/2})$ by assumption.
\end{proof}

For $\psi$ real-valued, define the $L$-fold cross-fitted estimator
\begin{align*}
    \widehat{\psi}_{\mathrm{cf}}:=\frac{1}{L}\sum_{\ell=1}^L \left[\widehat{h}_{k}^{(\ell)} + \frac{1}{|\mathcal{I}^{(\ell)}|}\sum_{i\in \mathcal{I}^{(\ell)}} \widehat{f}_{0}^{(\ell)}(Z_i)\right],
\end{align*}
where $\mathcal{I}^{(1)},\mathcal{I}^{(2)},\ldots,\mathcal{I}^{(L)}$ is a partition of $[n]$ into $L$ subsets of approximately equal size --- i.e., with cardinality in $n/L\pm 1$ --- and, for each $\ell\in [L]$, $(\widehat{h}_{k}^{(\ell)},\widehat{f}_{0}^{(\ell)})$ are the outputs of Algorithm~\ref{alg:estimatedBackprop} when the observations with indices in $[n]\backslash\mathcal{I}^{(\ell)}$ are used for nuisance estimation. Letting $P_{n}^{(\ell)}$ denote the empirical distribution of the observations with indices in $\mathcal{I}^{(\ell)}$, we define the following drift and remainder term for each $\ell\in [L]$:
\begin{align*}
    \mathcal{D}_{n}^{(\ell)}&:= \int [\widehat{f}_{0}^{(\ell)}(z)-\dot{\psi}_P^*(1)(z)] (P_{n}^{(\ell)}-P)(dz),\hspace{3.5em} 
    \mathcal{R}_{n}^{(\ell)}:= \widehat{h}_{k}^{(\ell)} + \int \widehat{f}_{0}^{(\ell)}(z)\, P(dz)  - \psi(P).
\end{align*}
The following result shows that $\widehat{\psi}_{\mathrm{cf}}$ is efficient under appropriate conditions.
\begin{theorem}[Efficiency of $\widehat{\psi}_{\mathrm{cf}}$]\label{thm:CFefficient}
    Suppose $\psi$ is real-valued, $f_k=1$, the conditions of Theorem~\ref{thm:backpropWorks} hold, and $L$ is fixed as $n\rightarrow\infty$. If $\mathcal{D}_{n}^{(\ell)}$ and $\mathcal{R}_{n}^{(\ell)}$ are $o_p(n^{-1/2})$ for all $\ell\in [L]$, then $\widehat{\psi}_{\mathrm{cf}}$ is efficient, in that
    \begin{align*}
        \widehat{\psi}_{\mathrm{cf}} - \psi(P)=\frac{1}{n}\sum_{i=1}^n \dot{\psi}_P^*(1)(Z_i) + o_p(n^{-1/2}).
    \end{align*}
\end{theorem}
Before proceeding with the proof, we note that the conditions of Theorem~\ref{thm:CFefficient} are typically weaker than those of Theorem~\ref{thm:efficient}, in that there is no need to require an empirical process condition to ensure negligibility of the drift term --- see \cite{schick1986asymptotically} or the proof of Lemma~4 in \cite{luedtke2023one} for details.
\begin{proof}[Proof of Theorem~\ref{thm:CFefficient}]
    By Theorem~\ref{thm:backpropWorks}, $\psi$ is pathwise differentiable at $P$ with EIF $f_0$. Adding and subtracting terms shows that
    \begin{align*}
        \widehat{\psi}_{\mathrm{cf}}-\psi(P)= \frac{1}{L}\sum_{\ell=1}^L \left[\frac{1}{|\mathcal{I}^{(\ell)}|}\sum_{i\in \mathcal{I}^{(\ell)}} \dot{\psi}_P^*(1)(Z_i) + \mathcal{D}_{n}^{(\ell)} + \mathcal{R}_{n}^{(\ell)}\right].
    \end{align*}
    Combining this with the fact that $\mathcal{R}_{n}^{(\ell)}$ and $\mathcal{D}_{n}^{(\ell)}$ are $o_p(n^{-1/2})$ by assumption gives that
    \begin{align*}
         \widehat{\psi}_{\mathrm{cf}}-\psi(P)= \frac{1}{L}\sum_{\ell=1}^L \frac{1}{|\mathcal{I}^{(\ell)}|}\sum_{i\in \mathcal{I}^{(\ell)}} \dot{\psi}_P^*(1)(Z_i).
    \end{align*}
    Since the sizes of the sets $\mathcal{I}^{(1)},\mathcal{I}^{(2)},\ldots,\mathcal{I}^{(L)}$ in the partition of $[n]$ differ by at most 1, the sum on the right is equal to $\frac{1}{n}\sum_{i=1}^n \dot{\psi}_P^*(1)(Z_i) + o_p(n^{-1/2})$, which gives the result.
\end{proof}

\subsection{Motivating the plausibility of $\mathcal{R}_n=o_p(n^{-1/2})$ using von Mises calculus}\label{app:vonMises}

We now study the plausibility of the condition that the remainder term in Theorem~\ref{thm:efficient} is $o_p(n^{-1/2})$. We do this by studying the plausibility of
\begin{align}
    \mathcal{R}_n:=\widehat{h}_k + \int \widehat{f}_0(z) \,P(dz)-\psi(P)=o_p(n^{-1/2}) \label{eq:vonMises}
\end{align}
when the nuisance estimators $\widehat{P}_1,\widehat{P}_2,\ldots,\widehat{P}_k$ in Algorithm~\ref{alg:estimatedBackprop} are replaced by perturbations $P_{1,\epsilon},P_{2,\epsilon},\ldots,P_{k,\epsilon}$ of $P$ belonging to smooth submodels $\mathscr{P}(P,\mathcal{M},s_1),\mathscr{P}(P,\mathcal{M},s_2),\ldots,\mathscr{P}(P,\mathcal{M},s_k)$, the estimators of the adjoint evaluations $\dot{\theta}_{j,P,\widehat{h}_{\mathrm{pa}(j)}}^*(\widehat{f}_j)$ are replaced by $\delta_{j,\epsilon}$-perturbations of $\dot{\theta}_{j,P_{j,\epsilon},h_{\mathrm{pa}(j),\epsilon}}^*(f_{j,\epsilon})$, and the $o_p(n^{-1/2})$ term is replaced by $o(\epsilon)$. Concretely, we investigate whether, as $\epsilon\rightarrow 0$, the outputs $h_{k,\epsilon}$ and $f_{0,\epsilon}$ from Algorithm~\ref{alg:submodelBackprop} yield a von Mises expansion of the form
\begin{align}
h_{k,\epsilon} + \int f_{0,\epsilon}(z) \,P(dz) - \psi(P)&= o(\epsilon). \label{eq:vonMisesDeterministic}
\end{align}
For any $\epsilon\in [0,1]$, we use the shorthand notation $\vartheta_{j,\epsilon,0}$ to denote the first entry of the quantity $\vartheta_{j,\epsilon}$ defined on line~\ref{ln:approxBackward} of the algorithm. The function $\vartheta_{j,\epsilon,0}$ is $P_{j,\epsilon}$-mean zero. This follows since it is defined as the sum of the first entries of $\dot{\theta}_{j,P_{j,\epsilon},h_{\mathrm{pa}(j),\epsilon}}^*(f_{j,\epsilon})$ and $\delta_{j,\epsilon}$, both of which are $P_{j,\epsilon}$-mean zero.

\begin{algorithm}[tb]
   \caption{Modification of Algorithm~\ref{alg:estimatedBackprop} with nuisance estimates replaced by deterministic perturbations}
   \label{alg:submodelBackprop}
   \linespread{1.05}\selectfont
\begin{algorithmic}[1]
    \Require \begin{itemize}[leftmargin=1.5em]
        \item  real-valued $\psi$ expressed as in Algorithm~\ref{alg:parameter} and $f_k=1$
        \item[]\hspace{\dimexpr\labelwidth+\labelsep}$\bullet$ smooth submodels $\{P_{j,\epsilon} : \epsilon\in [0,1]\}=\mathscr{P}(P,\mathcal{M},s_j)$, $j\in [k]$
        \item[]\hspace{\dimexpr\labelwidth+\labelsep}$\bullet$ errors $\delta_{j,\epsilon}\in L_0^2(P_{j,\epsilon})\oplus\left(\oplus_{i\in\mathrm{pa}(j)}\mathcal{W}_i\right)$ satisfying $\lim_{\epsilon\rightarrow 0}\delta_{j,\epsilon}= 0$, $j\in [k]$,
        \item[] \hspace{\dimexpr\labelwidth+\labelsep}\hphantom{$\bullet$} where $L_0^2(P_{j,\epsilon}):=\{s\in L^2(P_{j,\epsilon}) : \int s\,dP_{j,\epsilon}=0\}$ is a Hilbert subspace of $L^2(P_{j,\epsilon})$
        \item[]\hspace{\dimexpr\labelwidth+\labelsep}$\bullet$ magnitude of perturbation $\epsilon\in [0,1]$
    \end{itemize}
    \algrule
    \Statex $\triangleright$ {\color{CBblue}\textbf{forward pass}} to obtain an initial approximation
    \For {$j=1,2,\ldots,k$}
        \State \textbf{approximate} $\theta_j(P,{\color{CBteal}h_{\mathrm{pa}(j),\epsilon}})$ with ${\color{CBblue}h_{j,\epsilon}}=\theta(P_{j,\epsilon},{\color{CBteal}h_{\mathrm{pa}(j),\epsilon}})$, where ${\color{CBteal}h_{\mathrm{pa}(j),\epsilon}}:=({\color{CBblue}h_{i,\epsilon}})_{i\in\mathrm{pa}(j)}$
    \EndFor
    \algrule
    \Statex $\triangleright$ {\color{CBmagenta}\textbf{backward pass}} to enable first-order improvement of initial approximation
    \State \textbf{initialize} ${\color{CBmagenta}f_{0,\epsilon}},{\color{CBorange}f_{1,\epsilon}},\ldots,{\color{CBorange}f_{k-1,\epsilon}}$ to be the 0 elements of $\mathbb{R},\mathcal{W}_1,\ldots,\mathcal{W}_{k-1}$, and ${\color{CBorange}f_{k,\epsilon}}=f_k$
    \For {$j=k,k-1,\ldots,1$}
        \State \textbf{approximate} $\dot{\theta}_{j,P,{\color{CBteal}h_{\mathrm{pa}(j),\epsilon}}}^*({\color{CBorange}f_{j,\epsilon}})$ with ${\color{CBred}\vartheta_{j,\epsilon}}:=\dot{\theta}_{j,P_{j,\epsilon},{\color{CBteal}h_{\mathrm{pa}(j),\epsilon}}}^*({\color{CBorange}f_{j,\epsilon}}) + \delta_{j,\epsilon}$ \label{ln:approxBackward}\vspace{.2em}
        \Statex \hspace{\algorithmicindent} $\triangleright$ ${\color{CBred}\vartheta_{j,\epsilon,0}}$ compatible with ${\color{CBblue}h_{j,\epsilon}}, {\color{CBteal}h_{\mathrm{pa}(j),\epsilon}}$: $\exists {\color{slateGrey}\widetilde{P}}\in\mathcal{M}$ s.t. ${\color{CBblue}h_{j,\epsilon}}=\theta_j({\color{slateGrey}\widetilde{P}},{\color{CBteal}h_{\mathrm{pa}(j),\epsilon}})$, $\int {\color{CBred}\vartheta_{j,\epsilon,0}}\,d{\color{slateGrey}\widetilde{P}}=0$
        \Statex \hphantom{\algorithmicindent $\triangleright$ $\vartheta_{j,\epsilon}$ is compatible with ${\color{CBblue}h_{j,\epsilon}}, {\color{CBteal}h_{\mathrm{pa}(j),\epsilon}}$:}\hspace{-.5em}  \raisebox{1.5pt}[0pt][0pt]{\rotatebox[origin=m]{180}{\large$\Lsh$}} in particular, ${\color{slateGrey}\widetilde{P}}=P_{j,\epsilon}$\vspace{.2em}
        \State \textbf{augment} $({\color{CBmagenta}f_{0,\epsilon}},{\color{CBorange}f_{\mathrm{pa}(j),\epsilon}}) \mathrel{+{=}} {\color{CBred}\vartheta_{j,\epsilon}}$ 
    \EndFor
    \algrule
    \State \Return initial approximation {\color{CBblue}$h_{k,\epsilon}$} and approximate efficient influence operator ${\color{CBmagenta}f_{0,\epsilon}}$
\end{algorithmic}
\end{algorithm}

\begin{lemma}[Plausibility of condition on $\mathcal{R}_n$]\label{lem:vonMises}
Suppose all distributions in $\mathcal{M}$ are equivalent, in that, for all $P_1,P_2\in\mathcal{M}$, $P_1\ll P_2$ and $P_2\ll P_1$. Fix $P\in\mathcal{M}$ and, for $j\in [k]$, $s_j\in\check{\mathcal{M}}_P$ and $\{P_{j,\epsilon} : \epsilon\in [0,1]\}\in \mathscr{P}(P,\mathcal{M},s_j)$. Suppose each $\theta_j$ is totally pathwise differentiable at $(P_{j,\epsilon},h_{\mathrm{pa}(j),\epsilon})$ for all $\epsilon\in [0,1]$. Further suppose, for each $j\in [k]$, (i) there exists a constant $c_j<\infty$ such that $\vartheta_{j,\epsilon,0}\le c_j$ $P$-a.s. for all sufficiently small $\epsilon$ and (ii) $\vartheta_{j,\epsilon,0}\rightarrow \vartheta_{j,0,0}$ in $L^2(P)$ as $\epsilon\rightarrow 0$. Then, \eqref{eq:vonMisesDeterministic} holds.
\end{lemma}
The $L^2(P)$-consistency condition (ii) is a stability requirement that ensures the influence operator output $f_{0,\epsilon}$ of Algorithm~\ref{alg:submodelBackprop} converges to its true value $\dot{\psi}_P^*(1)$ as the perturbations vanish. Before providing a detailed proof, we give a sketch that indicates where (ii) and the compatibility condition are used. Our proof begins by noting that the output $h_{k,\epsilon}$ of the forward pass of Algorithm~\ref{alg:submodelBackprop} can be expressed as the evaluation $\phi(\prod_{j=1}^k P_{j,\epsilon})$ of a pathwise differentiable parameter $\phi$ defined on the $k$-fold tensor product of the model $\mathcal{M}$, where $\prod_{j=1}^k P_{j,\epsilon}$ denotes a product measure. A von Mises expansion of $\phi$ then yields
\begin{align*}
    h_{k,\epsilon}-\psi(P)&= \sum_{j=1}^k \int \vartheta_{j,0,0}(z) (P_{j,\epsilon}-P)(dz) + o(\epsilon).
\end{align*}
Compatibility ensures that each $\vartheta_{j,\epsilon,0}$ is $P_{j,\epsilon}$-mean zero, which allows us to add and subtract terms on the right-hand side above to find that
\begin{align*}
            h_{k,\epsilon}-\psi(P)&=  - \sum_{j=1}^k \int \vartheta_{j,\epsilon,0}(z) P(dz) - \sum_{j=1}^k \int [\vartheta_{j,\epsilon,0}(z)-\vartheta_{j,0,0}(z)] (P_{j,\epsilon}-P)(dz) + o(\epsilon).
\end{align*}
The condition that $\vartheta_{j,\epsilon,0}\rightarrow \vartheta_{j,0,0}$ in $L^2(P)$ as $\epsilon\rightarrow 0$ can then be used to ensure that the second term is $o(\epsilon)$, which will give the result.

\begin{proof}[Proof of Lemma~\ref{lem:vonMises}]
    Let $\mathcal{M}^{\otimes k}:=\{\prod_{j=1}^k P_j : P_1,P_2,\ldots,P_k\in\mathcal{M}\}$ and $P^k:=\prod_{j=1}^k P$. Define $\phi : \mathcal{M}^{\otimes k}\rightarrow\mathbb{R}$ so that $\phi(\prod_{j=1}^k P_j)$ is the output $\underline{h}_k$ of the modification of Algorithm~\ref{alg:parameter} that replaces each line $j\in [k]$ by $\underline{h}_j=\underline{\theta}_j(P_j,\underline{h}_{\mathrm{pa}(j)})$, where $\underline{\theta}_j(\prod_{j=1}^k P_j,\underline{h}_{\mathrm{pa}(j)})=\theta_j(P_j,\underline{h}_{\mathrm{pa}(j)})$. The total pathwise differentiability of $\theta_j$ at $(P,h_{\mathrm{pa}(j)})$ implies the total pathwise differentiability of $\underline{\theta}_j$ at $(P^k,h_{\mathrm{pa}(j)})$ with adjoint $\underline{\theta}_{j,\prod_{i=1}^k P_i,h_{\mathrm{pa}(j)}}^*(f_j) : z_{[k]}\mapsto \theta_{j,P_j,h_{\mathrm{pa}(j)}}^*(f_j)(z_j)$. Consequently, Theorem~\ref{thm:backpropWorks} shows that $\phi$ is pathwise differentiable at $P^k$ with EIF $z_{[k]}\mapsto \sum_{j=1}^k \vartheta_{j,0,0}(z_j)$. By the definition of pathwise differentiability and the fact that $\{\prod_{j=1}^k P_{j,\epsilon} : \epsilon\in [0,1]\}\in \mathscr{P}\big(P^k,\mathcal{M}^{\otimes k},\underline{s}\big)$ with $\underline{s} : z_{[k]}\mapsto \sum_{j=1}^k s_j(z_j)$, we see that $\phi\big(\prod_{j=1}^k P_{j,\epsilon}\big) - \phi\left(P^k\right) = \epsilon \dot{\phi}_{P^k}(\underline{s}) +  o(\epsilon)$, where $\dot{\phi}_{P^k}$ denotes the local parameter of $\phi$ at $P^k$. Using that the EIF $z_{[k]}\mapsto \sum_{j=1}^k \vartheta_{j,0,0}(z_j)$ at $P^k$ is $P^k$-a.s. bounded and that all distributions in $\mathcal{M}^{\otimes k}$ are equivalent, we can apply Lemma~S5 in \cite{luedtke2023one} to show that $\epsilon \dot{\phi}_{P^k}(\underline{s})=\int \big[\sum_{j=1}^k \vartheta_{j,0,0}(z_j)\big] (\prod_{j=1}^k P_{j,\epsilon}-P^k)(dz_{[k]})+ o(\epsilon)$. Combining the results from the preceding two sentences and then rearranging terms,
    \begin{align*}
        \phi\left({\textstyle\prod_{j=1}^k }P_{j,\epsilon}\right) - \phi\left(P^k\right)&= \int \left[\sum_{j=1}^k \vartheta_{j,0,0}(z_j)\right] \left(\prod_{j=1}^k P_{j,\epsilon}-P^k\right)(dz_{[k]}) + o(\epsilon) \\
        &= \sum_{j=1}^k \int \vartheta_{j,0,0}(z) (P_{j,\epsilon}-P)(dz) + o(\epsilon) \\
        &= - \sum_{j=1}^k \int \vartheta_{j,\epsilon,0}(z) P(dz) + \sum_{j=1}^k \int \vartheta_{j,\epsilon,0}(z) P_{j,\epsilon}(dz) \\
        &\quad- \sum_{j=1}^k \int [\vartheta_{j,\epsilon,0}(z)-\vartheta_{j,0,0}(z)] (P_{j,\epsilon}-P)(dz) + o(\epsilon).
    \end{align*}
    The second term on the right is $0$ since $\int \vartheta_{j,\epsilon,0}\,dP_{j,\epsilon}=0$ for all $j\in [k]$. Moreover, noting that $\phi(\prod_{j=1}^k P_{j,\epsilon})=h_{k,\epsilon}$, $\phi(P^k)=\psi(P)$, and $\sum_{j=1}^k \vartheta_{j,\epsilon,0}=f_{0,\epsilon}$, we can add the first term on the right to both sides to see that
    \begin{align}
        h_{k,\epsilon} + \int f_{0,\epsilon}(z) \,P(dz) - \psi(P)&= - \sum_{j=1}^k \int [\vartheta_{j,\epsilon,0}(z)-\vartheta_{j,0,0}(z)] (P_{j,\epsilon}-P)(dz) + o(\epsilon). \label{eq:eq:vonMisesDeterministicAlmost}
    \end{align}
    Eq.~\ref{eq:vonMisesDeterministic} will follow if we can show that the first term on the right is $o(\epsilon)$. To this end, fix $j\in [k]$. Let $\lambda$ be some measure that dominates all measures in $\mathcal{M}$, $p_{j,\epsilon}:=\frac{dP_{j,\epsilon}}{d\lambda}$, and $p:=\frac{dP}{d\lambda}$. Writing $\int f dQ$ as shorthand for $\int f(z)Q(dz)$, adding and subtracting terms shows that
    \begin{align*}
        \int [\vartheta_{j,\epsilon,0}(z)-\vartheta_{j,0,0}(z)] (P_{j,\epsilon}-P)(dz)
        &= \epsilon \int (\vartheta_{j,\epsilon,0}-\vartheta_{j,0,0}) s_jdP + \int (\vartheta_{j,\epsilon,0}-\vartheta_{j,0,0})(p_{j,\epsilon}^{1/2}-p^{1/2})^2 d\lambda \\
        &\quad+ 2\int (\vartheta_{j,\epsilon,0}-\vartheta_{j,0,0})p^{1/2} \left(p_{j,\epsilon}^{1/2}-p^{1/2}-\epsilon  s_jp^{1/2}/2\right) d\lambda.
    \end{align*}
    The first term is $o(\epsilon)$ by Cauchy-Schwarz and the fact that $\vartheta_{j,\epsilon,0}\rightarrow \vartheta_{j,0,0}$ in $L^2(P)$, the second is $O(\epsilon^2)$ since it is upper bounded by $2c_j$ times the squared Hellinger distance between $P_{j,\epsilon}\in \mathscr{P}(P,\mathcal{M},s_j)$ and $P$, and the third is $o(\epsilon)$ by Cauchy-Schwarz and the fact that $P_{j,\epsilon}\in \mathscr{P}(P,\mathcal{M},s_j)$. Hence, the first term on the right of \eqref{eq:eq:vonMisesDeterministicAlmost} is $o(\epsilon)$, yielding \eqref{eq:vonMisesDeterministic}.
\end{proof}

\subsection{Hilbert-valued parameters}\label{app:HilbertValParam}
Suppose $\psi(P)$ takes values in a separable, possibly infinite-dimensional Hilbert space $\mathcal{W}_\psi$. For example, $\mathcal{W}_\psi$ may be an $L^2$ space. An estimator can be constructed by calling Algorithm~\ref{alg:estimatedBackprop} multiple times, each time with the input $f_k$ corresponding to a different element of a truncation $(b_m)_{m\in [M]}$ of an orthonormal basis of $\mathcal{W}_\psi$. Denoting the outputs of these $M$ calls by $(\widehat{h}_{m,k},\widehat{f}_{m,0})$, $m\in [M]$, $\psi(P)$ can be estimated with
\begin{align*}
\widehat{\psi}_{M}:= \sum_{m=1}^M \left(\frac{1}{M}\widehat{h}_{m,k} + \left[\frac{1}{n}\sum_{i=1}^n \widehat{f}_{m,0}(Z_i)\right] b_m\right)
\end{align*}
or a cross-fitted variant of this estimator. In practice, $M$ can be selected via cross-validation \citep{luedtke2023one}. If the nuisances are all compatible and cross-fitting is used, then this estimator is exactly a regularized one-step estimator as presented in \cite{luedtke2023one}, and so the Hilbert-norm convergence guarantees from that work can be immediately applied here. More generally, only minor adaptations of the arguments from that work are required to provide guarantees for this estimator. Those guarantees will typically yield a rate of convergence that is slower than $n^{-1/2}$, which occurs because of a bias-variance tradeoff that arises when choosing $M$---see \cite{luedtke2023one} for details. 
The outputs of Algorithm~\ref{alg:estimatedBackprop} can also be used to construct confidence sets for a Hilbert-valued $\psi(P)$ --- see Section~5.2 of \cite{luedtke2023one}.

While the above describes a full algorithm for computing $\widehat{\psi}_{M}$, we emphasize that this algorithm is not currently implemented in the proof-of-concept dimple package. We leave its implementation to future work.

\section{Total pathwise differentiability of primitives from Table~\ref{tab:primitives}}\label{app:primitives}

\subsection{Overview}

Appendix~\ref{app:primitiveTPD} presents primitives that depend nontrivially on both of their arguments. Appendices~\ref{app:primitiveOnlyDist} and \ref{app:onlyHilbert} present primitives that depend only on their distribution-valued or Hilbert-valued argument, respectively. For each primitive, we present the form of the primitive and conditions under which it is totally pathwise differentiable. When not specified, $\mathcal{U}$ is taken to be a generic subset of a Hilbert space $\mathcal{T}$ and $\mathcal{W}$ is taken to be a Hilbert space that is a superset of the image of the map $\theta$ defined on $\mathcal{M}\times\mathcal{U}$.

For all primitives presented in this appendix, we suppose the model $\mathcal{M}$ is locally nonparametric at the distribution $P$ at which we wish to establish total pathwise differentiability. This means that the tangent space $\dot{\mathcal{M}}_P$ is equal to the collection of all $s\in L^2(P)$ satisfying $\int s\, dP=0$.

When establishing the differentiability of certain primitives, we suppose that all distributions in $\mathcal{M}$ are equivalent, which means that all pairs of distributions in $\mathcal{M}$ are mutually absolutely continuous. In these cases, we let $\rho$ denote an arbitrarily selected distribution in $\mathcal{M}$. When $X=\mathscr{C}(Z)$ for a coarsening $\mathscr{C} : \mathcal{Z}\rightarrow\mathcal{X}$, we let $\rho_X$ be the pushforward measure $\rho\circ\mathscr{C}^{-1}$. We define $L^0(\rho)$ (resp., $L^\infty(\rho)$) as the set of all $\rho$-a.s. equivalence classes of measurable (resp., bounded) $\mathcal{Z}\rightarrow\mathbb{R}$ functions; $L^0(\rho_X)$ and $L^\infty(\rho_X)$ are defined analogously. The measure $\rho$ is used to emphasize the ambient nature of the Hilbert spaces $\mathcal{T}$ and $\mathcal{W}$. For example, if $\mathcal{U}\subseteq L^\infty(\rho)$, $L^\infty(\rho)$ is a codomain for $\theta$, and  $\mathcal{T}=\mathcal{W}=L^2(Q)$ for $Q\in\mathcal{M}$, then different choices of $Q$ can be considered for the ambient Hilbert spaces $\mathcal{T}$ and $\mathcal{W}$, without impacting the form of the map $\theta$.

\subsection{Maps depend nontrivially on both of their arguments}\label{app:primitiveTPD}

\subsubsection{Conditional mean}\label{app:condExp}
This example is a special case of the $r$-fold conditional mean studied in Appendix~\ref{app:multilinearForm} with $r=1$. We refer readers there for details.

\subsubsection{Multifold conditional mean}\label{app:multilinearForm}

Fix $r\in\mathbb{N}$ and $Q\in\mathcal{M}$. Suppose all distributions in $\mathcal{M}$ are equivalent. To denote a generic tuple $(a_1,a_2,\ldots,a_r)$ or random variable $(A_1,A_2,\ldots,A_r)$, we write $a_{[r]}$ or $A_{[r]}$, respectively. For any probability measure $\lambda$ on $\mathcal{Z}$ (e.g., $\lambda=Q,P$), let $\lambda^r$ and $\lambda_X^r$ denote the $r$-fold product measures representing the distributions of $r$ independent draws from $\lambda$ and $\lambda_X$, respectively. We write $\lambda^r(\,\cdot\mymid x_{[r]})$ for the conditional distribution of $Z_{[r]}\sim \lambda^r$ given that $X_{[r]}=x_{[r]}$, where, for $j\in [r]$, $X_j:=\mathscr{C}(Z_j)$ for a coarsening $\mathscr{C} : \mathcal{Z}\rightarrow\mathcal{X}$. We study the map $\theta : \mathcal{M}\times\mathcal{U}\rightarrow L^2(Q_X^r)$ defined so that
\begin{align}
    \theta(P,u) : x_{[r]}\mapsto \int u(z_{[r]})\, P^r(dz_{[r]}\mymid x_{[r]}), \label{eq:multilinearForm}
\end{align}
where $x_{[r]}:=(\mathscr{C}(z_j))_{j=1}^r$ and, for some function $u^\star : \mathcal{Z}^r\rightarrow(0,\infty)$ satisfying
\begin{align}
    \sup_{P\in\mathcal{M}}\|E_{P^r}[u^\star(Z_{[r]})^2\mymid X_{[r]}=\cdot\,]\|_{L^\infty(P_X^r)}<\infty, \label{eq:ustarBd}
\end{align}
$\mathcal{U}$ is the set of (equivalence classes of) functions dominated by $u^\star$, namely $\{u\in L^0(\rho^r) : |u(z_{[r]})|\le u^\star(z_{[r]})\;\,Q^r\textnormal{-a.s.}\}\subset L^2(Q^r)$. 
The conditional second-moment condition in \eqref{eq:ustarBd} is satisfied if $u^\star$ is bounded, but it can also hold for unbounded $u^\star$.
When $\mathscr{C}$ is the zero map, so $X_j$ is a.s. constant, $\theta(P,u)$ is simply the estimand pursued by a U- or V-statistic with kernel $u$ \citep{hoeffding1948class}; more generally, $\theta(P,u)$ is a conditional variant of this quantity.

\begin{lemma}[Total pathwise differentiability of multifold conditional mean]\label{lem:UstatDiff}
    Fix $(P,u)\in\mathcal{M}\times\mathcal{U}$ and suppose $\|dQ_X/dP_X\|_{L^\infty(P_X)}<\infty$ and $\|dP/dQ\|_{L^\infty(Q)}<\infty$. The parameter $\theta$ defined in \eqref{eq:multilinearForm} is totally pathwise differentiable at $(P,u)$ with $\dot{\theta}_{P,u}(s,t) = \dot{\nu}_{P,u}(s) + \dot{\zeta}_{P,u}(t)$ and $\dot{\theta}_{P,u}^*(w)=(\dot{\nu}_{P,u}^*(w),\dot{\zeta}_{P,u}^*(w))$, where
    \begin{align}
        \dot{\nu}_{P,u}(s)(x_{[r]})&= \sum_{i=1}^r E_{P^r}\left[\{u(Z_{[r]})-\theta(P,u)(x_{[r]})\}\,s(Z_i)\middle| X_{[r]}=x_{[r]}\right], \label{eq:multifoldLocalP} \\
        \dot{\nu}_{P,u}^*(w)(z)&= \sum_{i=1}^r E_{P^r}\left[\frac{dQ_X^r}{dP_X^r}(X_{[r]})\left[u(Z_{[r]}) - \theta(P,u)(X_{[r]})\right]w(X_{[r]})\,\middle|\,Z_i=z\right], \label{eq:multifoldEIO} \\
        \dot{\zeta}_{P,u}(t)(x_{[r]})&= E_{P^r}[t(Z_{[r]})\mid X_{[r]}=x_{[r]}], \\
       \dot{\zeta}_{P,u}^*(w)(z_{[r]})&= \frac{dP^r}{dQ^r}(z_{[r]})\frac{dQ_X^r}{dP_X^r}(x_{[r]})w(x_{[r]}). \label{eq:multifoldZetaAdj}
    \end{align}
\end{lemma}
\begin{proof}[Proof of Lemma~\ref{lem:UstatDiff}]
    Throughout this proof, any unsubscripted norm $\|\cdot\|$ denotes the $L^2(Q_X^r)$ norm, and we write $m$ to denote the quantity on the left-hand side of \eqref{eq:ustarBd}. Fix $s\in\check{\mathcal{M}}_P$, $t\in \check{\mathcal{U}}_u$, $\{P_\epsilon : \epsilon\in [0,1]\}\in \mathscr{P}(P,\mathcal{M},s)$, and $\{u_\epsilon : \epsilon\in [0,1]\}\in \mathscr{P}(u,\mathcal{U},t)$. By the triangle inequality,
    \begin{align*}
    &\left\|\theta(P_\epsilon,u_\epsilon)-\theta(P,u)-\epsilon \dot{\nu}_{P,u}(s)-\epsilon\dot{\zeta}_{P,u}(t)\right\| \\
    &\le \epsilon \left\|\dot{\nu}_{P,u_\epsilon}(s)-\dot{\nu}_{P,u}(s)\right\| + \left\|\theta(P,u_\epsilon)-\theta(P,u)-\epsilon\dot{\zeta}_{P,u}(t)\right\| + \left\|\theta(P_\epsilon,u_\epsilon)-\theta(P,u_\epsilon)-\epsilon \dot{\nu}_{P,u_\epsilon}(s)\right\|. 
    \end{align*}
    Combining this with the inequality $(a+b+c)^2\le 3(a^2+b^2+c^2)$ yields
    \begin{align}
    &\frac{1}{3}\left\|\theta(P_\epsilon,u_\epsilon)-\theta(P,u)-\epsilon \dot{\nu}_{P,u}(s)-\epsilon\dot{\zeta}_{P,u}(t)\right\|^2 \label{eq:UstatTpd} \\
    &\le \epsilon^2 \left\|\dot{\nu}_{P,u_\epsilon}(s)-\dot{\nu}_{P,u}(s)\right\|^2 + \left\|\theta(P,u_\epsilon)-\theta(P,u)-\epsilon\dot{\zeta}_{P,u}(t)\right\|^2 + \left\|\theta(P_\epsilon,u_\epsilon)-\theta(P,u_\epsilon)-\epsilon \dot{\nu}_{P,u_\epsilon}(s)\right\|^2. \nonumber
    \end{align}
    We will denote the three terms on the right by (I), (II), and (III). In what follows we shall show that each of these terms is $o(\epsilon^2)$. When combined with the fact that $(s,t)\mapsto \dot{\nu}_{P,u}(s) + \dot{\zeta}_{P,u}(t)$ is a bounded linear operator, this will establish that $\theta$ is totally pathwise differentiable with $\dot{\theta}_{P,u}(s,t) = \dot{\nu}_{P,u}(s) + \dot{\zeta}_{P,u}(t)$. Once we have established total pathwise differentiability, we shall conclude this proof by deriving the form of the adjoint $\dot{\theta}_{P,u}^*$. \\[.5em]
    \textbf{Proof that (I) is $o(\epsilon^2)$:} Let $\bar{s}(z_{[r]}):=\sum_{j=1}^r s(z_j)$ and note that, by Cauchy-Schwarz, the fact that conditional variances are upper bounded by conditional second moments, and a change of measure,
    \begin{align}
        \textnormal{(I)}&= \epsilon^2 \left\|\dot{\nu}_{P,u_\epsilon}(s)-\dot{\nu}_{P,u}(s)\right\|^2 \nonumber \\
        &= \epsilon^2 \int \left(E_{P^r}\left[\{u(Z_{[r]})-\theta(P,u)(x_{[r]})-u_\epsilon(Z_{[r]})+\theta(P,u_\epsilon)(x_{[r]})\}\bar{s}(Z_{[r]})\middle| X_{[r]}=x_{[r]}\right]\right)^2 Q_X^r (dx_{[r]}) \nonumber \\
        &\le \epsilon^2 \int E_{P^r}\left[\{u(Z_{[r]})-\theta(P,u)(x_{[r]})-u_\epsilon(Z_{[r]})+\theta(P,u_\epsilon)(x_{[r]})\}^2\middle| X_{[r]}=x_{[r]}\right] \nonumber \\
        &\hspace{3em}\cdot E_{P^r}\left[\bar{s}(Z_{[r]})^2\middle| X_{[r]}=x_{[r]}\right] Q_X^r (dx_{[r]}) \nonumber \\
        &\le \epsilon^2 \int E_{P^r}\left[\{u(Z_{[r]})-u_\epsilon(Z_{[r]})\}^2\middle| X_{[r]}=x_{[r]}\right]E_{P^r}\left[\bar{s}(Z_{[r]})^2\middle| X_{[r]}=x_{[r]}\right] Q_X^r (dx_{[r]}) \nonumber \\
        &= \epsilon^2 \int E_{P^r}\left[\{u(Z_{[r]})-u_\epsilon(Z_{[r]})\}^2\middle| X_{[r]}=x_{[r]}\right] E_{P^r}\left[\bar{s}(Z_{[r]})^2\middle| X_{[r]}=x_{[r]}\right] \frac{dQ_X^r}{dP_X^r}(x_{[r]}) P_X^r(dx_{[r]}). \label{eq:termIUStat}
    \end{align}
    We will use the dominated convergence theorem to show that the integral on the right goes to zero as $\epsilon\rightarrow 0$. Using the inequality $(a+b)^2\le 2(a^2+b^2)$, the fact that $|u(\cdot)|$ and $|u_\epsilon(\cdot)|$ are upper bounded by $u^\star(\cdot)$ $Q^r$-a.s. (and, by the equivalence of $P^r$ and $Q^r$, also $P^r$-a.s.), and the definition of $m$, we see that the magnitude of the integrand is a.s. upper bounded pointwise by $F(x_{[r]}):=4m\|dQ_X/dP_X\|_{L^\infty(P_X)}^r E_{P^r}[\bar{s}(Z_{[r]})^2\mymid X_{[r]}=x_{[r]}]$. Moreover, because $s\in L^2(P)$, $F$ is $P_X^r$-integrable with $\int F(x_{[r]}) P_X^r(dx_{[r]}) = 4 m\|dQ_X/dP_X\|_{L^\infty(P_X)}^r \int \bar{s}(z_{[r]})^2 P^r(dz_{[r]})<\infty$. Hence, the integrand is dominated by an integrable function. To see that the integrand also converges to zero in probability as $\epsilon\rightarrow 0$, we apply Markov's inequality, the law of total expectation, the fact that $\|dP/dQ\|_{L^\infty(Q)}<\infty$, and H\"{o}lder's inequality with $(p,q)=(1,\infty)$ to see that, for any $\delta>0$,
    \begin{align*}
        P_X^r\left\{E_{P^r}\left[\{u(Z_{[r]})-u_\epsilon(Z_{[r]})\}^2\middle| X_{[r]}\right] > \delta\right\}&\le \left\|u-u_\epsilon\right\|_{L^2(P^r)}^2/\delta\le \left\|\frac{dP}{dQ}\right\|_{L^\infty(Q)}^r \left\|u-u_\epsilon\right\|_{L^2(Q^r)}^2/\delta.
    \end{align*}
    The right-hand side goes to zero as $\epsilon\rightarrow 0$ since $\{u_\epsilon : \epsilon\in [0,1]\}\in \mathscr{P}(u,\mathcal{U},t)$. As $\delta>0$ was arbitrary, this shows the integrand on the right-hand side of \eqref{eq:termIUStat} converges to zero in probability as $\epsilon\rightarrow 0$. As that integrand is also dominated by an integrable function, the dominated convergence theorem shows that the right-hand side of \eqref{eq:termIUStat} is $o(\epsilon^2)$, and so $\textnormal{(I)}=o(\epsilon^2)$.\\[.5em]
    \textbf{Proof that (II) is $o(\epsilon^2)$:} Observe that
    \begin{align*}
    \textnormal{(II)}&= \left\|\theta(P,u_\epsilon)-\theta(P,u)-\epsilon\dot{\zeta}_{P,u}(t)\right\|^2 \\
    &= \int \left(E_{P^r}\left[u_\epsilon(Z_{[r]})-u(Z_{[r]})-\epsilon t(Z_{[r]})\,\middle|\,X_{[r]}=x_{[r]}\right]\right)^2 dQ_X^r(x_{[r]}) \\
    &\le \int E_{P^r}\left[\left\{u_\epsilon(Z_{[r]})-u(Z_{[r]})-\epsilon t(Z_{[r]})\right\}^2\,\middle|\,X_{[r]}=x_{[r]}\right] dQ_X^r(x_{[r]}) \\
    &= \int E_{Q^r}\left[\frac{dP^r}{dQ^r}(Z_{[r]})\frac{dQ_X^r}{dP_{X}^r}(X_{[r]})\left\{u_\epsilon(Z_{[r]})-u(Z_{[r]})-\epsilon t(Z_{[r]})\right\}^2\,\middle|\,X_{[r]}=x_{[r]}\right] dQ_X^r(x_{[r]}) \\
    &= \int \frac{dP^r}{dQ^r}(z_{[r]})\frac{dQ_X^r}{dP_{X}^r}(x_{[r]})\left\{u_\epsilon(z_{[r]})-u(z_{[r]})-\epsilon t(z_{[r]})\right\}^2 dQ^r(z_{[r]}) \\
    &\le \left\|\frac{dP}{dQ}\right\|_{L^\infty(Q)}^r\left\|\frac{dQ_X}{dP_{X}}\right\|_{L^\infty(P_X)}^r\left\|u_\epsilon-u-\epsilon t\right\|_{L^2(Q^r)}^2,
    \end{align*}  
    where we used the linearity of conditional expectations, Jensen's inequality, a change of measure, the law of total expectation, H\"{o}lder's inequality, and finally the fact that $P_X$ and $Q_X$ are equivalent measures. The right-hand side is $o(\epsilon^2)$ since $\{u_\epsilon : \epsilon\in [0,1]\}\in \mathscr{P}(u,\mathcal{U},t)$.\\[.5em]
    \textbf{Proof that (III) is $o(\epsilon^2)$:} This part of the proof is an adaptation of the one given for Example~5 in \cite{luedtke2023one}, which establishes that a regression function is pathwise differentiable relative to a locally nonparametric model.

    Let $p_\epsilon^{r/2}(z_{[r]}):=\frac{dP_\epsilon^r}{dP^r}(z_{[r]})^{1/2}$, $p_{\epsilon,X}^{r/2}(x_{[r]}):=\frac{dP_{\epsilon,X}^r}{dP_X^r}(x_{[r]})^{1/2}$, and $p_{\epsilon,Z|X}^{r/2}(z_{[r]}):=p_\epsilon^{r/2}(z_{[r]})/p_{\epsilon,X}^{r/2}(x_{[r]})$. Further define $\bar{s}(z_{[r]}):=\sum_{j=1}^r s(z_j)$, $\bar{s}_X(x_{[r]}):=\sum_{j=1}^r E_P[s(Z_j)\mid X=x_j]$, and $\bar{s}_{Z|X}(z_{[r]}):=\bar{s}(z_{[r]})-\bar{s}_X(x_{[r]})$. Observe that, for any $\epsilon$, the following holds for $P_X^r$-almost all $x_{[r]}$:
    \begin{align}
        &[\theta(P_\epsilon,u_\epsilon)-\theta(P,u_\epsilon)-\epsilon \dot{\nu}_{P,u_\epsilon}(s)](x_{[r]}) \\
        &= E_{P^r}\left[u_\epsilon(Z_{[r]})\left\{[p_{\epsilon,Z|X}^{r/2}(Z_{[r]})+1][p_{\epsilon,Z|X}^{r/2}(Z_{[r]})-1] - \epsilon \bar{s}_{Z\mid X}(Z_{[r]})\right\}\middle|X_{[r]}=x_{[r]}\right] \nonumber \\
        &= E_{P^r}\left[u_\epsilon(Z_{[r]})\left\{[p_{\epsilon,Z|X}^{r/2}(Z_{[r]})+1]\left[p_{\epsilon,Z|X}^{r/2}(Z_{[r]})-1-\frac{\epsilon}{2}\bar{s}_{Z\mid X}(Z_{[r]}) \right]\right\}\middle|X_{[r]}=x_{[r]}\right] \nonumber \\
        &\quad+ \frac{\epsilon}{2} E_{P^r}\left[u_\epsilon(Z_{[r]})\bar{s}_{Z\mid X}(Z_{[r]})\left[p_{\epsilon,Z|X}^{r/2}(Z_{[r]})-1\right]\middle|X_{[r]}=x_{[r]}\right]. \label{eq:ABdecomp}
    \end{align}
    For shorthand, we refer to the two terms on the right as $A_\epsilon(x_{[r]})$ and $\frac{\epsilon}{2}B_\epsilon(x_{[r]})$. Combining the above with the basic inequality  $(a+b)^2\le 2(a^2+b^2)$ shows that $\textnormal{(III)}\le 2\|A_\epsilon\|_{L^2(Q_X^r)}^2 + \frac{\epsilon^2}{2}\|B_\epsilon\|_{L^2(Q_X^r)}^2$. In what follows we will show that $\|A_\epsilon\|_{L^2(Q_X^r)}=o(\epsilon)$ and $\|B_\epsilon\|_{L^2(Q_X^r)}=o(1)$, which will thus imply that $\textnormal{(III)}=o(\epsilon^2)$.

   For $A_\epsilon$, we apply Cauchy-Schwarz and the inequality $(a+b)^2\le 2(a^2+b^2)$ to show that, for $P_X^r$-almost all $x_{[r]}$:
    \begin{align*}
        |A_\epsilon(x_{[r]})|^2&\le 2\left(E_{P_\epsilon^r}\left[u_\epsilon(Z_{[r]})^2\middle|X_{[r]}=x_{[r]}\right]+E_{P^r}\left[u_\epsilon(Z_{[r]})^2\middle|X_{[r]}=x_{[r]}\right]\right) \\
        &\quad\cdot E_{P^r}\left[\left\{p_{\epsilon,Z|X}^{r/2}(Z_{[r]})-1-\frac{\epsilon}{2}\bar{s}_{Z\mid X}(Z_{[r]})\right\}^2\middle|X_{[r]}=x_{[r]}\right].
    \end{align*}
    Using that $u_\epsilon^2\le (u^\star)^2$ and the definition of $m$ shows that each of the first two conditional expectations above is upper bounded by $m$ $P^r$-almost surely. Plugging in this bound and integrating both sides above against $P_X^r$ shows that
    \begin{align*}
        \left\|A_\epsilon\right\|_{L^2(P_X^r)}^2&\le 4m E_{P^r}\left[\left\{p_{\epsilon,Z|X}^{r/2}(Z_{[r]})-1-\frac{\epsilon}{2}\bar{s}_{Z\mid X}(Z_{[r]})\right\}^2\right].
    \end{align*}
    Using that $\{P_\epsilon : \epsilon\}$ is quadratic mean differentiable and applying calculations akin to those used in \cite{luedtke2023one} to establish Lemma~S8 in  that work shows that 
    \begin{align}
    E_{P^r}\left[\left\{p_{\epsilon,Z|X}^{r/2}(Z_{[r]})-1-\frac{\epsilon}{2}\bar{s}_{Z\mid X}(Z_{[r]})\right\}^2\right] =o(\epsilon^2), \label{eq:condQMD}
    \end{align}
    and so $ \left\|A_\epsilon\right\|_{L^2(P_X^r)}=o(\epsilon)$ as desired.

     To study $B_\epsilon$, we define
    \begin{align*}
        B_{\epsilon,1}(x_{[r]}):=E_{P^r}\left[I\{|\bar{s}_{Z\mid X}(Z_{[r]})u^\star(Z_{[r]})|\le \epsilon^{-1/2}\}u_\epsilon(Z_{[r]})\bar{s}_{Z\mid X}(Z_{[r]})\left\{p_{\epsilon,Z|X}^{r/2}(Z_{[r]})-1\right\}\middle|X_{[r]}=x_{[r]}\right]
    \end{align*}
    and $B_{\epsilon,2}:=B_\epsilon-B_{\epsilon,1}$. By the triangle inequality, to show  $\|B_\epsilon\|_{L^2(Q_X)}=o(1)$ it suffices to show that $\|B_{\epsilon,j}\|_{L^2(Q_X)}=o(1)$ for $j\in \{1,2\}$. Using that it is $Q^r$-a.s. true that
    \begin{align*}
        &I\{|\bar{s}_{Z\mid X}(z_{[r]})u^\star(z_{[r]})|\le \epsilon^{-1/2}\}\bar{s}_{Z\mid X}(z_{[r]})^2u_\epsilon(z_{[r]})^2\\
        &\quad\le I\{|\bar{s}_{Z\mid X}(z_{[r]})u^\star(z_{[r]})|\le \epsilon^{-1/2}\}\bar{s}_{Z\mid X}(z_{[r]})^2u^\star(z_{[r]})^2\le \epsilon^{-1},
    \end{align*}
    Jensen's inequality, H\"{o}lder's inequality with $(p,q)=(1,\infty)$, and \eqref{eq:condQMD}, we see that
    \begin{align*}
        \|B_{\epsilon,1}\|_{L^2(Q_X^r)}^2&\le \epsilon^{-1}\left\|\frac{dQ_X}{dP_X}{}\right\|_{L^\infty(P_X)}^r \|p_{\epsilon,Z|X}^{r/2}-1\|_{L^2(P)}^2 = O(\epsilon).
    \end{align*}
    By Cauchy-Schwarz, the fact that $u_\epsilon^2\le (u^\star)^2$ $Q^r$-a.s. (and therefore also $P^r$-a.s. and $P_\epsilon^r$-a.s.), and the inequality $(a-b)^2\le 2(a^2+b^2)$, the following holds for $P_X^r$-almost all $x_{[r]}$:
    \begin{align*}
        |B_{\epsilon,2}(x_{[r]})|^2&\le 2\left(E_{P_\epsilon^r}\left[u^\star(Z_{[r]})^2\middle|X_{[r]}=x_{[r]}\right]+E_{P^r}\left[u^\star(Z_{[r]})^2\middle|X_{[r]}=x_{[r]}\right]\right) \\
        &\quad\cdot E_{P^r}\left[I\{|\bar{s}_{Z\mid X}(Z_{[r]})u^\star(Z_{[r]})|> \epsilon^{-1/2}\}\bar{s}_{Z\mid X}(Z_{[r]})^2\middle|X_{[r]}=x_{[r]}\right].
    \end{align*}
    Using \eqref{eq:ustarBd}, integrating both sides against $Q_X^r$, and applying H\"{o}lder's inequality with $(p,q)=(1,\infty)$ shows that 
    \begin{align*}
        \|B_{\epsilon,2}\|_{L^2(Q_X^r)}^2\le 4m\left\|\frac{dQ_X}{dP_X}\right\|_{L^\infty(P_X)}^rE_{P^r}\left[I\{|\bar{s}_{Z\mid X}(Z_{[r]})u^\star(Z_{[r]})|> \epsilon^{-1/2}\}\bar{s}_{Z\mid X}(Z_{[r]})^2\right].
    \end{align*}
    By the dominated convergence theorem, the expectation on the right is $o(1)$. Combining our study of $B_{\epsilon,1}$ and $B_{\epsilon,2}$, we have shown that $\|B_\epsilon\|_{L^2(Q_X)}=o(1)$. Returning to our discussion below \eqref{eq:ABdecomp}, we have shown that $\textnormal{(III)}=o(\epsilon^2)$.\\[.5em]
    \textbf{Derivation of $\dot{\theta}_{P,u}^*$.} Having established the total pathwise differentiability of $\theta$ at $(P,u)$, we now establish the claimed form of the adjoint $\dot{\theta}_{P,u}^*$. It can be verified that, at any $u$, the tangent space of $\mathcal{U}$ satisfies $\dot{\mathcal{U}}_u= L^2(Q^r)$. Moreover, for any $w\in L^2(Q_X^r)$, Lemma~\ref{lem:totalImpliesPartial} shows that $\dot{\theta}_{P,u}^*(w)=(\dot{\nu}_{P,u}^*(w),\dot{\zeta}_{P,u}^*(w))$, where $\dot{\nu}_{P,u}^*$ and $\dot{\zeta}_{P,u}^*$ are the adjoints of $\dot{\nu}_{P,u}$ and $\dot{\zeta}_{P,u}$. 
    To derive the claimed form of $\dot{\nu}_{P,u}^*$, observe that, for any $s\in\dot{\mathcal{M}}_P$ and $w\in L^2(Q_X^r)$,
    \begin{align*}
    \left\langle \dot{\nu}_{P,u}(s), w\right\rangle_{L^2(Q_X^r)}&= \left\langle \sum_{i=1}^r 
    E_{P^r}\left[\left\{u(Z_{[r]}) - \underline{\nu}(P^r)(X_{[r]})\right\}s(Z_i)\,\middle|\,X_{[r]}=\,\cdot\,\right],w\right\rangle_{L^2(Q_X^r)} \\
     &= E_{P^r}\left[\sum_{i=1}^r \frac{dQ_X^r}{dP_X^r}(X_{[r]})\left\{u(Z_{[r]}) - \underline{\nu}(P^r)(X_{[r]})\right\}w(X_{[r]})s(Z_i)\right] \\
     &= \int \left(\sum_{i=1}^r E_{P^r}\left[\frac{dQ_X^r}{dP_X^r}(X_{[r]})\left[u(Z_{[r]}) - \underline{\nu}(P^r)(X_{[r]})\right]w(X_{[r]})\,\middle|\,Z_i=z\right]\right)s(z) P(dz).
    \end{align*}
    The right-hand side equals $\langle \dot{\nu}_P^*(w),s\rangle_{L^2(P)}$, where $\dot{\nu}_P^*$ takes the form in \eqref{eq:multifoldEIO}. As for the claimed form of $\dot{\zeta}_{P,u}^*$, note that, for all $t\in \dot{\mathcal{U}}_u$ and $w\in L^2(Q_X^r)$,
    \begin{align*}
        \left\langle \dot{\zeta}_{P,u}(t),w\right\rangle_{L^2(Q_X^r)}&= \int E_{P^r}\left[t(Z_{[r]})\mid X_{[r]}=x_{[r]}\right]\,w(x_{[r]})\, Q_X^r(dx_{[r]})  \\
        &= E_{Q^r}\left[\frac{dP^r}{dQ^r}(Z_{[r]})\frac{dQ_X^r}{dP_X^r}(X_{[r]})w(X_{[r]})t(Z_{[r]})\right].
    \end{align*}
    The right-hand side is equal to $\langle \dot{\zeta}_{P,u}^*(w),t\rangle_{L^2(Q^r)}$ with $\dot{\zeta}_{P,u}^*(w)$ as defined in \eqref{eq:multifoldZetaAdj}.
\end{proof}

\subsubsection{Conditional covariance}\label{app:condCovar}
Fix $Q\in \mathcal{M}$ and suppose all distributions in $\mathcal{M}$ are equivalent. 
Define the conditional covariance map $\theta : \mathcal{M}\times\mathcal{U}\rightarrow L^2(Q_X)$ as
\begin{align}
    \theta(P,u)(x)= \mathrm{cov}_P[\,u_1(Z),u_2(Z)\mymid X=x], \label{eq:condCovar}
\end{align}
where $\mathcal{U}:=\{u:=(u_1,u_2)\in L^0(\rho)^2 : \max_{j\in \{1,2\}}\|u_j\|_{L^\infty(Q)}\le m\}\subset L^2(Q)\oplus L^2(Q)$ for some fixed $m<\infty$.

\begin{lemma}[Total pathwise differentiability of conditional covariance]\label{lem:condCovar}
    Fix $(P,u)\in\mathcal{M}\times\mathcal{U}$ and suppose $\|dQ_X/dP_X\|_{L^\infty(P_X)}<\infty$ and $\|dP/dQ\|_{L^\infty(Q)}<\infty$. The parameter $\theta$ defined in \eqref{eq:condCovar} is totally pathwise differentiable at $(P,u)$ with $\dot{\theta}_{P,u}^*(w)=(\dot{\nu}_{P,u}^*(w),\dot{\zeta}_{P,u}^*(w))$, where
    \begin{align*}
        \dot{\nu}_{P,u}^*(w)(z)&= \frac{dQ_X}{dP_X}(x)w(x)\left[\prod_{j=1}^2 \{u_j(z)-E_P[u_j(Z)\mymid X=x]\}-\theta(P,u)(x)\right], \\
       \dot{\zeta}_{P,u}^*(w)&= \left(z\mapsto \frac{dP}{dQ}(z)\frac{dQ_X}{dP_X}(x)w(x)\{u_{3-j}(z)-E_P[u_{3-j}(Z)\mymid X=x]\}\right)_{j=1}^2.
    \end{align*}
\end{lemma}
The form of the differential operator $\dot{\theta}_{P,u}$ is given in the proof.
\begin{proof}[Proof of Lemma~\ref{lem:condCovar}]
    For generic $(P',(u_1',u_2'))\in\mathcal{M}\times\mathcal{U}$, let $u_1'u_2' : z\mapsto u_1'(z)u_2'(z)$ and, for $\underline{u}:\mathcal{Z}\rightarrow[-m^2,m^2]$, let  $\underline{\theta}(P',\underline{u}):=E_{P'}[\underline{u}(Z)\mymid X=\cdot\,]$. Observe that $\theta(P',u')=\underline{\theta}(P',u_1'u_2')-\underline{\theta}(P',u_1')\underline{\theta}(P',u_2')$.
    
    Fix $s\in\check{\mathcal{M}}_P$, $t\in\check{\mathcal{U}}_u$, $\{P_\epsilon : \epsilon\}\in \mathscr{P}(P,\mathcal{M},s)$, and $\{u_\epsilon=(u_{\epsilon,1},u_{\epsilon,2}) : \epsilon\}\in \mathscr{P}(u,\mathcal{U},t)$. The total pathwise differentiability of the $\mathcal{U}\rightarrow L^2(Q)$ pointwise product $(u_1,u_2)\mapsto u_1u_2$ (Appendix~\ref{app:pointwise}) implies that $\{z\mapsto u_{\epsilon,1}(z)u_{\epsilon,2}(z) : \epsilon\}\in \mathscr{P}(u_1u_2,L^2(Q),t_1u_2+u_1t_2)$. Moreover, the total pathwise differentiability of coordinate projections (Lemma~\ref{lem:coordProj}) implies that $\{z\mapsto u_{\epsilon,j}(z) : \epsilon\}\in \mathscr{P}(u_j,L^2(Q),t_j)$, $j\in \{1,2\}$.
    When combined with the study of the conditional mean operator in Lemma~\ref{lem:UstatDiff}, this shows that
    \begin{align*}
        &\underline{\theta}(P_\epsilon,u_{\epsilon,1}u_{\epsilon,2})(\cdot) - \underline{\theta}(P,u_1 u_2)(\cdot) \\
        &\hspace{3em}- \epsilon E_P\left[t_{1}(Z)u_{2}(Z) + u_{1}(Z)t_{2}(Z) + \{u_1(Z)u_2(Z) - \underline{\theta}(P,u_1 u_2)(X)\} s(Z)\mymid X=\,\cdot\,\right] = o(\epsilon), \\
        &\underline{\theta}(P_\epsilon,u_{\epsilon,j})(\cdot) - \underline{\theta}(P,u_j)(\cdot) - \epsilon E_P\left[t_{j}(Z) + \{u_j(Z)-\underline{\theta}(P,u_j)(X)\}s(Z)\mymid X=\,\cdot\,\right] = o(\epsilon),\ j\in \{1,2\},
    \end{align*}
    where the $o(\epsilon)$ terms converge to zero in $L^2(Q_X)$ faster than $\epsilon$. Finally, applying the pointwise operation $(a,b,c)\mapsto a-bc$ and leveraging results from Appendix~\ref{app:pointwise} shows that $\theta(P_\epsilon,u_\epsilon)-\theta(P,u)-\epsilon\, \dot{\theta}_{P,u}(s,t)=o(\epsilon)$, where
    \begin{align*}
        \dot{\theta}_{P,u}(s,t)(\cdot)&:= E_P\left[t_{1}(Z)u_{2}(Z) + u_{1}(Z)t_{2}(Z) + \{u_1(Z)u_2(Z) - \underline{\theta}(P,u_1 u_2)\} s(Z)\mymid X=\,\cdot\,\right] \\
        &\quad\quad- \sum_{j=1}^2 E_P[u_{3-j}(Z)\mymid X=\cdot\,]E_P\left[t_{j}(Z) + \{u_j(Z)-\underline{\theta}(P,u_j)(X)\}s(Z)\mymid X=\,\cdot\,\right].
    \end{align*}
    As $(s,t)$ were arbitrary and $\dot{\theta}_{P,u}$ is bounded and linear, $\theta$ is totally pathwise differentiable at $(P,u)$ with differential operator $\dot{\theta}_{P,u}$.

    As for the form of the adjoint, fix $(s,t)\in \dot{\mathcal{M}}_P\oplus \dot{\mathcal{U}}_u$ and $w\in L^2(Q_X)$. The law of total expectation and definition of the inner product on $\dot{\mathcal{M}}_P\oplus \dot{\mathcal{U}}_u$ can be used to verify that
    \begin{align*}
        \langle \dot{\theta}_{P,u}(s,t),w\rangle_{L^2(Q_X)}&= \left\langle \frac{dQ_X}{dP_X}\dot{\theta}_{P,u}(s,t),w\right\rangle_{L^2(P_X)} = \left\langle (s,t),\dot{\theta}_{P,u}^*(w)\right\rangle_{\dot{\mathcal{M}}_P\oplus \dot{\mathcal{U}}_u},
    \end{align*}
    with $\dot{\theta}_{P,u}^*(w)$ as defined in the lemma statement.
\end{proof}

\subsubsection{Conditional variance}\label{app:condVar}
Fix $Q\in \mathcal{M}$ and suppose all distributions in $\mathcal{M}$ are equivalent. 
Define the conditional variance map $\theta : \mathcal{M}\times\mathcal{U}\rightarrow L^2(Q_X)$ as
\begin{align}
    \theta(P,u)(x)= \mathrm{Var}_P[\,u(Z)\mymid X=x], \label{eq:condVar}
\end{align}
where $\mathcal{U}:=\{u\in L^0(\rho) : \|u_j\|_{L^\infty(Q)}\le m\}\subset L^2(Q)$ for some fixed $m<\infty$.

\begin{lemma}[Total pathwise differentiability of conditional variance]\label{lem:condVar}
    Fix $(P,u)\in\mathcal{M}\times\mathcal{U}$ and suppose $\|dQ_X/dP_X\|_{L^\infty(P_X)}<\infty$ and $\|dP/dQ\|_{L^\infty(Q)}<\infty$. The parameter $\theta$ defined in \eqref{eq:condVar} is totally pathwise differentiable at $(P,u)$ with $\dot{\theta}_{P,u}^*(w)=(\dot{\nu}_{P,u}^*(w),\dot{\zeta}_{P,u}^*(w))$, where
    \begin{align}
        \dot{\nu}_{P,u}^*(w)(z)&= \frac{dQ_X}{dP_X}(x)w(x)\left[\{u(z)-E_P[u(Z)\mymid X=x]\}^2-\theta(P,u)(x)\right], \\
       \dot{\zeta}_{P,u}^*(w)(z)&= 2\frac{dP}{dQ}(z)\frac{dQ_X}{dP_X}(x)w(x)\{u(z)-E_P[u(Z)\mymid X=x]\}.
    \end{align}
\end{lemma}
The proof of this lemma is nearly identical to that of Lemma~\ref{lem:condCovar} and so is omitted.

\subsubsection{Kernel embedding}\label{app:kernelEmbed}

Let $K : \mathcal{X}\times\mathcal{X}\rightarrow\mathbb{R}$ be a bounded, symmetric, positive-definite kernel on a set $\mathcal{X}$. For example, $K$ may be a Gaussian or Laplace kernel on $\mathbb{R}^d$ \citep{sriperumbudur2011universality}. Let $\mathcal{W}$ denote the unique reproducing kernel Hilbert space (RKHS) associated with $K$, which exists by the Moore-Aronszajn theorem \citep[][Theorem 3]{berlinet2011reproducing}. For a coarsening $\mathscr{C} : \mathcal{Z}\rightarrow\mathcal{X}$ and a generic $P\in \mathcal{M}$, we let $P_X$ denote the marginal distribution of $X:=\mathscr{C}(Z)$ when $Z\sim P$. We suppose all distributions in $\mathcal{M}$ are equivalent.

We study the kernel embedding map $\theta : \mathcal{M}\times\mathcal{U}\rightarrow \mathcal{W}$ that is defined so that
\begin{align}
    \theta(P,u)(\cdot)&:= \int K(\,\cdot\,,x) \,u(x)\, P_X(dx), \label{eq:kernelEmbed}
\end{align}
where, for some $m\in (0,\infty)$ and $Q_X\in \{P_X : P\in\mathcal{M}\}$, $\mathcal{U}:=\{u\in L^2(\rho_X) : \|u\|_{L^\infty(\rho_X)}\le m\}\subset L^2(Q_X)$ \citep[Chapter 4.9.1.1 of][]{berlinet2011reproducing}. The integral above is a Bochner integral; we show this integral is well-defined in Lemma~\ref{lem:rkhsEmbed}. That lemma also shows $\theta(P,u)$ is the unique element of $\mathcal{W}$ that satisfies $\left\langle \theta(P,u), w\right\rangle_{\mathcal{W}}= \int w(x) u(x) P_X(dx)$ for all $w\in\mathcal{W}$. When $u(x)=1$ $P_X$-a.s., $\theta(P,u)$ is the kernel mean embedding of $P_X$ \citep{gretton2012kernel}.

\begin{lemma}[Total pathwise differentiability of kernel embeddings]\label{lem:kernelEmbed}
    Fix $(P,u)\in\mathcal{M}\times\mathcal{U}$ and suppose $\|dP_X/dQ_X\|_{L^\infty(Q_X)}<\infty$. The parameter $\theta$ defined in \eqref{eq:kernelEmbed} is totally pathwise differentiable at $(P,u)$ with $\dot{\theta}_{P,u}(s,t)(\cdot)= E_P\left[K(\,\cdot\,,X)\{u(X)s(Z)+t(X)\}\right]$ and $\dot{\theta}_{P,u}^*(w)=(\dot{\nu}_{P,u}^*(w),\dot{\zeta}_{P,u}^*(w))$, where
    \begin{align}
       \dot{\nu}_{P,u}^*(w)(z)&=  w(x)u(x)-E_P[w(X)u(X)]\hspace{1em}\textnormal{ with }x:=\mathscr{C}(z), \nonumber \\
       \dot{\zeta}_{P,u}^*(w)(x)&= w(x)\frac{dP_X}{dQ_X}(x). \label{eq:rkhsEmbedAdjointDef}
    \end{align}
\end{lemma}
We note that the condition that $\|dP_X/dQ_X\|_{L^\infty(Q_X)}<\infty$ is trivially satisfied if $Q_X=P_X$.
\begin{proof}[Proof of Lemma~\ref{lem:kernelEmbed}]
    Denote the claimed differential operator by $\dot{\Theta}_{P,u}$ so that, for any $(s,t)\in \dot{\mathcal{M}}_P\oplus\dot{\mathcal{U}}_u$, $\dot{\Theta}_{P,u}(s,t)(\cdot)=E_P\left[K(\,\cdot\,,X)\{u(X)s(Z)+t(X)\}\right]$. This proof is broken into the following parts:
    \begin{enumerate}[label=\textbf{Part~\arabic*)},ref=Part~\arabic*,leftmargin=*]
        \item\label{it:rkhsimageW} For all $(s,t)\in \dot{\mathcal{M}}_P\oplus\dot{\mathcal{U}}_u$, $\dot{\Theta}_{P,u}(s,t)\in\mathcal{W}$; consequently, $\mathcal{W}$ is a codomain of $\dot{\Theta}_{P,u}$.
        \item\label{it:rkhsBL} $\dot{\Theta}_{P,u} : \dot{\mathcal{M}}_P\oplus\dot{\mathcal{U}}_u\rightarrow\mathcal{W}$ is bounded and linear.
        \item\label{it:rkhsAdjoint} The Hermitian adjoint of $\dot{\Theta}_{P,u}$ is $w\mapsto (\dot{\nu}_{P,u}^*(w),\dot{\zeta}_{P,u}^*(w))$, with $\dot{\nu}_{P,u}^*$ and $\dot{\zeta}_{P,u}^*$ as defined in \eqref{eq:rkhsEmbedAdjointDef}.
        \item\label{it:rkhsTPD} The map $\theta$ is totally pathwise differentiable at $(P,u)$ with differential operator $\dot{\Theta}_{P,u}$.
    \end{enumerate}
    Taken together, these four parts of the proof show that $\theta$ is indeed totally pathwise differentiable with the claimed differential operator and adjoint, completing the proof.\\[.5em]
    \textbf{\ref{it:rkhsimageW} of proof:} We begin by showing that the image of $\dot{\Theta}_{P,u}$ is a subset of $\mathcal{W}$.  It is straightforward to verify that $\dot{\mathcal{M}}_P=\{s\in L^2(P) : E_P[s(Z)]=0\}$ and $\dot{\mathcal{U}}_u= L^2(Q_X)$. Hence, fixing $(s,t)\in \dot{\mathcal{M}}_P\oplus \dot{\mathcal{U}}_u$ and defining $s_X(x):=E_P[s(Z)\mymid X=x]$, \ref{it:rkhsEmbedWellDef} from Lemma~\ref{lem:rkhsEmbed} shows it suffices to show $g\in L^2(P_X)$, where
    \begin{align}
        g : x\mapsto u(x)s_X(x)+t(x). \label{eq:gDef}
    \end{align}
    Indeed, using that $P_X\ll Q_X$, we see that $g$ is uniquely defined up to $P_X$-null sets; moreover, using that $(a+b)^2\le 2(a^2+b^2)$, the bound on $u$, Jensen's inequality, a change of measure, and H\"{o}lder's inequality with exponents $(p,q)=(\infty,1)$,
    \begin{align}
        \int g(x)^2 P_X(dx)&\le 2\int \left\{u(x)^2E_P[s(Z)\mymid X=x]^2 + t(X)^2\right\}P_X(dx) \nonumber \\
        &\le 2\left(m^2\|s\|_{L^2(P)}^2 + \|t\|_{L^2(Q_X)}^2\left\|dP_X/dQ_X\right\|_{L^\infty(Q_X)}\right) < \infty. \label{eq:gSqInt}
    \end{align}
    Hence, $g\in L^2(P_X)$, and so $\mathcal{W}$ is indeed a codomain of our claimed differential operator $\dot{\Theta}_{P,u}$.\\[.5em]
    \textbf{\ref{it:rkhsBL} of proof:} We now show that $\dot{\Theta}_{P,u} : \dot{\mathcal{M}}_P\oplus\dot{\mathcal{U}}_u\rightarrow \mathcal{W}$ is bounded and linear. Linearity is clear from the form of the operator, so we focus on establishing boundedness. Fix $(s,t)\in \dot{\mathcal{M}}_P\oplus \dot{\mathcal{U}}_u$ and note that, by \ref{it:rkhsEmbedInnerProd} from Lemma~\ref{lem:rkhsEmbed} and letting $s_X(x):=E_P[s(Z)\mymid X=x]$,
    \begin{align*}
        &\|\dot{\Theta}_{P,u}(s,t)\|_{\mathcal{W}}^2 \\
        &= \iint K(x_1,x_2)\prod_{\ell=1}^2\left\{ [u(x_\ell)s_X(x_\ell)+t(x_\ell)]P_X(dx_\ell)\right\} \\
        &\le \iint \prod_{\ell=1}^2 \left[K(x_\ell,x_\ell)^{1/2}|u(x_\ell)s_X(x_\ell)+t(x_\ell)|P_X(dx_\ell)\right] \tag{Cauchy-Schwarz for the kernel $K$} \\
        &\le \sup_{x\in\mathcal{X}} K(x,x)\iint \prod_{\ell=1}^2 \left[|u(x_\ell)s_X(x_\ell)+t(x_\ell)|P_X(dx_\ell)\right]  \tag{H\"{o}lder} \\
        &= \sup_{x\in\mathcal{X}} K(x,x)\left[\int |u(x)s_X(x)+t(x)|P_X(dx)\right]^2 \\
        &\le \sup_{x\in\mathcal{X}} K(x,x)\int \left[u(x)s_X(x)+t(x)\right]^2P_X(dx) \tag{Jensen} \\
        &\le 2\sup_{x\in\mathcal{X}} K(x,x)\left[\int u(x)^2s_X(x)^2P_X(dx)+\int t(x)^2P_X(dx)\right] \tag{$[a+b]^2\le 2[a^2+b^2]$} \\
        &\le 2\sup_{x\in\mathcal{X}} K(x,x)\left[m^2\int s_X(x)^2P_X(dx)+\|dP_X/dQ_X\|_{L^\infty(Q_X)}\int t(x)^2Q_X(dx)\right] \tag{$u\in\mathcal{U}$ and H\"{o}lder} \\
        &\le 2\sup_{x\in\mathcal{X}} K(x,x)\max\{m^2,\|dP_X/dQ_X\|_{L^\infty(Q_X)}\}\left[\int s_X(x)^2P_X(dx)+\int t(x)^2Q_X(dx)\right] \\
        &= 2\sup_{x\in\mathcal{X}} K(x,x)\max\{m^2,\|dP_X/dQ_X\|_{L^\infty(Q_X)}\}\|(s,t)\|_{\dot{\mathcal{M}}_P\oplus\dot{\mathcal{U}}_u}^2.
    \end{align*}
    The right-hand side writes as a constant that does not depend on $(s,t)$ times $\|(s,t)\|_{\dot{\mathcal{M}}_P\oplus\dot{\mathcal{U}}_u}^2$, and so $\dot{\Theta}_{P,u}$ is indeed bounded.\\[.5em]
    \textbf{\ref{it:rkhsAdjoint} of proof:} We now show that $\dot{\Theta}_{P,u}^*$ is the map $w\mapsto (\dot{\nu}_{P,u}^*(w),\dot{\zeta}_{P,u}^*(w))$. To this end, fix $(s,t)\in \dot{\mathcal{U}}_u\oplus \dot{\mathcal{M}}_P$ and $w\in\mathcal{W}$, and observe that
    \begin{align*}
        \langle \dot{\Theta}_{P,u}(s,t),w\rangle_{\mathcal{W}}&= \int w(x) \left[u(x)s_X(x)+t(x)\right] P_X(dx)  \\
        &= \int \left\{w(x) u(x)-E_P[w(X)u(X)]\right\} s_X(x) P_X(dx)+ \int t(x)w(x)\frac{dP_X}{dQ_X}(x) Q_X(dx) \\
        &= \left\langle s, z\mapsto w(x)u(x)-E_P[w(X)u(X)] \right\rangle_{\dot{\mathcal{M}}_P} + \left\langle t, w\frac{dP_X}{dQ_X}\right\rangle_{\dot{\mathcal{U}}_u}  \\
        &= \left\langle (s,t), (\dot{\nu}_{P,u}^*(w),\dot{\zeta}_{P,u}^*(w))\right\rangle_{\dot{\mathcal{M}}_P\oplus \dot{\mathcal{U}}_u}.
    \end{align*}
    The first equality above used \ref{it:rkhsEmbedRiesz} from Lemma~\ref{lem:rkhsEmbed}; the second used that $E_P[s_X(X)]=0$ and a change of measure; the third used the law of total expectation, the fact that $\dot{\mathcal{M}}_P=\{s\in L^2(P) : \int s\, dP=0\}$, and the fact that $\dot{\mathcal{U}}_u=L^2(Q_X)$; and the final equality used the definition of the inner product in the direct sum space $\dot{\mathcal{M}}_P\oplus \dot{\mathcal{U}}_u$ and the definitions of $\dot{\nu}_{P,u}^*$ and $\dot{\zeta}_{P,u}^*$ from \eqref{eq:rkhsEmbedAdjointDef}. Hence, $\dot{\Theta}_{P,u}^*(w)=(\dot{\nu}_{P,u}^*(w),\dot{\zeta}_{P,u}^*(w))$.\\[.5em]
    \textbf{\ref{it:rkhsTPD} of proof:}     We now establish that $\theta$ is totally pathwise differentiable at $(P,u)$ with differential operator $\dot{\Theta}_{P,u}$. To this end, fix $s\in\check{\mathcal{M}}_P$, $t\in \check{\mathcal{U}}_u$, $\{P_\epsilon : \epsilon\in [0,1]\}\in \mathscr{P}(P,\mathcal{M},s)$, and $\{u_\epsilon : \epsilon\in [0,1]\}\in \mathscr{P}(u,\mathcal{U},t)$. Our goal will be to show that the following is $o(\epsilon)$:
    \begin{align*}
        r(\epsilon):= \left\|\theta(P_\epsilon,u_\epsilon)-\theta(P,u) - \epsilon\, \dot{\Theta}_{P,u}(s,t)\right\|_{\mathcal{W}}.
    \end{align*}
    Using that $\|\tilde{w}\|_{\mathcal{W}}=\sup_{w\in\mathcal{B}_1}\langle w,\tilde{w}\rangle_{\mathcal{W}}$, where $\mathcal{B}_1:=\{w\in\mathcal{W} : \|w\|_{\mathcal{W}}\le 1\}$, leveraging the bilinearity of inner products and \ref{it:rkhsEmbedRiesz} from Lemma~\ref{lem:rkhsEmbed}, 
    and then applying the triangle inequality, we find that
    \begin{align*}
        r(\epsilon)&= \sup_{w\in\mathcal{B}_1}\left\langle w, \theta(P_\epsilon,u_\epsilon)-\theta(P,u) - \epsilon\, \dot{\Theta}_{P,u}(s,t)\right\rangle_{\mathcal{W}}\\
        &= \sup_{w\in\mathcal{B}_1}\left[\int w(x)u_\epsilon(x) P_{\epsilon,X}(dx)-\int w(x)u(x) P_X(dx) - \epsilon \int w(x)\left[u(x)s_X(x)+t(x)\right] P_X(dx)\right] \\
        &\le \epsilon\sup_{w\in\mathcal{B}_1}\left[\int w(x)\left[u_\epsilon(x)-u(x)\right]s_X(x) P_X(dx)\right] \\
        &\quad+ \sup_{w\in\mathcal{B}_1}\left[\int w(x)u_\epsilon(x) P_{X}(dx)-\int w(x)u(x) P_X(dx) - \epsilon \int w(x)t(x) P_X(dx)\right] \\
        &\quad+ \sup_{w\in\mathcal{B}_1}\left[\int w(x)u_\epsilon(x) P_{\epsilon,X}(dx)-\int w(x)u_\epsilon(x) P_X(dx) - \epsilon \int w(x) u_\epsilon(x)s_X(x) P_X(dx)\right].
    \end{align*}
    We denote the three terms on the right as (I), (II), and (III) in what follows. We will show that each of these terms is $o(\epsilon)$. When doing so, we will leverage similar arguments to those used to bound the terms on the right-hand side of \eqref{eq:UstatTpd} in the proof of Lemma~\ref{lem:UstatDiff}. 
    Beginning with (I), we apply Cauchy-Schwarz in $L^2(P_X)$, a change of measure, H\"{o}lder's inequality, and the fact that Cauchy-Schwarz in $\mathcal{W}$ implies elements in $\mathcal{B}_1$ are uniformly bounded by $\|K\|_\infty^{1/2}:=\sup_{x\in\mathcal{X}} K(x,x)^{1/2}$ to show that:
    \begin{align*}
        \textnormal{(I)}&\le \epsilon  \|u_\epsilon-u\|_{L^2(P_X)} \|s_X\|_{L^2(P_X)} \sup_{w\in\mathcal{B}_1}\sup_{x\in\mathcal{X}}|w(x)| \\
        &\le \epsilon \|u_\epsilon-u\|_{L^2(Q_X)}\|dP_X/dQ_X\|_{L^\infty(Q_X)}^{1/2} \|s_X\|_{L^2(P_X)} \|K\|_\infty^{1/2}.
    \end{align*}
    The right-hand side is $o(\epsilon)$ since $u_\epsilon\rightarrow u$ in $L^2(Q_X)$ as $\epsilon\rightarrow 0$. A similar argument shows that
    \begin{align*}
        \textnormal{(II)}&\le \|K\|_{\infty}^{1/2}\|dP_X/dQ_X\|_{L^\infty(Q_X)}^{1/2}\left\|u_\epsilon-u-\epsilon t\right\|_{L^2(Q_X)}.
    \end{align*}
    The right-hand side is $o(\epsilon)$ since $\{u_\epsilon : \epsilon\in [0,1]\}\in \mathscr{P}(u,\mathcal{U},t)$. To study (III), we define $p_{\epsilon,X}^{1/2}(\cdot):=[\frac{dP_{\epsilon,X}}{dP_X}(\cdot)]^{1/2}$ and note that, by the triangle inequality,
    \begin{align*}
        \textnormal{(III)}&\le \sup_{w\in\mathcal{B}_1}\left[\int w(x)u_\epsilon(x) [p_{\epsilon,X}^{1/2}(x)+1]\left[p_{\epsilon,X}^{1/2}(x)-1-\frac{\epsilon}{2}s_X(x)\right]P_{X}(dx)\right] \\
        &\quad+ \frac{\epsilon}{2}\sup_{w\in\mathcal{B}_1}\left[\int w(x)u_\epsilon(x) s_X(x) [p_{\epsilon,X}^{1/2}(x)-1] P_X(dx)\right].
    \end{align*}
    Cauchy-Schwarz, the boundedness of $w$ and $u_\epsilon$, the inequality $(a+b)^2\le 2(a^2+b^2)$, and H\"{o}lder's inequality can be used together to show that the first term on the right upper bounds by $2m\|K\|_\infty^{1/2}\|p_{\epsilon,X}^{1/2} - 1 - \epsilon s_X/2\|_{L^2(P_X)}$. This upper bound is $o(\epsilon)$ because $\{P_\epsilon : \epsilon\in [0,1]\}\in \mathscr{P}(P,\mathcal{M},s)$ and quadratic mean differentiability is preserved under marginalization --- see Proposition A.5.5 in \cite{bickel1993efficient} or Lemma~S8 in \cite{luedtke2023one} for details. The second term above can be similarly bounded by $\frac{\epsilon}{2}m\|K\|_\infty^{1/2}\|s_X\|_{L^2(P_X)}\|p_{\epsilon,X}^{1/2} - 1\|_{L^2(P_X)}$, which is also $o(\epsilon)$ by the quadratic mean differentiability of $\{P_{\epsilon,X} : \epsilon\in [0,1]\}$. As we have shown that (I), (II), and (III) are all $o(\epsilon)$, we have established that $r(\epsilon)=o(\epsilon)$. The fact that $s$, $t$, $\{P_\epsilon : \epsilon\}$, and $\{u_\epsilon : \epsilon\}$ were arbitrary establishes that $\theta$ is pathwise differentiable at $(P,u)$ with differential operator $\dot{\theta}_{P,u}=\dot{\Theta}_{P,u}$.
\end{proof}

The following lemma provides versions of results from Chapter 4.9.1.1 of \cite{berlinet2011reproducing} that are convenient for our setting. We include the proof for completeness. In the lemma and its proof, for $r\ge 1$ and $P\in\mathcal{M}$, we let $L^r(P_X;\mathcal{W})$ denote the Bochner space of $P_X$-a.s. equivalence classes of Bochner measurable maps $g : \mathcal{X}\rightarrow \mathcal{W}$ satisfying $\|g\|_{L^r(P_X;\mathcal{W})}^r:= \int \|g(x)\|_{\mathcal{W}}^r P_X(dx) <\infty$. The lemma concerns a map $\underline{\theta}$, which is an extension of $\theta$ from $\mathcal{M}\times\mathcal{U}$ to $\{(P,f) : P\in \mathcal{M},f\in L^2(P_X)\}$.
\begin{lemma}\label{lem:rkhsEmbed}
    For any $P\in\mathcal{M}$ and $f\in L^2(P_X)$, define $\underline{\theta}(P,f) : x\mapsto E_P[K(x,X)f(X)]$. The following holds for all $P\in\mathcal{M}$ and $f\in L^2(P_X)$:
    \begin{enumerate}[label=(\roman*)]
        \item\label{it:rkhsEmbedBochner} $x\mapsto K(\,\cdot\,,x)f(x)$ belongs to $L^2(P_X;\mathcal{W})$;
        \item\label{it:rkhsEmbedWellDef} $\underline{\theta}(P,f)\in \mathcal{W}$;
        \item\label{it:rkhsEmbedRiesz} $\underline{\theta}(P,f)$ is the unique element $w_{P,f}$ of $\mathcal{W}$ that satisfies $\langle w_{P,f},w\rangle_{\mathcal{W}} = \int w(x)f(x)P_X(dx)$ for all $w\in\mathcal{W}$;
        \item\label{it:rkhsEmbedInnerProd} for all $\tilde{P}\in\mathcal{M}$ and $\tilde{f}\in L^2(P_X)$, 
        \begin{align*}
            \langle \underline{\theta}(P,f),\underline{\theta}(\tilde{P},\tilde{f})\rangle_{\mathcal{W}}=\iint K(x,\tilde{x}) f(x) \tilde{f}(\tilde{x}) P_X(dx) \tilde{P}_X(d\tilde{x}).
        \end{align*}
    \end{enumerate}
\end{lemma}
\begin{proof}[Proof of Lemma~\ref{lem:rkhsEmbed}]
    This proof borrows arguments from the proof of Theorem~105 from \cite{berlinet2011reproducing}. Throughout it, we fix $P\in\mathcal{M}$, $\tilde{P}\in\mathcal{M}$, $f\in L^2(P_X)$, $\tilde{f}\in L^2(\tilde{P}_X)$, and $w\in\mathcal{W}$.
    
    Starting with \ref{it:rkhsEmbedBochner}, observe that, by the boundedness of $K$ and the fact that $f\in L^2(P_X)$, $x\mapsto K(\,\cdot\,,x)f(x)$ belongs to $L^2(P_X;\mathcal{W})$. Indeed,
    \begin{align*}
        \left\|x\mapsto K(\,\cdot\,,x)f(x)\right\|_{L^2(P_X;\mathcal{W})}^2&= \int \left\langle K(\,\cdot\,,x)f(x), K(\,\cdot\,,x)f(x)\right\rangle_{\mathcal{W}} P_X(dx) = \int K(x,x) f(x)^2 P_X(dx)\\
        &\le \left[\sup_{x\in\mathcal{X}}K(x,x)\right] \|f\|_{L^2(P)}^2<\infty,
    \end{align*}
    where we used that the kernel is bounded. Consequently, $x\mapsto K(\,\cdot\,,x)f(x)$ belongs to $L^2(P_X;\mathcal{W})$.
    
    As for \ref{it:rkhsEmbedWellDef}, the fact that $P_X$ is a probability measure implies $\|\cdot\|_{L^1(P_X;\mathcal{W})}\le \|\cdot\|_{L^2(P_X;\mathcal{W})}$. Hence, \ref{it:rkhsEmbedBochner} implies $x\mapsto K(\,\cdot\,,x)f(x)$ is Bochner integrable, and so $\underline{\theta}(P,f)(\cdot):=\int K(\,\cdot\,,x) f(x)P_X(dx)$ is well-defined and the resulting function belongs to $\mathcal{W}$.

    We now turn to \ref{it:rkhsEmbedRiesz}. Since $\langle\,\cdot\,,w\rangle_{\mathcal{W}} : \mathcal{W}\rightarrow\mathbb{R}$ is a continuous linear operator, it can be interchanged with Bochner integration. Hence,
        \begin{align*}
        \left\langle\underline{\theta}(P,f),w\right\rangle_{\mathcal{W}}&= \left\langle \int K(\,\cdot\,,x) f(x)P_X(dx),w\right\rangle_{\mathcal{W}} = \int \left\langle K(\,\cdot\,,x),w\right\rangle_{\mathcal{W}}P_X(dx) = \int w(x) f(x)P_X(dx).
    \end{align*}
    To show uniqueness, let $w_{P,f}$ satisfy $ \left\langle w_{P,f},w\right\rangle_{\mathcal{W}}=\int w(x) f(x)P_X(dx)$ for all $w$. Then,
    \begin{align*}
        \left\|w_{P,f}-\underline{\theta}(P,f)\right\|_{\mathcal{W}}&= \sup_{w\in\mathcal{W} : \|w\|_{\mathcal{W}}\le 1}\left[\left\langle w_{P,f},w\right\rangle_{\mathcal{W}}-\left\langle\underline{\theta}(P,f),w\right\rangle_{\mathcal{W}}\right] = 0.
    \end{align*}
    Hence, $w_{P,f}=\underline{\theta}(P,f)$, which establishes \ref{it:rkhsEmbedRiesz}.

    We conclude by showing \ref{it:rkhsEmbedInnerProd}. For that result, we apply \ref{it:rkhsEmbedRiesz} twice, once for $\underline{\theta}(\tilde{P},\tilde{f})$ with $w=\underline{\theta}(P,f)$ and once for $\underline{\theta}(P,f)$ with $w=K(\,\cdot\,,\tilde{x})$, to show
    \begin{align*}
        \langle \underline{\theta}(P,f),\underline{\theta}(\tilde{P},\tilde{f})\rangle_{\mathcal{W}}&= \int \underline{\theta}(P,f)(\tilde{x})\, \tilde{f}(\tilde{x})\tilde{P}_X(d\tilde{x}) =  \int \langle\underline{\theta}(P,f),K(\,\cdot\,,\tilde{x})\rangle_{\mathcal{W}}\, \tilde{f}(\tilde{x})\tilde{P}_X(d\tilde{x}) \\
        &= \iint K(x,\tilde{x})f(x)\tilde{f}(\tilde{x}) P_X(dx)\tilde{P}_X(d\tilde{x}).
    \end{align*}
\end{proof}

\subsubsection{Optimal value}\label{app:optimalValue}

For each $y$ in a metric space $\mathcal{Y}$, let $F_y$ map from $\mathcal{M}\times\mathcal{U}$ to $\mathbb{R}$. We study the optimal value map $\theta$ given by
\begin{align}
    \theta(P,u)&= \inf_{y\in\mathcal{Y}}F_y(P,u). \label{eq:optVal}
\end{align}
If $F_y$ is totally pathwise differentiable at $(P,u)$, then we write $\dot{F}_{y,P,u}$ and $\dot{F}_{y,P,u}^*$ to denote its differential operator and corresponding adjoint. We call $\{F_y : y\in\mathcal{Y}\}$ uniformly totally pathwise differentiable at $(y(0),P,u)$ if there exists a neighborhood of $y(0)$ on which $F_y$ is totally pathwise differentiable and, for all $s\in\check{\mathcal{M}}_P$, $t\in\check{\mathcal{U}}_u$, $\{P_\epsilon : \epsilon\}\in \mathscr{P}(P,\mathcal{M},s)$, and $\{u_\epsilon : \epsilon\}\in \mathscr{P}(u,\mathcal{U},t)$,
\begin{align*}
    \lim_{(y,\epsilon)\rightarrow (y(0),0)}\left([F_{y}(P_\epsilon,u_\epsilon)-F_{y}(P,u)]/\epsilon - \dot{F}_{y,P,u}(s,t)\right) = 0.
\end{align*}

\begin{lemma}[Total pathwise differentiability of optimal value]\label{lem:optimalVal}
    Fix $(P,u)\in\mathcal{M}\times\mathcal{U}$.
    Suppose
    \begin{enumerate}[label=(\roman*),noitemsep]
        \item\label{it:ovYcompact} $\mathcal{Y}$ is a compact metric space;
        \item\label{it:ovUniqueMin} there exists a unique $y(0)\in\mathcal{Y}$ such that $\theta(P,u)=F_{y(0)}(P,u)$;
        \item $\{F_y : y\in\mathcal{Y}\}$ is uniformly totally pathwise differentiable at $(y(0),P,u)$;
        \item for all $s\in\check{\mathcal{M}}_P$ and $t\in\check{\mathcal{U}}_u$, $y\mapsto \dot{F}_{y,P,u}(s,t)$ is continuous at $y(0)$; and
        \item\label{it:ovContObj} there exists a neighborhood $\mathcal{N}$ of $(P,u)$ such that $(y,P',u')\mapsto F_y(P',u')$ is continuous on $\mathcal{Y}\times \mathcal{N}$.
    \end{enumerate}
    Then, $\theta$ from \eqref{eq:optVal} is totally pathwise differentiable at $(P,u)$ with $\dot{\theta}_{P,u}=\dot{F}_{y(0),P,u}$ and $\dot{\theta}_{P,u}^*=\dot{F}_{y(0),P,u}^*$.
\end{lemma}
If $\mathcal{Y}$ is not compact, the implication of the above theorem can still hold if \ref{it:ovYcompact} is replaced by other conditions. Starting with the simplest, if there is a compact $\mathcal{Y}_{P,u}\subseteq \mathcal{Y}$ such that, for all $(P',u')$ in a neighborhood of $(P,u)$, $\inf_{y\in\mathcal{Y}_{P,u}}F_y(P',u')= \inf_{y\in\mathcal{Y}}F_y(P',u')$, then the above lemma can be directly applied to the establish the differentiability of $\tilde{\theta}(P',u'):=\inf_{y\in\mathcal{Y}_{P,u}} F_y(P',u')$; since $\theta$ and $\tilde{\theta}$ agree in a neighborhood of $(P,u)$, they will have the same differentiability properties at $(P,u)$. Considerations for other replacements of \ref{it:ovYcompact} are grounded in the fact that the compactness condition is only used twice in the proof: it ensures that, given smooth paths $\{P_\epsilon : \epsilon\}$ and $\{u_\epsilon : \epsilon\}$, there exists a minimizer $y(\epsilon)\in \argmin_{y\in\mathcal{Y}} F_y(P_\epsilon,u_\epsilon)$ for all $\epsilon$ small enough, and it aids in showing that $y(\epsilon)\rightarrow y(0)$ as $\epsilon\rightarrow 0$. Hence, if the existence and convergence of these minimizers can be established by other means, then the compactness condition can be dropped. Similar conditions to those used to establish the consistency of M-estimators can be used to establish these properties. For example, \ref{it:ovYcompact} may be replaced by the requirement that $\argmin_{y\in\mathcal{Y}}F_y(P',u')\not=\emptyset$ for all $(P',u')$ in a neighborhood of $(P,u)$ and either
\begin{itemize}
    \item[\textit{(i')}] there is a well-separated minimizer $y(0)$ of $y\mapsto F_y(P,u)$ and $\sup_{y\in\mathcal{Y}}|F_y(P',u')-F_y(P,u)|\rightarrow 0$ as $(P',u')\rightarrow (P,u)$ \citep[see Theorem~5.7 of][]{van2000asymptotic}; or
    \item[\textit{(i'')}] $\mathcal{Y}=\mathbb{R}^d$ with $d<\infty$ and $y\mapsto F_y(P',u')$ convex for all $(P',u')$ in a neighborhood of $(P,u)$ \citep[see Theorem~2.7][]{newey1994large}.
\end{itemize}

\begin{proof}[Proof of Lemma~\ref{lem:optimalVal}]
        This proof is inspired by arguments from \cite{danskin1967directional}. Fix $s\in\check{\mathcal{M}}_P$, $t\in\check{\mathcal{U}}_u$, $\{P_\epsilon : \epsilon\}\in \mathscr{P}(P,\mathcal{M},s)$, and $\{u_\epsilon : \epsilon\}\in \mathscr{P}(u,\mathcal{U},t)$. Our goal is to show that $g(\epsilon):= [\theta(P_\epsilon,u_\epsilon)-\theta(P,u)]/\epsilon$ converges to $\dot{F}_{y(0),P,u}(s,t)$ as $\epsilon\downarrow 0$. This will show that $\theta$ is totally pathwise differentiable at $(P,u)$ with $\dot{\theta}_{P,u}=\dot{F}_{y(0),P,u}$ and $\dot{\theta}_{P,u}^*=\dot{F}_{y(0),P,u}^*$.

        Throughout this proof we take $\epsilon$ to be small enough so that $(P_\epsilon,u_\epsilon)\in\mathcal{N}$; this is necessarily possible since $(P_\epsilon,u_\epsilon)\rightarrow (P,u)$ as $\epsilon\rightarrow 0$. For each $\epsilon$, let $y(\epsilon)$ be a generic element of the set $\argmin_{y\in\mathcal{Y}}F_y(P_\epsilon,u_\epsilon)$; this set is necessarily nonempty since $\mathcal{Y}$ is compact and $y\mapsto F_y(P_\epsilon,u_\epsilon)$ is continuous. Also, as we will now show, $\lim_{\epsilon\rightarrow 0}y(\epsilon)= y(0)$. To see why, take a sequence $\epsilon_k\rightarrow 0$. Since $\mathcal{Y}$ is compact, there exists a subsequence $\epsilon_{k'}\rightarrow 0$ such that $y(\epsilon_{k'})\rightarrow \tilde{y}(0)$ for some $\tilde{y}(0)$. Moreover, $\tilde{y}(0)$ must equal the unique minimizer $y(0)$ of $y\mapsto F_y(P,u)$ since the continuity of $(y,P',u')\mapsto F_y(P',u')$ and the fact that $y(\epsilon_{k'})\in \argmin_{y\in\mathcal{Y}}F_y(P_{\epsilon_{k'}},u_{\epsilon_{k'}})$ together imply that $F_{\tilde{y}(0)}(P,u)=\lim_{k'} F_{y(\epsilon_{k'})}(P_{\epsilon_{k'}},u_{\epsilon_{k'}})\le \lim_{k'} F_{y(0)}(P_{\epsilon_{k'}},u_{\epsilon_{k'}})=F_{y(0)}(P,u)$. As $y(\epsilon_{k'})\rightarrow \tilde{y}(0)$ along the subsequence $(\epsilon_{k'})_{k'}$ and the original sequence $(\epsilon_k)_k$ was arbitrary, $\lim_{\epsilon\rightarrow 0}y(\epsilon)= y(0)$.
        
        We now show that $\limsup_{\epsilon\rightarrow 0}g(\epsilon)\le \dot{F}_{y(0),P,u}(s,t)$. Since $y(\epsilon)$ minimizes $y\mapsto F_y(P_\epsilon,u_\epsilon)$ and $F_{y(0)}$ is totally pathwise differentiable,
    \begin{align*}
        \epsilon g(\epsilon) = F_{y(\epsilon)}(P_\epsilon,u_\epsilon)-F_{y(0)}(P,u)&\le F_{y(0)}(P_\epsilon,u_\epsilon)-F_{y(0)}(P,u) = \epsilon \dot{F}_{y(0),P,u}(s,t) + o(\epsilon).
    \end{align*}
    Dividing both sides by $\epsilon$ and taking a limit superior as $\epsilon\rightarrow 0$ shows that $\limsup_{\epsilon\rightarrow 0}g(\epsilon)\le \dot{F}_{y(0),P,u}(s,t)$.
    
    We now show that $\liminf_{\epsilon\rightarrow 0}g(\epsilon)\ge \dot{F}_{y(0),P,u}(s,t)$. Using that $y(0)$ minimizes $y\mapsto F_y(P,u)$ and adding and subtracting terms,
    \begin{align*}
        \epsilon g(\epsilon)&= F_{y(\epsilon)}(P_\epsilon,u_\epsilon)-F_{y(0)}(P,u)\ge F_{y(\epsilon)}(P_\epsilon,u_\epsilon)-F_{y(\epsilon)}(P,u) \\
        &= \epsilon\dot{F}_{y(0),P,u}(s,t)+\epsilon [\dot{F}_{y(\epsilon),P,u}(s,t)-\dot{F}_{y(0),P,u}(s,t)] \\
        &\quad+ \left[F_{y(\epsilon)}(P_\epsilon,u_\epsilon)-F_{y(\epsilon)}(P,u) - \epsilon \dot{F}_{y(\epsilon),P,u}(s,t)\right].
    \end{align*}
    The second term on the right is $o(\epsilon)$ since $\lim_{\epsilon\rightarrow 0}y(\epsilon)= y(0)$ and $y\mapsto \dot{F}_{y,P,u}(s,t)$ is continuous at $y(0)$. The third is $o(\epsilon)$ because $\{F_y : y\in\mathcal{Y}\}$ is uniformly totally pathwise differentiable at $(y(0),P,u)$. Dividing both sides by $\epsilon$ and taking a limit inferior as $\epsilon\rightarrow 0$ shows that $\liminf_{\epsilon\rightarrow 0}g(\epsilon)\ge \dot{F}_{y(0),P,u}(s,t)$.
\end{proof}

\subsubsection{Optimal solution}\label{app:optimalSolution}

For each $w\in\mathcal{W}=\mathbb{R}^d$, let $F_w$ map from $\mathcal{M}\times\mathcal{U}$ to $\mathbb{R}$. We study the optimal solution map $\theta$ given by
\begin{align}
    \theta(P,u)&\in \argmin_{w\in\mathcal{W}}F_w(P,u). \label{eq:optSol}
\end{align}
When this minimization results in multiple solutions, $\theta$ is required to select one of them in a way that ensures this map is continuous at the point $(P,u)$ where we wish to differentiate it. For example, $\theta$ may always select the minimum norm solution. The following lemma assumes that this choice is made in such a way that $\theta$ is continuous. In the lemma, $\|\cdot\|_2$ denotes the Euclidean norm on $\mathbb{R}^d$.

\begin{lemma}[Total pathwise differentiability of optimal solution]\label{lem:optimalSol}
    Fix $(P,u)\in\mathcal{M}\times\mathcal{U}$. Suppose
    \begin{enumerate}[label=(\roman*),noitemsep]
        \item\label{it:osCont} $\theta$ is continuous at $(P,u)$;
        \item\label{it:osTwiceDiff} $w\mapsto F_w(P,u)$ is twice continuously differentiable at $w(0):=\theta(P,u)$ with positive definite Hessian matrix $H_{P,u}$; and
        \item\label{it:osPartialDiff} there exists a bounded linear operator $G_{P,u} : \dot{\mathcal{M}}_P\oplus \dot{\mathcal{U}}_u\rightarrow \mathbb{R}^d$ such that, for all $s\in\check{\mathcal{M}}_P$, $t\in\check{\mathcal{U}}_u$, $\{P_\epsilon : \epsilon\}\in \mathscr{P}(P,\mathcal{M},s)$, $\{u_\epsilon : \epsilon\}\in \mathscr{P}(u,\mathcal{U},t)$, and function $\widetilde{w} : [0,1]\rightarrow\mathbb{R}^d$ satisfying $\widetilde{w}(\epsilon)\rightarrow w(0)$,
        \begin{align*}
            &F_{\widetilde{w}(\epsilon)}(P_\epsilon,u_\epsilon)-F_{\widetilde{w}(\epsilon)}(P,u) - F_{w(0)}(P_\epsilon,u_\epsilon)+F_{w(0)}(P,u) \\
            &= \epsilon [\widetilde{w}(\epsilon)-w(0)]^\top G_{P,u}(s,t)+ o\left(\max\{\|\widetilde{w}(\epsilon)-w(0)\|_2^2,\epsilon^2\}\right).
        \end{align*}
    \end{enumerate}
    Under the above conditions, $\theta$ is totally pathwise differentiable at $(P,u)$ with $\dot{\theta}_{P,u}(s,t)=H_{P,u}^{-1}G_{P,u}(s,t)$ and $\dot{\theta}_{P,u}^*(w)=H_{P,u}^{-1}G_{P,u}^*(w)$, where $\dot{G}_{P,u}^*$ denotes the Hermitian adjoint of $\dot{G}_{P,u}$.
\end{lemma}
The above will often be applied when $F_w$ is totally pathwise differentiable for all $w$ in a neighborhood of $w(0)$. Denoting the differential operator of $F_w$ by $\dot{F}_{w,P,u}$, the vector $G_{P,u}(s,t)$ will usually then correspond to the gradient of $w\mapsto \dot{F}_{w,P,u}(s,t)$ at $w(0)$. An example of a situation where this lemma can be applied occurs when $\mathcal{U}=L^\infty(\rho_X)$ for $\rho_X$ equivalent to $P_X$ and $F_w(P,u)=\int \left[u(x)-w^\top x\right]^2 P_X(dx)$. In these cases, $\theta(P,u)$ is the $L^2(P_X)$ projection of onto linear functions of the form $x\mapsto w^\top x$.

Compared to Appendix~\ref{app:optimalValue}, which studied an optimal value map, this appendix requires additional structure on the feasible region of the optimization problem, namely that it be a Euclidean space. It also requires additional differentiability properties on the objective function, such as the second-order differentiability conditions \ref{it:osTwiceDiff} and \ref{it:osPartialDiff}, though does not necessarily require others, such as uniform total pathwise differentiability. This mismatch between the conditions for differentiability of the optimal value and solution maps is not overly surprising given that they capture different aspects of the optimization problem.

\begin{proof}[Proof of Lemma~\ref{lem:optimalSol}]
This proof is inspired by arguments from the proof of Theorem~3.2.16 of \cite{van1996weak}, which establish the asymptotic normality of an M-estimator.

Fix $s\in\check{\mathcal{M}}_P$, $t\in\check{\mathcal{U}}_u$, $\{P_\epsilon : \epsilon\}\in \mathscr{P}(P,\mathcal{M},s)$, $\{u_\epsilon : \epsilon\}\in \mathscr{P}(u,\mathcal{U},t)$, and $\widetilde{w}(\epsilon)\rightarrow w(0)$. By \ref{it:osTwiceDiff}, a second-order Taylor expansion holds for $w\mapsto F_w(P,u)$ at $w(0)$, and, as $w(0)$ is a minimizer, the first-order term is equal to zero --- expressed as an equation,
\begin{align*}
    F_{\widetilde{w}(\epsilon)}(P,u)-F_{w(0)}(P,u)&= \frac{1}{2}[\widetilde{w}(\epsilon)-w(0)]^\top H_{P,u} [\widetilde{w}(\epsilon)-w(0)] + o(\|\widetilde{w}(\epsilon)-w(0)\|_2^2).
\end{align*}
Combining this with \ref{it:osPartialDiff}, we find that
\begin{align}
    &F_{\widetilde{w}(\epsilon)}(P_\epsilon,u_\epsilon)-F_{w(0)}(P_\epsilon,u_\epsilon) \label{eq:osFepsDecomp} \\
    &= F_{\widetilde{w}(\epsilon)}(P,u)-F_{w(0)}(P,u) + \left[F_{\widetilde{w}(\epsilon)}(P_\epsilon,u_\epsilon)-F_{w(0)}(P_\epsilon,u_\epsilon)-F_{\widetilde{w}(\epsilon)}(P,u)+F_{w(0)}(P,u)\right] \nonumber \\
    &= \frac{1}{2}[\widetilde{w}(\epsilon)-w(0)]^\top H_{P,u} [\widetilde{w}(\epsilon)-w(0)] + \epsilon [\widetilde{w}(\epsilon)-w(0)]^\top G_{P,u}(s,t)+ o\left(\max\{\|\widetilde{w}(\epsilon)-w(0)\|_2^2,\epsilon^2\}\right). \nonumber
\end{align}
We shall use this result three times with different choices of $\widetilde{w}$ in what follows: once to establish the $\epsilon$-rate convergence of $w(\epsilon):=\theta(P_\epsilon,u_\epsilon)$ to $w(0)$ --- in that $w(\epsilon)=w(0) + O(\epsilon)$ --- and then twice more to derive a first-order approximation of $\theta(P_\epsilon,u_\epsilon)$ that will establish the total pathwise differentiability of $\theta$.

To establish the $\epsilon$-rate convergence result, we take $\widetilde{w}(\epsilon)$ to be equal to $w(\epsilon):=\theta(P_\epsilon,u_\epsilon)$. This choice satisfies the one required property of $\tilde{w}(\epsilon)$ --- namely, $\widetilde{w}(\epsilon)\rightarrow w(0)$ --- as follows by \ref{it:osCont} and the fact that $\{P_\epsilon : \epsilon\}\in \mathscr{P}(P,\mathcal{M},s)$ and $\{u_\epsilon : \epsilon\}\in \mathscr{P}(u,\mathcal{U},t)$ imply $(P_\epsilon,u_\epsilon)\rightarrow (P,u)$. Since $w(\epsilon)$ is a minimizer of $w\mapsto F_w(P_\epsilon,u_\epsilon)$, the left-hand side of \eqref{eq:osFepsDecomp} is nonpositive. Also, by the positive definiteness of $H_{P,u}$ and Cauchy-Schwarz, the first and second terms on the right-hand side lower bound by $\|w(\epsilon)-w(0)\|_2^2$ for $c_1>0$ and $-\epsilon c_2\|w(\epsilon)-w(0)\|_2$ for $c_2\ge 0$, respectively. Putting these observations together and then completing the square shows that
\begin{align*}
    0&\ge c_1\|w(\epsilon)-w(0)\|_2^2 - \epsilon c_2\|w(\epsilon)-w(0)\|_2 + o\left(\max\{\|\widetilde{w}(\epsilon)-w(0)\|_2^2,\epsilon^2\}\right) \\
    &= [c_1+o(1)]\left[\|w(\epsilon)-w(0)\|_2 + O(\epsilon)\right]^2.
\end{align*}
The above is only possible if $\|w(\epsilon)-w(0)\|_2=O(\epsilon)$, which is what we set out to show.

We now use \eqref{eq:osFepsDecomp} twice more to establish the total pathwise differentiability of $\theta$ at $(P,u)$. In the first application we again take $\widetilde{w}(\epsilon)=w(\epsilon)$, and in the second we take $\widetilde{w}(\epsilon)=\bar{w}(\epsilon):=w(0)-\epsilon H_{P,u}^{-1}G_{P,u}(s,t)$, yielding
\begin{align*}
    F_{w(\epsilon)}(P_\epsilon,u_\epsilon)-F_{w(0)}(P_\epsilon,u_\epsilon)&= \frac{1}{2}[w(\epsilon)-w(0)]^\top H_{P,u} [w(\epsilon)-w(0)] \\
    &\quad+ \epsilon [w(\epsilon)-w(0)]^\top G_{P,u}(s,t)+ o(\epsilon^2), \\
    F_{\bar{w}(\epsilon)}(P_\epsilon,u_\epsilon)-F_{w(0)}(P_\epsilon,u_\epsilon)&= - \frac{1}{2}\epsilon^2 [H_{P,u}^{-1}G_{P,u}(s,t)]^\top G_{P,u}(s,t)+ o(\epsilon^2).
\end{align*}
Subtracting the second display from the first yields
\begin{align*}
    &F_{w(\epsilon)}(P_\epsilon,u_\epsilon)- F_{\bar{w}(\epsilon)}(P_\epsilon,u_\epsilon) \\
    &= \frac{1}{2}[w(\epsilon)-w(0) + \epsilon H_{P,u}^{-1}G_{P,u}(s,t)]^\top H_{P,u} [w(\epsilon)-w(0) + \epsilon H_{P,u}^{-1}G_{P,u}(s,t)] + o(\epsilon^2).
\end{align*}
The left-hand side is nonpositive since $w(\epsilon)$ is a minimizer of $w\mapsto F_w(P_\epsilon,u_\epsilon)$. Combining this with the positive definiteness of $H_{P,u}$ shows there is a $c_1>0$ such that
\begin{align*}
    0&\ge c_1\left\|w(\epsilon)-w(0) + \epsilon H_{P,u}^{-1}G_{P,u}(s,t)\right\|_2^2 + o(\epsilon^2).
\end{align*}
This is only possible if the squared norm on the right-hand side is $o(\epsilon^2)$. Since $s\in\check{\mathcal{M}}_P$, $t\in\check{\mathcal{U}}_u$, $\{P_\epsilon : \epsilon\}\in \mathscr{P}(P,\mathcal{M},s)$, $\{u_\epsilon : \epsilon\}\in \mathscr{P}(u,\mathcal{U},t)$ were arbitrary and $G_{P,u}$ is bounded and linear, this shows that $\theta$ is pathwise differentiable at $(P,u)$ with differential operator $\dot{\theta}_{P,u}(s,t)=H_{P,u}^{-1}G_{P,u}(s,t)$.

To see that $\dot{\theta}_{P,u}^*(w)=H_{P,u}^{-1}G_{P,u}^*(w)$, observe that, for all $s\in\dot{\mathcal{M}}_P$, $t\in\dot{\mathcal{U}}_u$, and $w\in\mathcal{W}=\mathbb{R}^d$,
\begin{align*}
    \langle w,H_{P,u}^{-1}G_{P,u}(s,t)\rangle_{\mathcal{W}}&= w^\top H_{P,u}^{-1}G_{P,u}(s,t) = \left[w^\top H_{P,u}^{-1}\right]G_{P,u}(s,t) \\
    &= \left[H_{P,u}^{-1}w\right]^\top G_{P,u}(s,t)
    = \left\langle G_{P,u}^*\left(H_{P,u}^{-1}w\right),(s,t)\right\rangle_{\dot{\mathcal{M}}_P\oplus\dot{\mathcal{U}}_u} \\
    &= \left\langle H_{P,u}^{-1}G_{P,u}^*\left(w\right),(s,t)\right\rangle_{\dot{\mathcal{M}}_P\oplus\dot{\mathcal{U}}_u},
\end{align*}
where above we used the fact that the Hessian inverse $H_{P,u}^{-1}$ is symmetric and $\dot{G}_{P,u}^*$ is linear.
\end{proof}

\subsection{Maps that only depend on their distribution-valued argument}\label{app:primitiveOnlyDist}

\subsubsection{General case: pathwise differentiable parameter}\label{app:pathwiseDiff}

Suppose that there exists a pathwise differentiable map $\nu : \mathcal{M}\rightarrow\mathcal{W}$ that is such that $\theta(P,u)=\nu(P)$ for all $(P,u)\in\mathcal{M}\times\mathcal{U}$. 
We claim that $\theta$ is totally pathwise differentiable at any $(P,u)\in\mathcal{M}\times\mathcal{U}$ with $\dot{\theta}_{P,u}(s,t)=\dot{\nu}_P(s)$ and $\dot{\theta}_{P,u}^*(w)=(\dot{\nu}_P^*(w),0)$. Though the results in Table~\ref{tab:primitives} specify that $\mathcal{U}=\{0\}$ in this case, they actually hold for any subset $\mathcal{U}$ of a Hilbert space $\mathcal{T}$; however, since $\theta(P,u)$ does not depend on the value $u$ takes, the case where $\mathcal{U}$ is a trivial vector space $\{0\}$ is sufficiently rich to capture all interesting aspects of the setting.

Fix $P\in\mathcal{M}$, $u\in\mathcal{U}$, $s\in\check{\mathcal{M}}_P$, $t\in\check{\mathcal{U}}_u$, $\{P_\epsilon : \epsilon\in [0,1]\}\in\mathscr{P}(P,\mathcal{M},s)$, and $\{u_\epsilon : \epsilon\in [0,1]\}\in \mathscr{P}(u,\mathcal{U},t)$. For any $\epsilon\in (0,1]$, the definition of $\theta$ and the pathwise differentiability of $\nu$ imply that
\begin{align*}
    \left\|\theta(P_\epsilon,u_\epsilon)-\theta(P,u)-\epsilon \dot{\nu}_P(s)\right\|_{\mathcal{W}}&= \left\|\nu(P_\epsilon)-\nu(P)-\epsilon \dot{\nu}_P(s)\right\|_{\mathcal{W}} = o(\epsilon).
\end{align*}
This establishes our claim that $\dot{\theta}_{P,u}(s,t)=\dot{\nu}_P(s)$. The map $\dot{\nu}_P$ is bounded and linear, and so $\dot{\theta}_{P,u}$ is bounded and linear as well. As for the adjoint, observe that, for any $s\in \dot{\mathcal{M}}_P$, $t\in\dot{\mathcal{U}}_u$, and $w\in\mathcal{W}$,
\begin{align*}
    \left\langle \dot{\theta}_{P,u}(s,t), w\right\rangle_{\mathcal{W}}&= \left\langle \dot{\nu}_P(s), w\right\rangle_{\mathcal{W}} = \left\langle t, \dot{\nu}_P^*(w)\right\rangle_{\dot{\mathcal{M}}_P} = \left\langle (s,t), (\dot{\nu}_P^*(w),0)\right\rangle_{\dot{\mathcal{M}}_P\oplus \dot{\mathcal{U}}_u}.
\end{align*}
Hence,  $\dot{\theta}_{P,u}^*(w)=(\dot{\nu}_P^*(w),0)$.

In the remainder of Appendix~\ref{app:primitiveOnlyDist}, we establish the pathwise differentiability of maps $\nu : \mathcal{M}\rightarrow\mathcal{W}$. Applying the general results of this Appendix~\ref{app:pathwiseDiff} then shows that the parameter $\theta$ with $\theta(P,u)=\nu(P)$ is totally pathwise differentiable at any $(P,u)\in\mathcal{M}\times\mathcal{U}$ with $\dot{\theta}_{P,u}(s,t)=\dot{\nu}_P(s)$ and $\dot{\theta}_{P,u}^*(w)=(\dot{\nu}_P^*(w),0)$, which will verify the claimed form of the adjoint of the differential operator for the pathwise differentiable primitives in Table~\ref{tab:primitives}.

\subsubsection{Root-density}\label{app:rootDensity}

Suppose there is a $\sigma$-finite measure $\lambda$ that dominates all $P\in \mathcal{M}$, and let $\nu(P)(z)=\frac{dP}{d\lambda}(z)^{1/2}$. By Example~4 in \cite{luedtke2023one}, $\nu : \mathcal{M}\rightarrow L^2(\lambda)$ is pathwise differentiable at each $P$ with local parameter $\dot{\nu}_P(s)(z)=\frac{1}{2}s(z)\nu(P)(z)$ and efficient influence operator
$$\dot{\nu}_P^*(w)(z) = \frac{w(z)}{2\nu(P)(z)} - E_P \left[\frac{w(Z)}{2\nu(P)(Z)} \right].$$

\subsubsection{Conditional density}\label{app:condDensity}

Let $\mathcal{M}$ be a locally nonparametric model of distributions on $\mathcal{Z}=\mathcal{X}\times\mathcal{Y}$. Suppose there exists a $\sigma$-finite measure $\lambda$ on $\mathcal{Y}$ such that the conditional distribution $P_{Y|X=x}$ of $Y\mymid X=x$ under every $P\in\mathcal{M}$ is dominated by $\lambda$ for $P_X$-almost all $x$. Further suppose there exists $m<\infty$ such that, for all $P\in\mathcal{M}$, the conditional density $p_{Y|X}(y\mymid x):=\frac{dP_{Y|X=x}}{d\lambda}(y)$ is $P$-a.s. bounded by $m$. To ease notation, we will write $p_{Y|X}(z)$ rather than $p_{Y|X}(y\mymid x)$ hereafter. Suppose all distributions in $\mathcal{M}$ are equivalent. For fixed $Q\in\mathcal{M}$, define $\nu : \mathcal{M}\rightarrow L^2(Q)$ so that $\nu(P)(z)=p_{Y|X}(z)$. 

For any $P\in\mathcal{M}$ satisfying $\|dQ/dP\|_{L^\infty(P)}<\infty$, we claim that $\nu$ is pathwise differentiable with
\begin{align}
    \dot{\nu}_P(s)(z)=\Pi_P(s)(z)\,p_{Y|X}(z)\ \ \textnormal{ and }\ \ \dot{\nu}_P^*(w)(z)=\Pi_P\big(\tfrac{dQ}{dP}wp_{Y|X}\big)(z), \label{eq:condDens}
\end{align}
where $\Pi_P(f)(z):=f(z)-E_P[f(Z)\mymid X=x]$. 
We establish this by expressing $\nu$ as a composition $\zeta_2\circ\zeta_1\circ \nu_1$, where $\nu_1 : \mathcal{M}\rightarrow L^2(P)$ is pathwise differentiable and $\zeta_1 : L^2(P)\rightarrow L^2(P)$ and $\zeta_2 : L^2(P)\rightarrow L^2(Q)$ are both Hadamard differentiable. The chain rule will then establish the pathwise differentiability of $\nu$ with the claimed local parameter and efficient influence operator.

The map $\nu_1$ is defined so that $\nu_1(\tilde{P})=(\tilde{p}_{Y|X}/p_{Y|X})^{1/2}$. Lemma~S8 in \cite{luedtke2023one} shows that $\nu_1$ is pathwise differentiable at $P$ with $\dot{\nu}_{1,P}(s)(z)=\frac{1}{2}\Pi_P(s)(z)$. The adjoint of this map is $\dot{\nu}_{1,P}^*(w)(z)=\frac{1}{2}\{w(z)-E_P[w(Z)\mymid X=x]\}$. The map $\zeta_1$ is the pointwise operation $\zeta_1(u)(z)=u^2(z)p_{Y|X}(z)$.  By Appendix~\ref{app:pointwise} and the fact that conditional densities of distributions in $\mathcal{M}$ are uniformly bounded, this map is Hadamard differentiable with $\dot{\zeta}_{1,u}(t)=2tup_{Y|X}$ and $\dot{\zeta}_{1,u}^*(w)=2wup_{Y|X}$. The map $\zeta_2$ is the change of measure studied in Appendix~\ref{app:changeMeasure} where, in the notation of that appendix, $\lambda_1=P$ and $\lambda_2=Q$. By results in that appendix, the fact that $\|dQ/dP\|_{L^\infty(P)}<\infty$ implies this map is Hadamard differentiable with $\dot{\zeta}_{2,u}(t)=t$ and $\dot{\zeta}_{2,u}^*(w)=\frac{dQ}{dP}w$. Theorem~\ref{thm:backpropWorks} from this work (see also Theorem~3.1 in \citealp{van1991differentiable}) then shows that $\nu=\zeta_2\circ\zeta_1\circ \nu_1$ is Hadamard differentiable with $\dot{\nu}_P^*(w)=\dot{\nu}_{1,P}^*\circ \dot{\zeta}_{1,\nu_1(P)}^*\circ \dot{\zeta}_{2,\zeta_1\circ \nu_1(P)}^*(w)$, and properties of adjoints of compositions show that $\dot{\nu}_P(s)=\dot{\zeta}_{2,\zeta_1\circ \nu_1(P)}\circ \dot{\zeta}_{1,\nu_1(P)}\circ \dot{\nu}_{1,P}(s)$. Plugging in the provided forms of the differential operators and their adjoints gives \eqref{eq:condDens}.

\subsubsection{Dose-response function}\label{app:doseResponse}

Let $Z=(X,A,Y)\in\mathcal{Z}=\mathcal{X}\times [0,1]\times \mathbb{R}$. Here, $X$ is a covariate, $A$ is a continuous treatment, and $Y$ is an outcome. Let $\lambda$ denote the Lebesgue measure on $[0,1]$ and $\pi_P(\,\cdot\mymid x)$ denote the conditional density of $A$ given $X=x$ under sampling from $P$. Suppose there exists $\delta\in (0,\infty)$ such that, for all $P\in\mathcal{M}$, $\pi_P$ is $P$-a.s. bounded below by $\delta$ and $E_P[Y^2|A,X]$ is $P$-a.s. bounded above by $\delta^{-1}$. The parameter of interest $\nu : \mathcal{M}\rightarrow L^2(\lambda)$ is defined so that $\nu(P)(a)=\int \mu_P(a,x)P_X(dx)$, where $\mu_P(a,x):=E_P[Y\mymid A=a,X=x]$. Example~2 in \cite{luedtke2023one} shows that $\nu$ is pathwise differentiable with
\begin{align*}
\dot{\nu}_P(s)(a) &= \int E_P\left[\{Y-\mu_P(a,x)\}s_{Y \mymid A,X}(Y\mid a,x)\,\middle|\,A=a,X=x\right] P_X(dx)  \\
&\quad+ \int [\mu_P(a,x)-\nu(P)(a)]s_X(x)P_X(dx),\textnormal{ and } \\
\dot{\nu}_P^*(w)(y,a,x) &= \frac{y-\mu_P(a,x)}{\pi_P(a\mid x)}w(a) + \int [\mu_P(a',x)-\nu(P)(a')]w(a')\lambda(da').
\end{align*}
Above $s_{Y \mymid A,X}(y\mymid a,x):=s(z)-E_P[s(Z)\mymid A=a,X=x]$ and $s_X(x):=E_P[s(Z)\mymid X=x]$.

\subsubsection{Counterfactual density}\label{app:countDens}

Let a generic observation $Z=(X,A,Y)$ drawn from a distribution $P\in\mathcal{M}$ consist of covariates $X$, a binary treatment $A$, and an outcome $Y$ with support on $\mathcal{Y}$. Suppose all distributions in $\mathcal{M}$ are equivalent and there is a $\sigma$-finite measure $\lambda_Y$ such that, for all $P \in \mathcal{M}$, there is a regular conditional probability $P_{Y \mymid A,X}$ such that $P_{Y \mymid A,X}(\cdot \mymid a,x) \ll \lambda_Y$ for $P$-almost all $(a,x) \in \{0,1\} \times \mathcal{X}$. Define the propensity to receive treatment $a$ as $\pi_P(a\mymid x) := P(A=a \mymid X=x)$ and let $p_{Y|A,X}(\cdot\mymid 1,x)$ denote the conditional density of $Y$ given $(A,X)=(1,x)$. Suppose that $\mathcal{M}$ is such that
\begin{equation*}
\inf_{P\in\mathcal{P}}\essinf_{x \in \mathcal{X}} \pi_P(1\mymid x) > 0\ \ \textnormal{ and } \  \sup_{P\in\mathcal{P}}\esssup_{(x,y)\in\mathcal{X}\times\mathcal{Y}} p_{Y|A,X}(y\mymid 1,x) < \infty,
\end{equation*}
where the essential infimum is under the marginal distribution $P_X$ of $X$ under sampling from $P$ and the essential supremum is under $P_X\times \lambda_Y$.

Define $\nu: \mathcal{M}\to L^2(\lambda_Y)$ so that
\begin{align*}
\nu(P)(y) = \int p_{Y|A,X}(y\mymid 1,x)\, P_X(dx).
\end{align*}
Under causal conditions, this parameter corresponds to the counterfactual density that would be observed if everyone received treatment $A=1$ \citep{kennedy2021semiparametric}. Letting $s_{Y\mymid A,X}$ and $s_X$ be as defined in Appendix~\ref{app:doseResponse}, the results in \cite{luedtke2023one} show $\nu$ is pathwise differentiable with local parameter
\begin{align*}
\dot{\nu}_P(s)(y) = \int \big\{&s_{Y\mymid A,X}(y\mymid 1,x) + s_X(x)\big\}p_{Y|A,X}(y\mymid 1,x) P_X(dx)
\end{align*}
and efficient influence operator
\begin{align}
\dot{\nu}_P^* (w)(y,a,x) &= \frac{1\{a=1\}}{\pi_P(a\mymid x)}\left\{w(y)-E_P\left[w(Y) \mymid A=a, X=x \right]\right\} \nonumber \\
&\quad+ \Big(E_P\left[w(Y)\mymid A=1, X=x\right] - \int E_P\left[w(Y) \mymid A=1, X=x'\right]P_X(dx')\Big). \label{eq:countDensEIO}
\end{align}

\subsection{Maps that only depend on their Hilbert-valued argument}\label{app:onlyHilbert}

\subsubsection{General case: Hadamard differentiable map}\label{app:hadDiff}

Suppose that there exists a Hadamard differentiable map $\zeta : \mathcal{U}\rightarrow\mathcal{W}$ that is such that $\theta(P,u)=\zeta(u)$ for all $(P,u)\in\mathcal{M}\times\mathcal{U}$. We claim that $\theta$ is totally pathwise differentiable at any $(P,u)\in\mathcal{M}\times\mathcal{U}$ with $\dot{\theta}_{P,u}(s,t)=\dot{\zeta}_u(t)$ and $\dot{\theta}_{P,u}^*(w)=(0,\dot{\zeta}_u^*(w))$.

Fix $P\in\mathcal{M}$, $u\in\mathcal{U}$, $s\in\check{\mathcal{M}}_P$, $t\in\check{\mathcal{U}}_u$, $\{P_\epsilon : \epsilon\in [0,1]\}\in\mathscr{P}(P,\mathcal{M},s)$, and $\{u_\epsilon : \epsilon\in [0,1]\}\in \mathscr{P}(u,\mathcal{U},t)$. For any $\epsilon\in (0,1]$, the definition of $\theta$ and the Hadamard differentiability of $\zeta$ imply that
\begin{align*}
    \left\|\theta(P_\epsilon,u_\epsilon)-\theta(P,u)-\epsilon \dot{\zeta}_u(t)\right\|_{\mathcal{W}}&= \left\|\zeta(u_\epsilon)-\zeta(u)-\epsilon \dot{\zeta}_u(t)\right\|_{\mathcal{W}} = o(\epsilon).
\end{align*}
This establishes our claim of total pathwise differentiability with $\dot{\theta}_{P,u}(s,t)=\dot{\zeta}_u(t)$. The map $\dot{\zeta}_u$ is bounded and linear, and so $\dot{\theta}_{P,u}$ is bounded and linear as well. As for the adjoint, observe that, for any $s\in \dot{\mathcal{M}}_P$, $t\in\dot{\mathcal{U}}_u$, and $w\in\mathcal{W}$,
\begin{align*}
    \left\langle \dot{\theta}_{P,u}(s,t), w\right\rangle_{\mathcal{W}}&= \left\langle \dot{\zeta}_u(t), w\right\rangle_{\mathcal{W}} = \left\langle t, \dot{\zeta}_u^*(w)\right\rangle_{\dot{\mathcal{U}}_u} = \left\langle (s,t), (0,\dot{\zeta}_u^*(w))\right\rangle_{\dot{\mathcal{M}}_P\oplus \dot{\mathcal{U}}_u}.
\end{align*}
Hence,  $\dot{\theta}_{P,u}^*(w)=(0,\dot{\zeta}_u^*(w))$.

In the remainder of Appendix~\ref{app:onlyHilbert}, we establish the Hadamard differentiability of maps $\zeta : \mathcal{U}\rightarrow\mathcal{W}$. Applying the general results of this Appendix~\ref{app:hadDiff} then shows that the parameter $\theta$ with $\theta(P,u)=\zeta(u)$ is totally pathwise differentiable at any $(P,u)\in\mathcal{M}\times\mathcal{U}$ with $\dot{\theta}_{P,u}(s,t)=\dot{\zeta}_u(t)$ and $\dot{\theta}_{P,u}^*(w)=(0,\dot{\zeta}_u^*(w))$, which will verify the claimed form of the adjoint of the differential operator for the Hadamard differentiable primitives in Table~\ref{tab:primitives}.

\subsubsection{Inner product}\label{app:innerProd}

Let $\mathcal{R}$ be a Hilbert space and $\mathcal{U}=\mathcal{R}\oplus\mathcal{R}$. Define $\zeta : \mathcal{U}\rightarrow\mathbb{R}$ so that $\zeta(u)=\langle u_1,u_2\rangle_{\mathcal{R}}$, where $u=(u_1,u_2)$. For any $u\in \mathcal{U}$, $t=(t_1,t_2)\in \check{\mathcal{U}}_u=\mathcal{R}\oplus\mathcal{R}$, and $\{u_\epsilon=(u_{1,\epsilon},u_{2,\epsilon}) : \epsilon\}\in \mathscr{P}(u,\mathcal{U},t)$,
\begin{align*}
    \zeta(u_\epsilon)-\zeta(u) - \epsilon \left[\langle t_1,u_2\rangle_{\mathcal{R}}+\langle u_1,t_2\rangle_{\mathcal{R}}\right]&= \left\langle u_{1,\epsilon}-u_1-\epsilon t_1,u_2\right\rangle_{\mathcal{R}} + \left\langle u_1,u_{2,\epsilon}-u_2-\epsilon t_2\right\rangle_{\mathcal{R}} \\
    &\quad+ \left\langle u_{1,\epsilon}-u_1,u_{2,\epsilon}-u_2-\epsilon t_2\right\rangle_{\mathcal{R}} + \epsilon\left\langle u_{1,\epsilon}-u_1, t_2\right\rangle_{\mathcal{R}}.
\end{align*}
By Cauchy-Schwarz and the fact that $\{u_\epsilon : \epsilon\}\in \mathscr{P}(u,\mathcal{U},t)$, each of the four terms on the right-hand side is $o(\epsilon)$. Moreover, $\dot{\zeta}_u(t)=\langle t_1,u_2\rangle_{\mathcal{R}}+\langle u_1,t_2\rangle_{\mathcal{R}}$ is a bounded linear operator, establishing the Hadamard differentiability of $\zeta$. The Hermitian adjoint of $\dot{\zeta}_u$ is $\dot{\zeta}_u^*(w)=(wu_2,wu_1)$.

Importantly, the Hilbert space $\mathcal{T}=\mathcal{R}\oplus\mathcal{R}$ is not ambient for this primitive: the evaluation of this primitive inherently depends on the inner product in $\mathcal{R}$. Hence, when used in Algorithm~\ref{alg:parameter}, this primitive does not allow $\mathcal{R}$ to depend on the input to the algorithm, $P$. For example, consider the case where $\mathcal{R}=L^2(P)$. In this case, a $P$-dependent inner product primitive would take the form $\theta(P,(u_1,u_2))=\langle u_1,u_2\rangle_{L^2(P)}$, which depends nontrivially on its distribution-valued argument and therefore should be studied as in Appendix~\ref{app:primitiveTPD}. For the particular case where $\mathcal{R}=L^2(P)$, the $P$-dependent inner product primitive is a special case of the multifold mean studied in Appendix~\ref{app:multilinearForm}, and so is totally pathwise differentiable under conditions given there.

\subsubsection{Squared norm}\label{app:squaredNorm}

Let $\zeta(u)=\|u\|_{\mathcal{T}}^2$, where $\mathcal{U}$ equals the Hilbert space $\mathcal{T}$. For any $u\in \mathcal{U}$, $t\in \check{\mathcal{U}}_u=\mathcal{U}$, and $\{u_\epsilon : \epsilon\}\in \mathscr{P}(u,\mathcal{U},t)$, $\zeta(u_\epsilon)-\zeta(u) - 2\epsilon \langle t,u\rangle_{\mathcal{T}}= o(\epsilon)$. Moreover, $\dot{\zeta}_u(t)=2\langle t,u\rangle_{\mathcal{T}}$ is a bounded linear operator, establishing the Hadamard differentiability of $\zeta$. The Hermitian adjoint of $\dot{\zeta}_u$ is $\dot{\zeta}_u^*(w)=2wu$.

Like the inner product primitive in Appendix~\ref{app:innerProd}, the Hilbert space $\mathcal{T}$ is not ambient for this primitive: evaluating $\|\cdot\|_{\mathcal{T}}^2$ inherently depends on $\mathcal{T}$. Hence, when used in Algorithm~\ref{alg:parameter}, this inner product primitive does not allow $\mathcal{T}$ to depend on the input to the algorithm $P$.

\subsubsection{Differentiable functions}\label{app:diffFun}
Let $\mathcal{U}$ be an open subset of $\mathcal{T}=\mathbb{R}^d$. Any differentiable map $\zeta : \mathcal{U}\rightarrow\mathbb{R}$ is Fr\'{e}chet and, therefore, Hadamard differentiable with $\dot{\zeta}_u : t\mapsto t^\top \nabla \zeta(u)$. The adjoint of this map is $\dot{\zeta}_u^* : w\mapsto w\, \nabla \zeta(u)$, that is, the scalar multiplication of $w$ with $\nabla \zeta(u)$.

\subsubsection{Pointwise operations}\label{app:pointwise}

Let $C^1(\mathbb{R}^d)$ denote the set of continuously differentiable $g : \mathbb{R}^d\rightarrow \mathbb{R}$ and $\nabla_j\, g(a)$ the $j$-th entry of the gradient $\nabla g(a)$. Let $\lambda$ be a $\sigma$-finite measure on $\mathcal{X}$ and note that the tangent space of $L^2(\lambda)^{\oplus d}:=\oplus_{j=1}^d L^2(\lambda)$ at any $u$ is $L^2(\lambda)^{\oplus d}$ itself. For a generic $u\in L^2(\lambda)^{\oplus d}$, write $u_j\in L^2(\lambda)$ for the $j$-th entry of $u$ and $\bar{u}$ for $x\mapsto (u_j(x))_{j=1}^d$. Let $\|\cdot\|_2$ denote the Euclidean norm.
\begin{theorem}[Hadamard differentiability of pointwise operations]\label{thm:pointwise}
Suppose $f_x\in C^1(\mathbb{R}^d)$ for each $x\in\mathcal{X}$ and $\sup_{a\in\mathbb{R}^d,x\in\mathcal{X}}\|\nabla f_x(a)\|_2< \infty$. Define $\zeta : L^2(\lambda)^{\oplus d}\rightarrow L^2(\lambda)$ so that $\zeta(u) : x\mapsto f_x\circ \bar{u}(x)$. The map $\zeta$ is Hadamard differentiable at each $u\in L^2(\lambda)^{\oplus d}$ with $\dot{\zeta}_u(t) : x\mapsto \bar{t}(x)^\top \nabla f_x\circ \bar{u}(x)$ and $\dot{\zeta}_u^*(w)=\left(x\mapsto w(x)\nabla_j f_x(\bar{u}(x))\right)_{j=1}^d$.
\end{theorem}
\begin{proof}[Proof of Theorem~\ref{thm:pointwise}]
Fix $u\in L^2(\lambda)^{\oplus d}$ and define the map $\dot{\underline{\zeta}}_u : L^2(\lambda)^{\oplus d}\rightarrow L^2(\lambda)$ so that $\dot{\underline{\zeta}}_u(t) : x\mapsto \bar{t}(x)^\top \nabla f_x\circ \bar{u}(x)$ for all $t\in L^2(\lambda)^{\oplus d}$. This map is linear. It is also bounded since, by Cauchy-Schwarz, the assumption that $m:=\sup_{a\in\mathbb{R}^d,x\in\mathcal{X}}\|\nabla f_x(a)\|_2<\infty$, and the definition of the $L^2(\lambda)^{\oplus d}$-norm,
\begin{align*}
    \big\|\dot{\underline{\zeta}}_u(t)\big\|_{L^2(\lambda)}^2&\le \left\|x\mapsto \bar{t}(x)^\top \nabla f_x\circ \bar{u}(x)\right\|_{L^2(\lambda)}^2\le \int \left\|\bar{t}(x)\right\|_2^2 \left\|\nabla f_x\circ \bar{u}(x)\right\|_2^2 \lambda(dx) \\
    &\le m^2 \int \left\|\bar{t}(x)\right\|_2^2\lambda(dx) = m^2\|t\|_{L^2(\lambda)^{\oplus d}}^2\,.
\end{align*}
The Hermitian adjoint of $\dot{\underline{\zeta}}_u$ is $\dot{\underline{\zeta}}_u^*(w)=\left(x\mapsto w(x)\nabla_j f_x\circ \bar{u}(x)\right)_{j=1}^d$ since, for $t\in L^2(\lambda)^{\oplus d}$ and $w\in L^2(\lambda)$,
\begin{align*}
    \left\langle \dot{\underline{\zeta}}_u(t),w\right\rangle_{L^2(\lambda)}&= \int \bar{t}(x)^\top \nabla f_x\circ \bar{u}(x)w(x)\, \lambda(dx) = \sum_{j=1}^d \int t_j(x) \,\nabla_j f_x\circ \bar{u}(x)\,w(x)\, \lambda(dx) \\
    &= \left\langle t,\big(x\mapsto w(x)\nabla_j f_x\circ \bar{u}(x)\big)_{j=1}^d\right\rangle_{L^2(\lambda)^{\oplus d}}.
\end{align*}
Our proof will be complete if we can establish that \eqref{eq:Hadamard} holds with $\dot{\zeta}_u=\dot{\underline{\zeta}}_u$ --- we do this in what follows.

Let $u=(u_j)_{j=1}^d$ and $t=(t_j)_{j=1}^d$ belong to $L^2(\lambda)^{\oplus d}$ and $\{u_\epsilon=(u_{\epsilon,j})_{j=1}^d : \epsilon\}\in \mathscr{P}(u,L^2(\lambda)^{\oplus d},t)$. Define $t_\epsilon:=\{u_\epsilon - u\}/\epsilon$ and $u^{(\epsilon)}:= u + \epsilon t$. By the inequality $(a+b)^2\le 2(a^2+b^2)$,
\begin{align}
    &\frac{1}{2}\left\|\zeta(u_\epsilon) - \zeta(u) - \epsilon\,\dot{\underline{\zeta}}_u(t)\right\|_{L^2(\lambda)}^2= \frac{1}{2}\int \left[f_x\circ \bar{u}_\epsilon(x) - f_x\circ \bar{u}(x) - \epsilon\, \bar{t}(x)^\top\nabla f_x\circ\bar{u}(x)\right]^2 \lambda(dx) \label{eq:pointwiseDecomp} \\
    &\le \int \left[f_x\circ \bar{u}_\epsilon(x) - f_x\circ \bar{u}^{(\epsilon)}(x)\right]^2 \lambda(dx) + \int \left[f_x\circ \bar{u}^{(\epsilon)}(x) - f_x\circ \bar{u}(x) - \epsilon\, \bar{t}(x)^\top\nabla f_x\circ\bar{u}(x)\right]^2 \lambda(dx). \nonumber
\end{align}
We shall show that each of the two terms on the right-hand side is $o(\epsilon^2)$. Beginning with the first, we use that, for each $x\in\mathcal{X}$, $f_x$ is $m$-Lipschitz since it belongs to $C^1(\mathbb{R}^d)$ and has its gradient bounded by $m$. Hence,
\begin{align*}
    &\int \left[f_x\circ \bar{u}_\epsilon(x) - f_x\circ \bar{u}^{(\epsilon)}(x)\right]^2 \lambda(dx)\le m^2\int \left\|\bar{u}_\epsilon(x) - \bar{u}^{(\epsilon)}(x)\right\|_2^2 \lambda(dx)\\
    &=m^2\int \left\|\bar{u}_\epsilon(x)-\bar{u}-\epsilon\bar{t}(x)\right\|_2^2 \lambda(dx) = \sum_{j=1}^d \int \left[u_{\epsilon,j}(x)-u_{\epsilon,j}(x)-\epsilon\,t_j(x)\right]^2 \lambda(dx) \\
    &= \left\|u_\epsilon-u - \epsilon t\right\|_{L^2(\lambda)^{\oplus d}}^2 = o(\epsilon^2),
\end{align*}
where the final equality holds since $\{u_\epsilon : \epsilon\}\in \mathscr{P}(u,L^2(\lambda)^{\oplus d},t)$. Hence, the first term on the right-hand side of \eqref{eq:pointwiseDecomp} is $o(\epsilon^2)$. For the second, the fundamental theorem of calculus, the chain rule, Jensen's inequality, and Cauchy-Schwarz together show that
\begin{align*}
    &\int \left[f_x\circ \bar{u}^{(\epsilon)}(x) - f_x\circ \bar{u}(x) - \epsilon\, \bar{t}(x)^\top\nabla f_x\circ\bar{u}(x)\right]^2 \lambda(dx) \\
    &\quad= \int \left[\int_0^1 \left.\frac{d}{da'}f_x(\bar{u}(x) + a' \epsilon \bar{t}(x))\right|_{a'=a} da - \epsilon\, \bar{t}(x)^\top\nabla f_x\circ\bar{u}(x)\right]^2 \lambda(dx) \\
    &\quad= \epsilon^2 \int \left[\int_0^1 \bar{t}(x)^\top \left\{\nabla f_x\left(\bar{u}(x)+a\epsilon \bar{t}(x)\right)  - \nabla f_x\circ\bar{u}(x)\right\}da\right]^2 \lambda(dx) \\
    &\quad\le \epsilon^2 \int \int_0^1 \left[\bar{t}(x)^\top \left\{\nabla f_x\left(\bar{u}(x)+a\epsilon \bar{t}(x)\right)  - \nabla f_x\circ\bar{u}(x)\right\}\right]^2 da\, \lambda(dx)\le \epsilon^2 \int \left\|\bar{t}(x)\right\|_2^2 I_\epsilon(x) \, \lambda(dx),
\end{align*}
where $I_\epsilon(x):=\int_0^1 \left\|\nabla f_x\left(\bar{u}(x)+a\epsilon \bar{t}(x)\right)  - \nabla f_x\circ\bar{u}(x)\right\|_2^2 da$.  We use the dominated convergence theorem to show that the integral on the right-hand side is $o(1)$, which will imply that the second term in \eqref{eq:pointwiseDecomp} is $o(\epsilon^2)$, completing the proof. To see that $\|\bar{t}(\cdot)\|_2^2 I_\epsilon(\cdot)$ has an integrable dominating function, observe that this function is nonnegative and pointwise upper bounded by $g(x):=2m^2\|\bar{t}(x)\|_2^2$; furthermore, $g$ is $\lambda$-integrable since $t\in L^2(\lambda)^{\oplus d}$. To see that $\lim_{\epsilon\rightarrow 0}\|\bar{t}(x)\|_2^2 I_\epsilon(x)=0$ for $\lambda$-almost all $x$, observe first that $\|\bar{t}(x)\|_2<\infty$ for $\lambda$-almost all $x$ by virtue of the fact that $t\in L^2(\lambda)^{\oplus d}$. Combining this with the continuity of $\nabla f_x$ for each $x$ and the bound $I_\epsilon(x)\le \sup_{b\in\mathbb{R}^d : \|b-\bar{u}(x)\|_2\le \epsilon\|\bar{t}(x)\|_2}\|\nabla f_x(b)  - \nabla f_x\circ\bar{u}(x)\|_2^2$, we see that $\lim_{\epsilon\rightarrow 0}I_\epsilon(x)= 0$ for $\lambda$-almost all $x$. Hence, by the dominated convergence theorem, the integral on the right-hand side above goes to zero as $\epsilon\rightarrow 0$.
\end{proof}

\subsubsection{Bounded affine map}\label{app:affine}
Let $\mathcal{U}=\mathcal{T}$ and $\mathcal{W}$ be Hilbert spaces and fix a bounded linear map $\kappa : \mathcal{U}\rightarrow\mathcal{W}$. 
For $c\in\mathcal{W}$, define $\zeta : \mathcal{U}\rightarrow\mathcal{W}$ so that $\zeta(u)=\kappa(u)+c$. We will show that $\zeta$ is Hadamard differentiable at any $u\in\mathcal{U}$. The differential operator and its adjoint are given by $\dot{\zeta}_u=\kappa$ and $\dot{\zeta}_u^*=\kappa^*$, where $\kappa^*$ is the Hermitian adjoint of $\kappa$.

To see that Hadamard differentiability holds, fix $t\in\mathcal{T}$ and $\{u_\epsilon : \epsilon\in [0,1]\}\in \mathscr{P}(u,\mathcal{U},t)$. By the linearity and boundedness of $\kappa$,
\begin{align*}
    \left\|\zeta(u_\epsilon)-\zeta(u)-\epsilon \kappa(t)\right\|_{\mathcal{W}}&= \left\|\kappa(u_\epsilon - u - \epsilon t)\right\|_{\mathcal{W}}\le \|\kappa\|_{\mathrm{op}} \left\|u_\epsilon - u - \epsilon t\right\|_{\mathcal{T}} = o(\epsilon),
\end{align*}
where $\|\cdot\|_{\mathrm{op}}$ denotes the operator norm. Since $\kappa$ is bounded and linear, this shows that $\zeta$ is Hadamard differentiable with differential operator $\dot{\zeta}_u=\kappa$. The adjoint of $\dot{\zeta}_u$ equals the adjoint of $\kappa$, namely $\kappa^*$.

\subsubsection{Constant map}\label{app:constant}

Fix $c$ in a Hilbert space $\mathcal{W}$. Let $\zeta$ be the constant map that takes as input an element $u$ belonging to a subset $\mathcal{U}$ of a Hilbert space $\mathcal{T}$ and, regardless of form $u$ takes, returns $c$. This map is a special case of the affine maps considered in Appendix~\ref{app:affine} with $\kappa$ equal to a zero operator. Hence, $\zeta$ is Hadamard differentiable with $\dot{\zeta}_u(t)=0$ and $\dot{\zeta}_u^*(w)=0$.

\subsubsection{Coordinate projection}\label{app:coorProj}

Let $\mathcal{U}=\mathcal{R}_1\oplus\mathcal{R}_2$, where $\mathcal{R}_1$ and $\mathcal{R}_2$ are Hilbert spaces. Define $\zeta$ so that $\zeta(u)=u_1$, where $u=(u_1,u_2)$. Note that $\zeta$ is a linear map. It is also bounded since, for all $u=(u_1,u_2)\in \mathcal{U}$, $\|\zeta(u)\|_{\mathcal{R}_1}=\|u_1\|_{\mathcal{R}_1}\le \|u\|_{\mathcal{R}_1\oplus\mathcal{R}_2}$. Hence, by the results in Appendix~\ref{app:affine}, $\zeta$ is Hadamard differentiable with $\dot{\zeta}_{u}=\zeta$ and $\dot{\zeta}_{u}^*=\zeta^*$. To see that $\zeta^*(w)=(w,0)$, observe that, for any $w\in \mathcal{R}_1$ and $u\in \mathcal{U}$, $\langle \zeta(u),w\rangle_{\mathcal{R}_1} = \langle u_1,w\rangle_{\mathcal{R}_1} = \langle u,(w,0)\rangle_{\mathcal{R}_1\oplus \mathcal{R}_2}$.

\subsubsection{Change of measure}\label{app:changeMeasure}
Let $\lambda_1,\lambda_2$ be $\sigma$-finite measures on a measurable space $(\mathcal{X},\Sigma)$ with $\lambda_2\ll \lambda_1$ and $\lambda_1$-essentially bounded Radon Nikodym derivative $\frac{d\lambda_2}{d\lambda_1}$. In this subappendix only, we distinguish between functions $f : \mathcal{X}\rightarrow\mathbb{R}$ and the corresponding elements of $L^2(\lambda)$ spaces, which we denote by $[f]_\lambda$ to indicate that these elements are equivalence classes of functions that are equal to $f$ $\lambda$-almost everywhere; for $j\in\{1,2\}$, we write $[f]_j:=[f]_{\lambda_j}$ for shorthand. Consider the embedding map $\zeta : L^2(\lambda_1)\rightarrow L^2(\lambda_2)$, defined so that $\zeta([u]_{1})=[u]_{2}$; this map is well-defined since $\lambda_2\ll \lambda_1$ implies that, for any $u_1,u_2\in [u]_{1}\in L^2(\lambda_1)$, $[u_1]_{2}=[u_2]_{2}$ and, moreover, the $\lambda_1$-essential boundedness of $\frac{d\lambda_2}{d\lambda_1}$ implies that $[u]_{2}\in L^2(\lambda_2)$.

We establish the Hadamard differentiability of $\zeta$ by showing that it is a special case of the affine maps considered in Appendix~\ref{app:affine}. Clearly $\zeta$ is linear. To see that it is bounded, note that, for all $[u]_1\in L^2(\lambda_1)$,
\begin{align*}
\|\zeta([u]_{1})\|_{L^2(\lambda_2)}^2=\|[u]_{2}\|_{L^2(\lambda_2)}^2=\int u(x)^2 \lambda_2(dx)\le \left\| \frac{d\lambda_2}{d\lambda_1}\right\|_{L^\infty(\lambda_1)}\left\|[u]_{1}\right\|_{L^2(\lambda_1)}^2,
\end{align*}
where the essential supremum norm on the right-hand side is finite by assumption. Hence, by Appendix~\ref{app:affine}, $\zeta$ is Hadamard differentiable at any $u\in\mathcal{U}$ with $\dot{\zeta}_u=\zeta$ and $\dot{\zeta}_u^*=\zeta^*$.

We now derive a closed-form expression for $\zeta^*$. Observe that, for any $[u]_{1}\in L^2(\lambda_1)$ and $[w]_{2}\in L^2(\lambda_2)$,
\begin{align*}
    \langle \zeta([u]_{1}), [w]_{2}\rangle_{L^2(\lambda_2)}&= \langle [u]_{2}, [w]_{2}\rangle_{L^2(\lambda_2)} = \int u(x)\frac{d\lambda_2}{d\lambda_1}(x) w(x)\lambda_1(dx) = \left\langle [u]_{1},\left[\tfrac{d\lambda_2}{d\lambda_1}w\right]_{1}\right\rangle_{L^2(\lambda_1)}.
\end{align*}
Above, we used that the map $[w]_{2}\mapsto [\frac{d\lambda_2}{d\lambda_1}w]_{1}$ is well-defined, in the sense that, for any $w_1,w_2\in [w]_{2}$, $[\frac{d\lambda_2}{d\lambda_1}w_1]_{1}=[\frac{d\lambda_2}{d\lambda_1}w_2]_{1}$ since $\frac{d\lambda_2}{d\lambda_1}=0$ $\lambda_2$-almost everywhere; we also used that $\frac{d\lambda_2}{d\lambda_1}w$ is $\lambda_1$-square integrable since $w$
is $\lambda_1$-square integrable and $\frac{d\lambda_2}{d\lambda_1}$ is $\lambda_1$-essentially bounded.
Hence, $\zeta^*([w]_{2})=[\frac{d\lambda_2}{d\lambda_1} w]_{1}$.

\subsubsection{Lifting to new domain}\label{app:lifting}

Fix $Q\in\mathcal{M}$. For a coarsening $\mathscr{C} : \mathcal{Z}\rightarrow\mathcal{X}$, define the pushforward measure $Q_X:=Q\circ \mathscr{C}^{-1}$. Let $\zeta : L^2(Q_X)\rightarrow L^2(Q)$ be the lifting $\zeta(u):=u\circ \mathscr{C}$ --- put another way, $\zeta(u)(z)=u(x)$ for $Q$-almost all $z$, where $x=\mathscr{C}(z)$.

We establish the Hadamard differentiability of $\zeta$ by showing that it is a special case of the affine maps considered in Appendix~\ref{app:affine}. Clearly $\zeta$ is linear. It is also bounded since, for any $u\in L^2(Q_X)$, $\|\zeta(u)\|_{L^2(Q)}=\|u\|_{L^2(Q_X)}$. Hence, by Appendix~\ref{app:affine}, $\zeta$ is Hadamard differentiable at any $u\in\mathcal{U}$ with $\dot{\zeta}_u=\zeta$ and $\dot{\zeta}_u^*=\zeta^*$.

We now derive a closed-form expression for $\zeta^*$. Observe that, for any $u\in L^2(Q_X)$ and $w\in L^2(Q)$,
\begin{align*}
    \langle u\circ\mathscr{C},w\rangle_{L^2(Q)} = E_Q\left[u(X)w(Z)\right] = E_Q\left[u(X)E_Q[w(Z)\mymid X]\right] = \langle u,E_Q[w(Z)\mymid X=\cdot\;]\rangle_{L^2(Q_X)}.
\end{align*}
Hence, $\zeta^*(w)=E_Q[w(Z)\mymid X=\cdot\;]$.

\subsubsection{Fix binary argument}\label{app:fix}

Let $Q_{A,X}$ be a distribution of $(A,X)\in \{0,1\}\times\mathcal{X}$ and $Q_X$ be the corresponding marginal distribution of $X$. Suppose that $\pi(x):=Q_{A,X}(A=1\mymid X=x)$ is $Q_X$-a.s. bounded away from zero. Consider the map $\zeta : L^2(Q_{A,X})\rightarrow L^2(Q_X)$ defined so that $\zeta(u)(x)=u(1,x)$.

To show $\zeta$ is Hadamard differentiable at a generic $u\in L^2(Q_{A,X})$, we express it as a composition $\zeta_2\circ \zeta_1$ of Hadamard differentiable maps and then apply the chain rule. The first of these maps is the change of measure $\zeta_1 : L^2(\lambda_1)\rightarrow L^2(\lambda_2)$ studied in Appendix~\ref{app:changeMeasure}, where $\lambda_1=Q_{A,X}$ and $\lambda_2=Q_{A=1,X}$ with $Q_{A=1,X}$ the distribution on $\{0,1\}\times\mathcal{X}$ with Radon-Nikodym derivative $\frac{dQ_{A=1,X}}{dQ_{A,X}}(a,x)=\frac{a}{\pi(x)}$. This map is Hadamard differentiable with $\dot{\zeta}_{1,u}(t)(a,x)=t(a,x)$ and $\dot{\zeta}_{1,u}^*(w)(a,x)=\frac{a}{\pi(x)}w(a,x)$. The second is the map $\zeta_2 : L^2(Q_{A=1,X})\rightarrow L^2(Q_X)$ defined so that $\zeta_2(\tilde{u})(x)=\tilde{u}(1,x)$. This is a bounded linear map, so by Appendix~\ref{app:affine} it is Hadamard differentiable with $\dot{\zeta}_{2,u}=\zeta_2$ and $\dot{\zeta}_{2,u}^*=\zeta_2^*$. It is also straightforward to verify that $\zeta_2^*(w)(a,x)=w(x)$. Since $\dot{\zeta}_u^*$ is the adjoint of a composition, we also have that $\dot{\zeta}_u^*(w)= \dot{\zeta}_{1,u}^*\circ \dot{\zeta}_{2,\zeta_1(u)}^*(w)$. Plugging in the values for the differential operators and adjoints of $\zeta_1,\zeta_2$, this yields that $\dot{\zeta}_u(t)(x)=t(1,x)$ and $\dot{\zeta}_u^*(w)(a,x)=\frac{a}{\pi(x)}w(x)$.

\subsection{Primitives for marginal quantities and semiparametric models}\label{app:primitiveRemarks}

Some of the primitives in Table~\ref{tab:primitives} map to $L^2(Q_X)$, where $Q_X$ is a probability distribution on $\mathcal{X}$. This is the case, for example, for the conditional mean map $\theta(P,u)(\cdot)=E_P[u(Z)\mymid X=\cdot\,]$. If $\mathcal{X}=\{x_0\}$ is a singleton set, then it may be more natural to think of $\theta(P,u)$ as a real number than as a function. This can be accomplished by simply replacing a primitive $\theta : \mathcal{U}\rightarrow L^2(Q_X)$ by $\underline{\theta} : \mathcal{U}\rightarrow \mathbb{R}$, with $\underline{\theta}(P,u):=\theta(P,u)(x_0)$. This new primitive is totally pathwise differentiable with $\underline{\dot{\theta}}_{P,u}^*(\underline{w})=\dot{\theta}_{P,u}^*(x\mapsto \underline{w})$, $\underline{w}\in\mathbb{R}$. For instance, this approach allows us to establish the differentiability of a real-valued marginal mean map $\underline{\theta}(P,u)=E_P[u(Z)]$.

The primitives in Table~\ref*{tab:primitives} can be adapted to a semiparametric model. This can be done by using that, for any semiparametric $\mathcal{M}'\subset\mathcal{M}$,  the restriction of a totally pathwise differentiable primitive $\theta : \mathcal{M}\times\mathcal{U}\rightarrow\mathcal{W}$ to $\mathcal{M}'\times\mathcal{U}$ is itself totally pathwise differentiable. The adjoint of this new primitive at $(P,u)$ is the same as that of $\theta$, except its first entry is replaced by its orthogonal projection onto the tangent space $\dot{\mathcal{M}}_P'$. Thus, when $\psi : \mathcal{M}'\rightarrow\mathbb{R}$ can be expressed as a composition of restrictions of primitives in Table~\ref*{tab:primitives}, Algorithm~\ref*{alg:backprop} returns its semiparametric EIF relative to $\mathcal{M}'$. Naturally, the returned form of the EIF will be easiest to work with when the projection onto $\dot{\mathcal{M}}_P'$ has a known closed form. While this is the case in some settings \citep[e.g.,][Chapter 3.2]{bickel1993efficient}, it is certainly not the case for a general infinite-dimensional model. Identifying a means to circumvent this projection is an interesting area for future work.

\section{Step-by-step illustrations of Algorithm~\ref{alg:backprop}}\label{app:illustratingBackprop}

Like Table~\ref{tab:backpropIllustrationR2} from the main text, Tables \ref{tab:backpropIllustrationISD} and \ref{tab:backpropIllustrationECC} provide step-by-step illustrations of Algorithm~\ref{alg:backprop}. The first studies the expected density parameter, and the second studies an expected conditional covariance parameter. The primitives used to express these parameters are detailed in the table captions. The value of $f_0$ in the bottom row of each table is the EIF at a distribution $P$ in a nonparametric model $\mathcal{M}$. In these examples, $Q\in\mathcal{M}$ is a fixed distribution used to define the ambient Hilbert spaces of some of the primitives used to represent $\psi$. We suppose $Q$ is mutually absolutely continuous with $P$ and let $\frac{p}{q}:=\frac{dP}{dQ}$ and $\frac{p_X}{q_X}:=\frac{dP_X}{dQ_X}$.

\begin{table}[tb]
\centering
\caption{Step-by-step evaluation of Algorithm~\ref{alg:backprop}'s computations of the EIF of the expected density $\psi(\tilde{P})=\int \tilde{p}(z)^2 \lambda(dz)$, where $\tilde{p}=d\tilde{P}/d\lambda$ for a $\sigma$-finite measure $\lambda$. Under regularity conditions, the parameter $\psi$ is expressible as a composition of a root-density $\theta_1(\tilde{P},0)=\tilde{p}^{1/2}\in  L^2(\lambda)$, pointwise square $\theta_2(\tilde{P},h_1)=h_1^2\in L^2(\lambda)$, change of measure $\theta_3(\tilde{P},h_2)=h_2\in L^2(Q)$, and mean operator $\theta_4(\tilde{P},h_3)=E_{\tilde{P}}[h_3(Z)]\in \mathbb{R}$. Here, $\mathrm{pa}(1)=\emptyset$, $\mathrm{pa}(2)=\{1\}$, $\mathrm{pa}(3)=\{2\}$, and $\mathrm{pa}(4)=\{3\}$. The value of $f_0$ in the bottom row is the {\color{CBmagenta}efficient influence function} of $\psi$ at $P$.}
 \begin{tabular}{c || c | c | c | c | c } 
 \multirow{2}{*}{\shortstack{Step\\($j$)}} & \multicolumn{5}{|c}{Variables in Algorithm~\ref{alg:backprop}}\\\cline{2-6}
 & $f_4$ & $f_3$ & $f_2$ & $f_1$ & $f_0$ \\\hline\hline
4 & $1$ & $0$ & $0$ & $0$ & $0$ \\
3 &  & $p/q$ & $0$ & $0$ & $p-\psi(P)$ \\
2 &  & & $p$ & $0$ & $p-\psi(P)$ \\
1 &  & & & $2p^{3/2}$ & $p-\psi(P)$  \\
0 &  & & &  & {\color{CBmagenta}$\bm{2[p-\psi(P)]}$}  \\
\end{tabular}
 \label{tab:backpropIllustrationISD}
\end{table}

Since the EIF is a feature of the map $\psi : \mathcal{M}\rightarrow\mathbb{R}$, rather than the particular composition of primitives chosen to represent it, the EIF returned by Algorithm~\ref{alg:backprop} will not depend on the choice of $Q$ provided all of the primitives are totally pathwise differentiable. As noted in Section~\ref{sec:primitives}, our sufficient conditions for differentiability tend to be weakest when $Q=P$, and so we recommend this choice in practice.

Some of the primitives used in these examples do not appear in Table~\ref{tab:primitives}. For example, though the primitive $\theta(P,0)=E_P[Y\mymid X=\cdot\,]$ is closely related to the conditional mean primitive $\underline{\theta}(P,u)=E_P[u(Z)\mymid X=\cdot\,]$ from Table~\ref{tab:primitives}, $\theta$ returns the conditional mean of a predetermined transformation of $Z$, $\underline{u}(z)=y$, and so is not the same map as $\underline{\theta}$. In Figure~\ref{fig:compGraph}, we avoided this issue by expressing $E_P[Y\mymid X=\cdot\,]$ as a composition of primitives from Table~\ref{tab:primitives}: the constant map that returns $z\mapsto y$ and conditional mean primitive $\underline{\theta}(P,u)=E_P[u(Z)\mymid X=\cdot\,]$. A closely related approach, which is the one we take in this appendix, is to use the following result.
\begin{lemma}[Fixing the Hilbert-valued input of a primitive]
    Let $\underline{\theta} : \mathcal{M}\times\mathcal{U}\rightarrow\mathcal{W}$ be totally pathwise differentiable at $(P,\underline{u})\in\mathcal{M}\times\mathcal{U}$. Then, the primitive $\theta : \mathcal{M}\times\{0\}\rightarrow \mathcal{W}$, defined so that $\theta(\,\cdot\,,0)=\underline{\theta}(\,\cdot\,,\underline{u})$, is totally pathwise differentiable at $(P,0)$ with $\dot{\theta}_{P,0}^*(w)=(\dot{\underline{\theta}}_{P,\underline{u}}^*(w)_0,0)$, where $\dot{\underline{\theta}}_{P,\underline{u}}^*(w)_0$ denotes the first entry of $\dot{\underline{\theta}}_{P,\underline{u}}^*(w)$.
\end{lemma}
This lemma is a direct consequence of Lemma~\ref{lem:totalImpliesPartial} and Appendix~\ref{app:pathwiseDiff} and so the proof is omitted.

\begin{table}[tb]
\centering
\caption{Step-by-step evaluation of Algorithm~\ref{alg:backprop}'s computations of the EIF of the expected conditional covariance $\psi(\tilde{P})=E_{\tilde{P}}[\mathrm{Cov}_{\tilde{P}}(A,Y\mymid X)]$, where $Z=(X,A,Y)$. Let $\mathscr{C} : z\mapsto x$. Under regularity conditions, the parameter $\psi$ is expressible as a composition of conditional means $\theta_1(\tilde{P},0)(\cdot)=E_{\tilde{P}}[Y\mymid X=\cdot\,]\in L^2(Q_X)$ and $\theta_2(\tilde{P},0)(\cdot)=E_{\tilde{P}}[A\mymid X=\cdot\,]\in L^2(Q_X)$, liftings $\theta_3(\tilde{P},h_1)=(z\mapsto h_1(x))\in L^2(Q)$ and $\theta_4(\tilde{P},h_2)=(z\mapsto h_2(x))\in L^2(Q)$, a pointwise operation $\theta_5(\tilde{P},(h_3,h_4))=(z\mapsto ay-h_3(z)h_4(z))\in L^2(Q)$, and a mean operator $\theta_6(\tilde{P},h_5)=E_{\tilde{P}}[h_5(Z)]\in\mathbb{R}$. Here, $\mathrm{pa}(1)=\mathrm{pa}(2)=\emptyset$, $\mathrm{pa}(3)=\{1\}$, $\mathrm{pa}(4)=\{2\}$, $\mathrm{pa}(5)=\{3,4\}$, and $\mathrm{pa}(6)=\{5\}$. The value of $f_0$ in the bottom row is the {\color{CBmagenta}efficient influence function} of $\psi$ at $P$, where $\mu_P^Y(x):=E_P[Y\mymid X=x]$ and $\mu_P^A(x):=E_P[A\mymid X=x]$. Other expressions of this parameter, such as ones using the conditional covariance primitive studied in Appendix~\ref{app:condCovar}, can also be used to compute the EIF.}
 \resizebox{\textwidth}{!}{
 \begin{tabular}{c || c | c | c | c | c | c | c} 
 \multirow{2}{*}{\shortstack{Step\\($j$)}} & \multicolumn{7}{|c}{Variables in Algorithm~\ref{alg:backprop}}\\\cline{2-8}
 & $f_6$ & $f_5$ & $f_4$ & $f_3$ & $f_2$ & $f_1$ & $f_0$ \\\hline\hline
6 & $1$ & $0$ & $0$ & $0$ & $0$ & $0$ & $0$ \\
5 &  & $p/q$ & $0$ & $0$ & $0$ & $0$ & $z\mapsto ay-\mu_P^A(x)\mu_P^Y(x) - \psi(P)$ \\
4 &  & & $-\frac{p}{q}\mu_P^Y\circ \mathscr{C}$ & $-\frac{p}{q}\mu_P^A\circ \mathscr{C}$ & $0$ & $0$ & $z\mapsto ay-\mu_P^A(x)\mu_P^Y(x) - \psi(P)$ \\
3 &  & & & $-\frac{p}{q}\mu_P^A\circ \mathscr{C}$ & $-\frac{p_X}{q_X}\mu_P^Y$ & $0$ & $z\mapsto ay-\mu_P^A(x)\mu_P^Y(x) - \psi(P)$ \\
2 &  & & & & $-\frac{p_X}{q_X}\mu_P^Y$ & $-\frac{p_X}{q_X}\mu_P^A$ & $z\mapsto ay-\mu_P^A(x)\mu_P^Y(x) - \psi(P)$ \\
1 &  & & & & & $-\frac{p_X}{q_X}\mu_P^A$ & $z\mapsto a[y-\mu_P^Y(x)] - \psi(P)$ \\
0 &  & & & & & & {\color{CBmagenta}$\bm{z\mapsto [a-\mu_P^A(x)][y-\mu_P^Y(x)] - \psi(P)}$} \\
\end{tabular}
}
 \label{tab:backpropIllustrationECC}
\end{table}

\section{Study of remainder in von Mises expansion of nonparametric $R^2$}\label{app:vonMisesR2}

Table~\ref{tab:estimatedBackpropIllustrationR2} provides a step-by-step illustration of how Algorithm~\ref{alg:estimatedBackprop} estimates the efficient influence operator when it is applied to the nonparametric $R^2$ parameter as expressed in Figure~\ref{fig:compGraph}. The value of $\widehat{f}_0$ in the last row is the estimated efficient influence operator returned by the algorithm.

\begin{table}[tb]
\centering
\caption{Step-by-step evaluation of Algorithm~\ref{alg:estimatedBackprop} for the nonparametric $R^2$ parameter when it is expressed as in Figure~\ref{fig:compGraph}. This table is the estimation counterpart of Table~\ref{tab:backpropIllustrationR2} from the main text, which computes the EIF at a known distribution $P$.}
 \resizebox{\textwidth}{!}{
\begin{tabular}{c || c | c | c | c | c | c | c | c | c} 
 \multirow{2}{*}{\shortstack{Step\\($j$)}} & \multicolumn{7}{|c}{Variables in Algorithm~\ref{alg:estimatedBackprop}}\\[.2em]
 & \hspace{.25em}$f_8$\hspace{.25em} & $\widehat{f}_7$ & $\widehat{f}_6$ & $\widehat{f}_5$ & $\widehat{f}_4$& \hspace{-.25em}\phantom{$^\S$\hspace{.25em}}$\widehat{f}_3$\hspace{.25em}\phantom{$^\S$}\hspace{-.25em} & \hspace{.5em}$\widehat{f}_2$\hspace{.5em} & \hspace{-.25em}\phantom{$^\S$\hspace{.25em}}$\widehat{f}_1$\hspace{.25em}\phantom{$^\S$}\hspace{-.25em} & $\widehat{f}_0$ \\\hline\hline
8 & $1$\vphantom{$\frac{1}{\widehat{h}_2}$} & $0$ & $0$ & $0$ & $0$ & $-$ & $0$ & $-$ & $0$ \\
7 & & $-\frac{1}{\widehat{h}_2}$ & $0$ & $0$ & $0$ & $-$ & $\frac{\widehat{h}_7}{\widehat{h}_2^2}$ & $-$ & $0$ \\
6 & & & $-\frac{1}{\widehat{h}_2}$ & $0$ & $0$ & $-$ & $\frac{\widehat{h}_7}{\widehat{h}_2^2}$ & $-$ & $-\frac{\widehat{h}_6(\cdot)-\widehat{h}_7}{\widehat{h}_2}$ \\
5 & & & & $z\mapsto \frac{2[y-\widehat{h}_4(x)]}{\widehat{h}_2}$ & $0$ & $-$ & $\frac{\widehat{h}_7}{\widehat{h}_2^2}$ & $-$ & $-\frac{\widehat{h}_6(\cdot)-\widehat{h}_7}{\widehat{h}_2}$ \\
4 & & & & & $\frac{2\{\widehat{E}_5[Y\mymid X=\,\cdot\,]-\widehat{h}_4(\cdot)\}}{\widehat{h}_2}$ & $-$ & $\frac{\widehat{h}_7}{\widehat{h}_2^2}$ & $-$ & $-\frac{\widehat{h}_6(\cdot)-\widehat{h}_7}{\widehat{h}_2}$ \\
3 & & & & & & $-$ &  $\frac{\widehat{h}_7}{\widehat{h}_2^2}$ & $-$ & $z\mapsto -\frac{\widehat{h}_6(z)-\widehat{h}_7}{\widehat{h}_2}+\frac{2\{\widehat{E}_5[Y\mymid X=x]-\widehat{h}_4(x)\}\{y-\widehat{h}_4(x)\}}{\widehat{h}_2}$ \\
2 & & & & & & & $\frac{\widehat{h}_7}{\widehat{h}_2^2}$ & $-$ & $z\mapsto -\frac{\widehat{h}_6(z)-\widehat{h}_7}{\widehat{h}_2}+\frac{2\{\widehat{E}_5[Y\mymid X=x]-\widehat{h}_4(x)\}\{y-\widehat{h}_4(x)\}}{\widehat{h}_2}$ \\
1 & & & & & & & & $-$ & $z\mapsto -\frac{\widehat{h}_6(z)-\widehat{h}_7}{\widehat{h}_2}+\frac{2\{\widehat{E}_5[Y\mymid X=x]-\widehat{h}_4(x)\}\{y-\widehat{h}_4(x)\}}{\widehat{h}_2} + \frac{\widehat{h}_7\left\{(y-\widehat{E}_2[Y])^2-\widehat{h}_2\right\}}{\widehat{h}_2^2}$ \\
0 & & & & & & & & & $z\mapsto -\frac{\widehat{h}_6(z)-\widehat{h}_7}{\widehat{h}_2}+\frac{2\{\widehat{E}_5[Y\mymid X=x]-\widehat{h}_4(x)\}\{y-\widehat{h}_4(x)\}}{\widehat{h}_2} + \frac{\widehat{h}_7\left\{(y-\widehat{E}_2[Y])^2-\widehat{h}_2\right\}}{\widehat{h}_2^2}$ \\
\end{tabular}
}
 \label{tab:estimatedBackpropIllustrationR2}
\end{table}

Forward and backward pass nuisances are estimated as described in Appendix~\ref{app:nuisanceEstimation}. In the forward routines, this entails setting $\widehat{h}_1 : z\mapsto y$, $\widehat{h}_2=\mathrm{Var}_{P_{\mathscr{I}}}(Y)$ with $P_{\mathscr{I}}$ the empirical distribution of $Z_{\mathscr{I}}:=\{Z_i : i\in\mathscr{I}\}$, $\widehat{h}_3 : z\mapsto y$, $\widehat{h}_4 : x\mapsto \widehat{E}_4[Y\mymid X=x]$ with the conditional mean estimate returned by a nonparametric regression procedure, $\widehat{h}_5 : z\mapsto \widehat{h}_4(x)$, $\widehat{h}_6 : z\mapsto [y-\widehat{h}_5(z)]^2$, $\widehat{h}_7=E_{P_{\mathscr{I}}}[\widehat{h}_6(Z)]$, and $\widehat{h}_8=1-\widehat{h}_7/\widehat{h}_2$. In the backward routines, the only nuisances that require estimation are $E_P[Y\mymid X=\,\cdot\,]$ and $E_P[Y]$, which arise when obtaining the estimates $\widehat{\vartheta}_5$ and $\widehat{\vartheta}_2$, respectively. We denote these estimates by $\widehat{E}_5[Y\mymid X=\,\cdot\,]$ and $\widehat{E}_2[Y]$, respectively.

We now study the remainder $\mathcal{R}_n=\widehat{h}_8 - \psi(P) + \int \widehat{f}_0(z) P(dz)$ in this example. We let $h_1, h_2,\ldots,h_8$ be as defined in Algorithm~\ref{alg:parameter}, so that $\psi(P)=1-h_7/h_2$. Adding and subtracting terms then yields
\begin{align*}
    \widehat{h}_8 - \psi(P)&= - \frac{\widehat{h}_7-h_7}{\widehat{h}_2} + \frac{\widehat{h}_7[\widehat{h}_2-h_2]}{\widehat{h}_2^2} + \left(\frac{h_7}{h_2}-\frac{\widehat{h}_7}{\widehat{h}_2}\right)\frac{\widehat{h}_2-h_2}{\widehat{h}_2}\,.
\end{align*}
Combining this with the definition of $\mathcal{R}_n$, plugging in the value of $\widehat{f}_0$ returned by Algorithm~\ref{alg:estimatedBackprop}, and then simplifying shows that
\begin{align*}
    \mathcal{R}_n&= - \frac{\widehat{h}_7-h_7}{\widehat{h}_2} + \frac{\widehat{h}_7[\widehat{h}_2-h_2]}{\widehat{h}_2^2} + \left(\frac{h_7}{h_2}-\frac{\widehat{h}_7}{\widehat{h}_2}\right)\frac{\widehat{h}_2-h_2}{\widehat{h}_2} + \int \widehat{f}_0(z) P(dz) \\
    &= - \frac{\int (\widehat{h}_6-h_6)(z)\, P(dz)}{\widehat{h}_2} + \frac{\widehat{h}_7\int [\{y-\widehat{E}_2[Y]\}^2-\{y-E_P[Y]\}^2]P(dz)}{\widehat{h}_2^2} \\
    &\quad+ \left(\frac{h_7}{h_2}-\frac{\widehat{h}_7}{\widehat{h}_2}\right)\frac{\widehat{h}_2-h_2}{\widehat{h}_2} + \frac{2\int \{\widehat{E}_5[Y\mymid X=x]-\widehat{h}_4(x)\}\{h_4(x)-\widehat{h}_4(x)\}\,P_X(dx)}{\widehat{h}_2} \\
    &= - \frac{\int \{\widehat{h}_4(x)-h_4(x)\}^2\, P_X(dx)}{\widehat{h}_2} + \frac{\widehat{h}_7\{\widehat{E}_2[Y]-E_P[Y]\}^2}{\widehat{h}_2^2} \\
    &\quad+ \left(\frac{h_7}{h_2}-\frac{\widehat{h}_7}{\widehat{h}_2}\right)\frac{\widehat{h}_2-h_2}{\widehat{h}_2} + \frac{2\int \{\widehat{E}_5[Y\mymid X=x]-\widehat{h}_4(x)\}\{h_4(x)-\widehat{h}_4(x)\}\,P_X(dx)}{\widehat{h}_2}.
\end{align*}
The right-hand side consists of four terms. When analyzing them, we suppose that $h_2$ is strictly positive and the number of observations $n_{\mathscr{I}}$ in $Z_{\mathscr{I}}$ diverges with $n$ so that $\widehat{h}_2$ is bounded away from zero with probability tending to one. In this case, $\mathcal{R}_n$ will be $o_p(n^{-1/2})$ whenever the numerators of the four terms on the right are all $o_p(n^{-1/2})$. The first is the mean-squared prediction error of the estimate of $E_P[Y\mymid X=\,\cdot\,]$, which will be $o_p(n^{-1/2})$ provided appropriate smoothness or sparsity conditions hold and the estimator $\widehat{h}_4$ leverages them. The second numerator is $\widehat{h}_7$ times the squared error of $\widehat{E}_2$ as an estimate of $E_P[Y]$. If the empirical mean $E_{P_{\mathscr{I}}}[Y]$ is used to estimate this quantity and $\widehat{h}_7$ is $O_p(1)$, then this term will be $O_p(n_{\mathscr{I}}^{-1})$, and so $o_p(n^{-1/2})$ whenever $n_{\mathscr{I}}/n^{1/2}$ diverges with $n$. Up to a multiplicative factor of $1/\widehat{h}_2$, the third term is a product of the errors for estimating the real-valued quantities $h_7/h_2$ and $h_2$, and so should be $o_p(n^{-1/2})$ under appropriate conditions. The final term will be exactly zero if a common, deterministic regression scheme is used to construct the estimates $\widehat{E}_5[Y\mymid X=\,\cdot\,]$ and $\widehat{h}_4$ of $E_P[Y\mymid X=\,\cdot\,]$, since then these two quantities would be equal. Otherwise, Cauchy-Schwarz can be used to bound the magnitude of that quantity, which then shows that it is a product of $L^2(P_X)$ norms of $\widehat{E}_5[Y\mymid X=\,\cdot\,]-\widehat{h}_4(\cdot)$ and $h_4(\cdot)-\widehat{h}_4(\cdot)$. This product will be $o_p(n^{-1/2/})$ under appropriate smoothness or sparsity conditions.

\section{Nuisance estimation routines for primitives in Table~\ref{tab:primitives}}\label{app:nuisanceEstimation}

\subsection{Overview}
We now present forward and backward nuisance routines for some of the primitives appearing in Table~\ref{tab:primitives}. As detailed in Algorithm~\ref{alg:nuisanceRoutines}, the forward routine takes as input data $Z_{\mathscr{I}}:=\{Z_i : i\in\mathscr{I}\subseteq [n]\}$ and $u\in\mathcal{U}$ and returns an estimate $\widehat{h}$ of $\theta(P,u)$. The backward routine takes as input $Z_{\mathscr{I}}$, $u$, $\widehat{h}$, and $w\in\mathcal{W}$ and returns an estimate $\widehat{\vartheta}$ of $\dot{\theta}_{P,u}^*(w)$; $\widehat{\vartheta}$ must be compatible with the routine's inputs, in that there exists $\widehat{P}\in\mathcal{M}$ with $\widehat{h}=\theta(\widehat{P},u)$ and $\int \widehat{\vartheta}_0\,d\widehat{P}=0$. For compatibility to be achievable, the forward routine must ensure that $\widehat{h}$ belongs to the image of $\theta(\,\cdot\,,u)$, which we denote by $\theta(\mathcal{M},u)$.

When presenting these routines, we always take $u$ and $w$ to be generic elements of $\mathcal{U}$ and $\mathcal{W}$, respectively. Some of our proposed nuisance estimation routines use sample splitting or cross-fitting, and so take folds as an input; we denote these folds by $\mathscr{I}^{(1)}$ and $\mathscr{I}^{(2)}$, where $\{\mathscr{I}^{(1)},\mathscr{I}^{(2)}\}$ is a partition of $\mathscr{I}$ into subsets of sizes $n^{(1)}$ and $n^{(2)}$, $n^{(1)}\approx n^{(2)}$. For primitives whose domain or codomain depends on some $Q\in\mathcal{M}$, we follow the recommendation given in Section~\ref{sec:primitives} and focus on the case where $Q=P$, with $P$ the data-generating distribution.

The nuisance estimation routines presented hereafter are only examples --- other formulations are possible. This may include, for example, using different loss functions, penalty terms, observation weights, or forms of cross-fitting or sample splitting.

\subsection{Maps depend nontrivially on both of their arguments}

\subsubsection{Conditional mean}\label{sec:condExpNuisance}
As in Appendix~\ref{app:condExp}, $\theta(P,u)(x)=E_P[u(Z)\mymid X=x]$ and $\dot{\theta}_{P,u}^*(w)=(z\mapsto [u(z)-\theta(P,u)(x)]w(x),z\mapsto w(x))$. In the forward routine, $\theta(P,u)$ can be estimated by regressing $u(Z_i)$ against $X_i$, $i\in \mathscr{I}$, with the squared error loss. As long as the image of the regression estimate $\widehat{h}$ respects known bounds on $u$ and the locally nonparametric model $\mathcal{M}$ is suitably large, this estimate will fall in $\theta(\mathcal{M},u)$.  In the backward routine, we take $\widehat{\vartheta}(z):=(z\mapsto [u(z)-\widehat{h}(x)]w(x),z\mapsto w(x))$. Compatibility holds since $\widehat{h}\in \theta(\mathcal{M},u)$ implies there exists $\widehat{P}\in\mathcal{M}$ such that $\widehat{h}=\theta(\widehat{P},u)$ and, by the law of total expectation, $\int \widehat{\vartheta}_0\,d\widehat{P}=0$.

\subsubsection{Multifold conditional mean}\label{app:multilinearFormNuisance}

As in Appendix~\ref{app:multilinearForm}, $\theta(P,u)(\cdot)=E_{P^r}[u(Z_{[r]})\mymid X_{[r]}=\cdot\,]$ and $\dot{\theta}_{P,u}^*(w)=(\dot{\nu}_{P,u}^*(w),\dot{\zeta}_{P,u}^*(w))$, where $\dot{\zeta}_{P,u}^*(w)(z_{[r]})=w(x_{[r]})$ and
\begin{align*}
    \dot{\nu}_{P,u}^*(w)(z)&= \sum_{j=1}^r E_{P^r}\left[\left\{u(Z_{[r]}) - \theta(P,u)(X_{[r]})\right\}w(X_{[r]})\,\middle|\,Z_j=z\right].
\end{align*}

In the forward routine, $\theta(P,u)$ can be estimated by regressing $u(Z_{i[r]})$ against $(X_{i[r]})$, with $i[r]:=(i(1),i(2),\ldots,i(r))$ varying over the set $\mathscr{I}_r^{(1)}$ of $n^{(1)}!/(n^{(1)}-r)!$ choices of ordered tuples of $r$ unique elements from the first fold $\mathscr{I}^{(1)}$; to reduce runtime, a random subsample of the indices in $\mathscr{I}_r^{(1)}$ may be used while fitting this regression. Just as in Appendix~\ref{app:condExp}, this estimate, $\widehat{h}$, will belong to $\theta(\mathcal{M},u)$ as long as it respects any known bounds on $u$ and $\mathcal{M}$ is suitably large. 

In the backward routine, the estimate $\widehat{\vartheta}_0$ is obtained by running a pooled regression of $\left\{u(Z_{i[r]}) - \widehat{h}(X_{i[r]})\right\}w(X_{i[r]})$ against $Z_{i(j)}$, where $(i[r],j)\in \mathscr{I}_r^{(2)}\times [r]$ with $\mathscr{I}_r^{(2)}$ defined analogously to $\mathcal{I}_r^{(1)}$, but using the second fold rather than the first; runtime can be reduced by fitting the regression using only a random subset of the indices in $\mathscr{I}_r^{(2)}\times [r]$. This estimate is compatible with $\widehat{h}$ provided $\mathcal{M}$ is large enough. Because the second entry $\dot{\zeta}_{P,u}^*(w)$ of $\dot{\theta}_{P,u}^*(w)$ does not depend on $P$, the second entry of $\widehat{\vartheta}$ can be set to its known value, $z_{[r]}\mapsto w(x_{[r]})$.

\subsubsection{Conditional covariance}\label{app:condCovNuisance} As in Appendix~\ref{app:condCovar}, $\theta(P,u)(x)= \mathrm{cov}_P[u_1(Z),u_2(Z)\mymid X=x]$ and $\dot{\theta}_{P,u}^*(w)=(\dot{\nu}_{P,u}^*(w),\dot{\zeta}_{P,u}^*(w))$, where
\begin{align*}
    \dot{\nu}_{P,u}^*(w)(z)&= w(x)\left[\prod_{j=1}^2 \{u_j(z)-E_P[u_j(Z)\mymid X=x]\}-\theta(P,u)(x)\right], \\
    \dot{\zeta}_{P,u}^*(w)&= \left(z\mapsto w(x)\{u_{3-j}(z)-E_P[u_{3-j}(Z)\mymid X=x]\}\right)_{j=1}^2.
\end{align*}

In the forward routine, three regressions are fit. First, $u_1(Z_i)$ is regressed against $X_i$, $i\in \mathscr{I}^{(1)}$, yielding an estimate $\widehat{\mu}_1$ of $E_P[u_1(Z)\mymid X=\cdot\,]$. Second, $u_2(Z_i)$ is regressed against $X_i$, $i\in \mathscr{I}^{(1)}$, yielding an estimate $\widehat{\mu}_2$ of $E_P[u_2(Z)\mymid X=\cdot\,]$. Third, $\prod_{j=1}^2 \{u_j(Z_i)-\widehat{\mu}_j(X_i)\}$ is regressed against $X_i$, $i\in\mathscr{I}^{(2)}$, yielding the estimate $\widehat{h}$ of $\theta(P,u)$.

In the backward routine, we either reuse the saved estimates $\widehat{\mu}_1$ and $\widehat{\mu}_2$ fitted in the forward routine, or refit them if they were not saved. We then let
\begin{align*}
    \widehat{\vartheta}:= \left(z\mapsto w(x)\left[\prod_{j=1}^2 \{u_j(z)-\widehat{\mu}_j(x)\}-\widehat{h}(x)\right],\left(z\mapsto w(x)\{u_{3-j}(z)-\widehat{\mu}_{3-j}(x)\}\right)_{j=1}^2\right).
\end{align*}
Compatibility between $\widehat{\vartheta}_0$ and $\widehat{h}$ holds provided there exists a $\widehat{P}\in\mathcal{M}$ satisfying the following system of equations $\widehat{P}_X$-a.s.:
\begin{align*}
    E_{\widehat{P}}\left[ u_1(Z) \mymid X=x\right]&= \widehat{\mu}_1(x), \\
    E_{\widehat{P}}\left[ u_2(Z) \mymid X=x\right]&= \widehat{\mu}_2(x), \\
    E_{\widehat{P}}\left[ u_1(Z)u_2(Z) \mymid X=x\right]&= \widehat{h}(x) + \widehat{\mu}_1(x)\widehat{\mu}_2(x). \\
\end{align*}

\subsubsection{Conditional variance}\label{app:condVarNuisance} When $\theta(P,u)(x)= \mathrm{Var}_P[u(Z)\mymid X=x]$ as in Appendix~\ref{app:condVar}, the nuisance estimation routines for the conditional covariance can be used with $u_1=u_2=u$ (see Appendix~\ref{app:condCovNuisance}).

\subsubsection{Kernel embedding}\label{app:kernelEmbedNuisance} As in Appendix~\ref{app:kernelEmbed}, $\theta(P,u)(\cdot)= \int K(\,\cdot\,,x) \,u(x)\, P_X(dx)$ and $\dot{\theta}_{P,u}^*(w)=(wu-\int wu\,dP_X,w)$. In the forward routine, we can estimate the marginal distribution $P_X$ of $X$ under sampling from $P$. We denote this estimate by $\widehat{P}_X$. Then, we let $\widehat{h}(\cdot)= \int K(\,\cdot\,,x) \,u(x)\, \widehat{P}_X(dx)$ and, in the backward routine, we let $\dot{\theta}_{P,u}^*(w)=(wu-\int wu\,d\widehat{P}_X,w)$. These estimates are compatible provided there exists some $\widehat{P}\in\mathcal{M}$ with marginal distribution $\widehat{P}_X$. If $\mathcal{M}$ is large enough so that there exists a distribution whose marginal distribution of $X$ is equal to the empirical distribution of $\{X_i : i\in\mathscr{I}\}$, then $\widehat{P}_X$ can be chosen to be this empirical distribution.

\subsection{Maps that only depend on their distribution-valued argument}

\subsubsection{Root-density} As in Appendix~\ref{app:rootDensity}, $\theta(P,u)(z)=\frac{dP}{d\lambda}(z)^{1/2}$ and $\dot{\theta}_{P,u}^*(w)=(\dot{\nu}_P^*(w),0)$ with $\dot{\nu}_P^*(w)(z) = \frac{w(z)}{2\nu(P)(z)} - E_P \left[\frac{w(Z)}{2\nu(P)(Z)} \right]$. In the forward routine, the estimate $\widehat{h}$ of $\nu(P)$ can be obtained by taking the square root of a kernel density estimate, if $\lambda$ is a Lebesgue measure, or the empirical probability mass function, if $\lambda$ is a counting measure. In either case, the backward routine sets $\widehat{\vartheta}=([w/\widehat{h}-\int w\widehat{h} d\lambda]/2,0)$. The estimates $\widehat{\vartheta}$ and $\widehat{h}$ are plug-in estimators based on the distribution $\widehat{P}$ satisfying $d\widehat{P}/d\lambda=\widehat{h}^2$, and so the compatibility condition will hold if the nonparametric model $\mathcal{M}$ is large enough to contain this distribution.

\subsubsection{Conditional density} As in Appendix~\ref{app:condDensity}, $\theta(P,u)(z)=p_{Y\mymid X}(z)$ and $\dot{\nu}_P^*(w)(z)=w(z)p_{Y|X}(z)-E_P[w(Z)p_{Y|X}(Z)\mymid X=x]$. In the forward routine, the estimate $\widehat{h}$ of $p_{Y|X}$ can be obtained using kernel density estimation, if $\lambda$ is a Lebesgue measure, or nonparametric multinomial logistic regression, if $\lambda$ is a counting measure. In either case, the backward routine sets $\widehat{\vartheta}=(\widehat{\vartheta}_0,0)$ with $\widehat{\vartheta}_0(z)=w(z)\widehat{h}(z)-\int w(x,\tilde{y})\widehat{h}(x,\tilde{y}) \lambda(d\tilde{y})$. The estimates $\widehat{\vartheta}$ and $\widehat{h}$ are plug-in estimators based on any distribution $\widehat{P}$ whose conditional density of $Y| X$ relative to $\lambda$ is $\widehat{h}$, and so the compatibility condition will hold if the nonparametric model $\mathcal{M}$ is large enough to contain such a distribution.

\subsubsection{Dose-response function} As in Appendix~\ref{app:doseResponse}, $\theta(P,u)(a)=\int \mu_P(a,x)P_X(dx)$ and
\begin{align}
    \dot{\theta}_{P,u}^*(w) &= \left(z\mapsto \frac{y-\mu_P(a,x)}{\pi_P(a\mid x)}w(a) + \int [\mu_P(a',x)-\nu(P)(a')]w(a')\lambda(da'),0\right), \label{eq:doseResponseAdj}
\end{align}
where $z=(x,a,y)$, $\mu_P(a,x):=E_P[Y\mymid A=a,X=x]$ and $\pi_P(a\mymid x)=P(A=a\mymid X=x)$. In the forward routine, $\mu_P$ can be estimated using regression, yielding $\widehat{\mu}$, and then $\theta(P,u)$ can be estimated by $\widehat{h}(a)=\frac{1}{|\mathscr{I}|}\sum_{i\in\mathscr{I}} \widehat{\mu}(a,X_i)$. In the backward routine, $\pi_P$ can be estimated using kernel density estimation, yielding $\widehat{\pi}$, and then $\dot{\theta}_{P,u}^*(w)$ can be estimated using the plug-in estimator $\widehat{\vartheta}$ that takes the same form as \eqref{eq:doseResponseAdj}, but with $\mu_P$, $\pi_P$, and $\nu(P)$ replaced by $\widehat{\mu}$, $\widehat{\pi}$, and $\widehat{h}$, respectively. The compatibility condition holds provided there exists $\widehat{P}\in\mathcal{M}$ with $\mu_{\widehat{P}}=\widehat{\mu}$, $\pi_{\widehat{P}}=\widehat{\pi}$, and $\nu(\widehat{P})=\widehat{h}$; this will typically hold provided the nonparametric model $\mathcal{M}$ is large enough.

\subsubsection{Counterfactual density} As in Appendix~\ref{app:countDens}, $\theta(P,u)=\int p_{Y| A,X}(y\mymid 1,x)\, P_X(dx)$ and $\dot{\theta}_{P,u}^*(w)=(\dot{\nu}_P^*(w),0)$, where $\dot{\nu}_P^*$ is defined in \eqref{eq:countDensEIO}. To streamline the discussion, we suppose $\lambda_Y$ is a Lebesgue measure. In the forward routine, an estimate $(y,x)\mapsto \widehat{p}_{Y|1,X}(y\mymid x)$ of $p_{Y|A,X}(y\mymid 1,x)$ can be obtained using kernel density estimation among all $i\in\mathscr{I}$ with $A_i=1$. The counterfactual density $\theta(P,u)$ can then be estimated by $\widehat{h}(y)=\frac{1}{|\mathscr{I}|}\sum_{i\in\mathscr{I}} \widehat{p}_{Y|1,X}(y\mymid X_i)$. In the backward routine, $\pi_P$ can be estimated using regression, yielding $\widehat{\pi}$. Alternatively, $1/\pi_P$ can be estimated directly, yielding $1/\widehat{\pi}$ \citep{robins2007comment,chernozhukov2022automatic}. Reusing the nuisance $\widehat{p}_{Y|1,X}$ from the forward routine,  $E_P\left[w(Y) \mymid A=1, X=x \right]$ can be estimated using $\widehat{\mu}_w(x)=\int w(y)\widehat{p}_{Y|1,X}(y\mymid x)\, dy$. The adjoint can then be estimated with the plug-in estimate $\widehat{\vartheta}=(\widehat{\vartheta}_0,0)$, where
\begin{align*}
    \widehat{\vartheta}_0(z)&=\frac{1\{a=1\}}{\widehat{\pi}(a\mymid x)}\left\{w(y)-\widehat{\mu}_w(x)\right\} + \left[\widehat{\mu}_w(x) - \frac{1}{|\mathscr{I}|}\sum_{i : i\in\mathscr{I}}\widehat{\mu}_w(X_i)\right].
\end{align*}
The compatibility condition will typically hold provided the nonparametric model $\mathcal{M}$ is large enough.

\subsection{Maps that only depend on their Hilbert-valued argument}

\subsubsection{Overview} We now consider primitives that only depend on their Hilbert-valued input in the sense described in Appendix~\ref{app:hadDiff}, so that there is a Hadamard differentiable map $\zeta : \mathcal{U}\rightarrow\mathcal{W}$ such that $\theta(P',u')=\zeta(u')$ for all $(P',u')\in\mathcal{M}\times\mathcal{U}$.

These primitives can be evaluated explicitly in the forward pass, and so we let the forward nuisance routine return $\widehat{h}=\zeta(u)$. If the ambient Hilbert spaces $\mathcal{T}$ and $\mathcal{W}$ do not depend on $P$ and the form of $\dot{\zeta}_u^*$ is known, we let the backward nuisance routine return $\widehat{\vartheta}=(0,\dot{\zeta}_u^*(w))$. If one or both ambient Hilbert spaces depend on $P$ or the form of the adjoint is unknown, then estimating $\dot{\zeta}_u^*(w)$ may be necessary. Denoting this estimate by $\widehat{t}$, the backward nuisance routine then returns $\widehat{\vartheta}=(0,\widehat{t})$. The compatibility condition in Algorithm~\ref{alg:nuisanceRoutines} is necessarily satisfied since $\widehat{\vartheta}_0=0$.

In the remainder of this appendix, we present backward nuisance routines for cases where either the ambient Hilbert spaces depend on $P$ or $\dot{\zeta}_u^*(w)$ must be estimated.

\subsubsection{Pointwise operations} 
The adjoint $\dot{\zeta}_u^*(w)$ for the pointwise operation primitive from Appendix~\ref{app:pointwise} is invariant to the choice of $\lambda$ indexing the ambient Hilbert spaces $\mathcal{T}=L^2(\lambda)^{\oplus d}$ and $\mathcal{W}=L^2(\lambda)$. Therefore, if $\lambda=P$, then the backward nuisance routine can set $\widehat{\vartheta}$ equal to its true value, $\dot{\theta}_{P,u}^*(w)=(0,\dot{\zeta}_u^*(w))$.

\subsubsection{Bounded affine map}
We present a backward nuisance routine for the bounded affine maps $\zeta(\cdot)=\kappa(\cdot)+c$ from Appendix~\ref{app:affine} in cases where the form of the adjoint $\kappa^*$ is unknown. This routine is based on a Riesz loss \citep{chernozhukov2022automatic}. Recalling that $\dot{\zeta}_u=\kappa$ and $\dot{\zeta}_u^*=\kappa^*$, the key observation used is that $\dot{\zeta}_u^*(w)=\kappa^*(w)$ is the unique solution to
\begin{equation}\label{eq:rieszLoss}
\begin{aligned}
& \underset{t}{\text{minimize}}
& &\hspace{.5em}  \Phi_{u,w}(t):=\left\|t\right\|_{\mathcal{T}}^2 - 2\langle \kappa(t),w\rangle_{\mathcal{W}} \\
& \text{subject to}
& &\hspace{.5em} t\in\mathcal{T}.
\end{aligned}
\end{equation}
This holds since
\begin{align*}
    \kappa^*(w)&= \argmin_{t\in\mathcal{T}}\left\|t- \kappa^*(w)\right\|_{\mathcal{T}}^2 = \argmin_{t\in\mathcal{T}}\left[\left\|t\right\|_{\mathcal{T}}^2 - 2\langle t,\kappa^*(w)\rangle_{\mathcal{T}}\right] = \argmin_{t\in\mathcal{T}} \left[\left\|t\right\|_{\mathcal{T}}^2 - 2\langle \kappa(t),w\rangle_{\mathcal{W}}\right].
\end{align*}
When $\mathcal{T}$ and $\mathcal{W}$ do not depend on the unknown distribution $P$, the minimization problem in \eqref{eq:rieszLoss} can be approximated by approximating $\mathcal{T}$ with a $d$-dimensional subspace $\mathcal{T}_d$, yielding an approximation $\widehat{t}$ of $\dot{\zeta}_u^*(w)$.

We also consider cases where one or both of $\mathcal{T}$ and $\mathcal{W}$ depend on $P$. When doing this, we suppose that the $P$-dependent Hilbert spaces are equal to $L^2(P_X)$ or $L^2(P)$. Then, the squared norm or inner product in \eqref{eq:rieszLoss} can be estimated with an empirical mean over the observations in $Z_{\mathscr{I}}$. Using the case where $\mathcal{T}=L^2(P)$ and $\mathcal{W}=L^2(P_X)$ as an illustration, $\Phi_{u,w}(t)$ in \eqref{eq:rieszLoss} would be estimated by
\begin{align*}
    \widehat{\Phi}_{u,w}(t):= \frac{1}{n_{\mathscr{I}}}\sum_{i\in\mathscr{I}} t(Z_i)^2 - 2\frac{1}{n_{\mathscr{I}}} \sum_{i=1}^n \kappa(t)(Z_i)\, w(X_i).
\end{align*}
The risk function $\widehat{\Phi}_{u,w}$ can be used to obtain an estimate $\widehat{t}$ of $\dot{\zeta}_u^*(w)$ using a statistical learning tool such as random forests, gradient boosting, or neural networks \citep{breiman2001random,friedman2001greedy,rosenblatt1958perceptron}.

Since the primitives in Appendices~\ref{app:constant}-\ref{app:fix} are all special cases of bounded affine maps, the backward nuisance routine from this appendix can be used for all of those primitives.

\end{document}